%% file: main.tex
\documentclass[sigconf, nonacm]{acmart}
\usepackage{braket}
\usepackage{csquotes}
\input{macros}

\begin{document}
\title[Quantum Computing for Query Containment of Conjunctive Queries]{Quantum Computing for Query Containment\break of Conjunctive Queries}

\author{Luisa Gerlach}
\email{luisa.gerlach@hu-berlin.de}
\orcid{0009-0002-9114-5856}
\affiliation{\institution{Humboldt-Universität zu Berlin}
	\city{Berlin}
	\country{Germany}
}

\author{Tobias Köppl}
\email{tobias.koeppl@fokus.fraunhofer.de}
\orcid{0000-0003-3548-2807}
\affiliation{\institution{Fraunhofer-Institut für Offene Kommunikationssysteme (FOKUS)}
	\city{Berlin}
	\country{Germany}
}

\author{René Zander}
\email{rene.zander@fokus.fraunhofer.de}
\orcid{0000-0003-0603-5637}
\affiliation{\institution{Fraunhofer-Institut für Offene Kommunikationssysteme (FOKUS)}
	\city{Berlin}
	\country{Germany}
}

\author{Nicole Schweikardt}
\email{schweikn@hu-berlin.de}
\orcid{0000-0001-5705-1675}
\affiliation{\institution{Humboldt-Universität zu Berlin}
	\city{Berlin}
	\country{Germany}
}

\author{Stefanie Scherzinger}
\email{stefanie.scherzinger@uni-passau.de}
\orcid{0000-0002-1960-6171}
\affiliation{\institution{Universität Passau}
	\city{Passau}
	\country{Germany}
}

\begin{abstract}
We address the problem of checking query containment, a foundational problem in database research. Although extensively studied in theory research, optimization opportunities arising from query containment are not fully leveraged in commercial database systems, due to the high computational complexity and sometimes even undecidability of the underlying decision problem. In this article, we present the first approach to applying quantum computing to the query containment problem for conjunctive queries under set semantics. We propose a novel formulation as an optimization problem that can be solved on gate-based quantum hardware, and in some cases directly maps to quantum annealers.
We formally prove this formulation to be correct and present a prototype implementation which we evaluate using simulator software as well as quantum devices. Our experiments successfully demonstrate that our approach is sound and scales within the current limitations of quantum hardware.
In doing so, we show that quantum optimization can effectively address this problem. 
Thereby, we contribute a new computational perspective on the query containment problem.
\end{abstract}

\maketitle

\begingroup\small\noindent\raggedright %
The source code, data, and/or other artifacts have been made available at \url{https://github.com/miniHive/QC4QCoCQ}.
\endgroup

\section{Introduction}

\input{introduction}

\section{Preliminaries}\label{sec:preliminaries}

\input{preliminaries}

\section{Optimization Problem for  $\QCP$}
\label{sec:theory}

\input{theory}

\section{System and Workflow}
\label{sec:system}

\input{system}

\section{Experiments}
\label{sec:experiments}
\input{experiments}

\section{Investigation of Energy Landscapes}
\input{investigation}\label{sec:investigations}

\section{Discussion}
\label{sec:discussion}

\input{discussion}

\section{Related Work}
\label{sec:related_work}
\input{related}

\section{Conclusion}\label{sec:conclusion}
\input{conclusion}

\begin{acks}
 Luisa Gerlach's contribution was 
 funded by the \grantsponsor{DFG-SFB-1404}{Deutsche For\-schungs\-gemeinschaft (DFG, German Research Foundation)}{https://gepris.dfg.de/gepris/projekt/414984028} – Project-ID \grantnum[https://fonda.hu-berlin.de]{DFG-SFB-1404}{414984028 -- CRC 1404 FONDA}.

We are grateful to D-Wave's Leap Quantum LaunchPad program. 

We acknowledge the use of IBM Quantum Credits for this work. The
views expressed are those of the authors, and do not reflect the official
policy or position of IBM or the IBM Quantum team.
\end{acks}

\bibliographystyle{ACM-Reference-Format}
\bibliography{literature}

\newpage
\appendix

\section*{APPENDIX}
 \ \\

\input{appendix_quantum_comp}

\input{appendix_theory}

\input{appendix_system}

\input{appendix_generated_families}
\input{appendix_energy_landscapes}

\end{document}

%% file: macros.tex
\usepackage[export]{adjustbox}

\usepackage{array}

\usepackage{stmaryrd}
\usepackage{xspace}
\usepackage{amsmath}
\usepackage{paralist}
\usepackage{varwidth}
\usepackage{algorithm}
\usepackage{algpseudocode}

\usepackage{textcomp} %

\usepackage{caption}    %
\usepackage{subcaption} %
\usepackage{wrapfig}

\usepackage{tikz}
\usetikzlibrary{automata, graphs,positioning,chains,arrows,snakes,decorations.pathmorphing}

\usetikzlibrary{arrows.meta}

\usepackage{booktabs}
\usepackage{multirow}

\usepackage[normalem]{ulem}

\renewcommand{\epsilon}{\varepsilon}

	\newtheorem{theorem}{Theorem}[section] 
	\newtheorem{lemma}[theorem]{Lemma}
	\newtheorem{corollary}[theorem]{Corollary}

	\theoremstyle{definition}
	\newtheorem{definition}[theorem]{Definition}

	\newtheorem{example}[theorem]{Example}

	\newtheorem{fact}[theorem]{Fact}
	
	\newenvironment{exampleWithEndmarker}{\begin{example}}{\end{example}}      
	
	\tikzset{
		defaultstyle/.style={
			>=stealth, 
			semithick, 
			auto,
			initial text= {},
			initial distance= {3mm},
			accepting distance= {3mm}}}

	\DeclareMathOperator*{\argmin}{\arg\min}            
	
	\newcommand{\nc}[1]{\newcommand{#1}}
	\newcommand{\rnc}[1]{\renewcommand{#1}}
	
	\nc{\myparagraph}[1]{\emph{#1.}}

	\rnc{\leq}{\ensuremath{\leqslant}}
	\rnc{\geq}{\ensuremath{\geqslant}}
	
	\rnc{\le}{\leq}
	\rnc{\ge}{\geq}
	
	\nc{\isdef}{\ensuremath{:=}}
	\nc{\deff}{\isdef}
	
	\nc{\myset}[1]{\ensuremath{\{#1\}}}
	\nc{\setsize}[1]{\ensuremath{|#1|}}
	\nc{\Setsize}[1]{\ensuremath{\big|#1\big|}}
	\nc{\setc}[2]{\set{#1 \, : \, #2}}
	\nc{\Setc}[2]{\Set{#1 \, : \, #2}}
	
	\nc{\aufgerundet}[1]{\ensuremath{\lceil #1 \rceil}}
	\nc{\abgerundet}[1]{\ensuremath{\lfloor #1 \rfloor}}
	
	\nc{\dcup}{\ensuremath{\dot\cup}}
	
	\nc{\ov}[1]{\ensuremath{\overline{#1}}}
	\nc{\tup}[1]{\ensuremath{\bar{#1}}}
	
	\nc{\NN}{\ensuremath{\mathbb{N}}}
	\nc{\NNpos}{\ensuremath{\NN_{\geq 1}}}
        \nc{\ZZ}{\ensuremath{\mathbb{Z}}}
	\nc{\RR}{\ensuremath{\mathbb{R}}}
	\nc{\RRpos}{\ensuremath{\RR_{\geq 0}}}
	\nc{\QQ}{\ensuremath{\mathbb{Q}}}
	\nc{\QQpos}{\ensuremath{\QQ_{> 0}}}
	
	\nc{\poly}{\operatorname{\textit{poly}}}
	\nc{\bigoh}{O}
	\nc{\bigOh}{\bigoh}
	\nc{\trans}{^{\,\mkern-1.5mu\mathsf{T}}}
	
	\nc{\img}[1]{\ensuremath{\textrm{\upshape img}(#1)}} %
	\nc{\dom}[1]{\ensuremath{\textrm{\upshape dom}(#1)}} %
	
	\nc{\emptytuple}{\ensuremath{()}}
	\nc{\emptyword}{\ensuremath{\varepsilon}}

	\nc{\und}{\ensuremath{\wedge}}
	\nc{\Und}{\ensuremath{\bigwedge}}
	\nc{\oder}{\ensuremath{\vee}}
	\nc{\Oder}{\ensuremath{\bigvee}}
	\nc{\nicht}{\ensuremath{\neg}}
	\nc{\impl}{\ensuremath{\to}}
	\nc{\gdw}{\ensuremath{\leftrightarrow}}
	
	\nc{\Semijoin}{\ensuremath{\ltimes}}
	\nc{\proj}{\ensuremath{\pi}}
	\nc{\select}{\ensuremath{\sigma}}
	
	\nc{\free}{\ensuremath{\textrm{\upshape free}}}
	\nc{\quant}{\ensuremath{\textrm{\upshape quant}}}
	\nc{\ar}{\ensuremath{\operatorname{ar}}}
	
	\nc{\sort}{\ensuremath{\textit{sort}}}

	\nc{\Structure}[1]{\ensuremath{\mathcal{#1}}}
	\nc{\A}{\Structure{A}}
	\nc{\B}{\Structure{B}}
	\nc{\C}{\Structure{C}}
	
	\nc{\isom}{\ensuremath{\cong}}
	
	\nc{\Domain}{\ensuremath{\textbf{dom}}}
	\nc{\Var}{\ensuremath{\textbf{var}}}
	\nc{\dbschema}{\ensuremath{\textbf{S}}}
	\nc{\schema}{\dbschema}
	
	\nc{\DB}{\ensuremath{\textbf{I}}} %
	\nc{\DBstrich}{\ensuremath{\textbf{I'}}} %

	\nc{\sql}[1]{\textsf{#1}}	
	\nc{\querycont}{\ensuremath{\sqsubseteq}}
	
	\nc{\QCP}{\ensuremath{\textup{QC-CQ}}}
	
	\nc{\eval}[2]{\ensuremath{#1(#2)}}
	\nc{\semantik}[1]{\ensuremath{\left\llbracket#1\right\rrbracket}}
	\nc{\sem}[1]{{\semantik{#1}}}
	
	\nc{\Vars}{\ensuremath{\textrm{\upshape vars}}}
	\nc{\Cons}{\ensuremath{\textrm{\upshape cons}}}
	\nc{\Atoms}{\ensuremath{\textrm{\upshape atoms}}}
	
	\nc{\myatom}{\ensuremath{\alpha}}
	\nc{\query}{\ensuremath{q}}
     	\nc{\queryi}{\ensuremath{q_i}}
     	\nc{\queryell}{\ensuremath{q_{\ell}}}        
	\nc{\queryOne}{\ensuremath{q_1}}
	\nc{\queryTwo}{\ensuremath{q_2}}
	\nc{\TableauOne}{\ensuremath{\textbf{T}_1}}
	\nc{\AnswerTupleOne}{\ensuremath{v_1}}
	\nc{\TableauTwo}{\ensuremath{\textbf{T}_2}}
	\nc{\AnswerTupleTwo}{\ensuremath{v_2}}
	\nc{\Tableaui}{\ensuremath{\textbf{T}_i}}
	\nc{\AnswerTuplei}{\ensuremath{v_i}}
	\nc{\Tableauell}{\ensuremath{\textbf{T}_{\ell}}}
	\nc{\AnswerTupleell}{\ensuremath{v_{\ell}}}

	\nc{\queryExampleOne}{\ensuremath{q_1}}
	\nc{\queryExampleTwo}{\ensuremath{q_2}}
	\nc{\TableauExampleOne}{\ensuremath{\textbf{T}_1}}
	\nc{\AnswerTupleExampleOne}{\ensuremath{v_1}}
	\nc{\TableauExampleTwo}{\ensuremath{\textbf{T}_2}}
	\nc{\AnswerTupleExampleTwo}{\ensuremath{v_2}}
	\nc{\TableauExampleCycle}{\ensuremath{\textbf{T}_3}}
	\nc{\AnswerTupleExampleCycle}{\ensuremath{()}}

	\nc{\Adom}{\ensuremath{\textrm{\upshape adom}}}
	\nc{\adom}[1]{\ensuremath{\Adom(#1)}} %
	
	\nc{\Ans}{\ensuremath{\textit{ans}}}
	
	\nc{\Hom}{\ensuremath{\textrm{Hom}}}
	\nc{\Embeddings}{\ensuremath{\textrm{Emb}}}
	
	\nc{\Tableau}{\ensuremath{\textbf{T}}}
	\nc{\AnswerTuple}{\ensuremath{v}}

	\nc{\Yes}{\ensuremath{\texttt{yes}}}
	\nc{\No}{\ensuremath{\texttt{no}}}
	
	\nc{\True}{\ensuremath{\texttt{true}}}
	\nc{\False}{\ensuremath{\texttt{false}}}
	
	\nc{\EOE}{\texttt{\upshape EOE}\xspace} %

	\nc{\valuation}{\ensuremath{\nu}}
	
	\nc{\Partition}{\ensuremath{P}}
	\nc{\Attr}[1]{\ensuremath{\textsf{\upshape #1}}}
	
	\nc{\SMOP}{\ensuremath{\text{\normalfont{\textsc{GMOP}}}}}
	\nc{\DreiSMOP}{\ensuremath{\text{\normalfont{\textsc{3GMOP}}}}}
	\nc{\DDreiSMOP}{\ensuremath{\text{\normalfont{\textsc{D-3GMOP}}}}}
	
	\nc{\DreiSAT}{\ensuremath{\text{\normalfont{\textsc{3SAT}}}}}
	\nc{\SDreiSAT}{\ensuremath{\text{\normalfont{\textsc{Special-3SAT}}}}}
	
	\nc{\SetOfClauses}{\ensuremath{\mathscr{C}}}
	\nc{\AnotherSetOfClauses}{\ensuremath{\mathscr{K}}}
	
	\nc{\BVars}[1]{\ensuremath{\textit{Vars}(#1)}}
	
	\nc{\Sol}[1]{\ensuremath {\textit{Sol}(#1)}}
	\nc{\OptSol}[1]{\ensuremath {\textit{Opt}(#1)}}
	\nc{\MaxVal}[1]{\ensuremath {\textit{max}(#1)}}
	
	\nc{\tvalid}{\ensuremath{t_{\textit{Sol}}}}
	
	\nc{\degree}[2]{\ensuremath{\textit{deg}_{#1}(#2)}}
	
	\nc{\Simpl}[1]{{\ensuremath{\hat{#1}}}}
	
	\nc{\fnonisol}{\ensuremath{f_{\textup{no-isol}}}}
	\nc{\ffan}{\ensuremath{f_{\textup{fan-shaped}}}}
	\nc{\favg}{\ensuremath{f_{\textup{avg}}}}
	
	\nc{\sfnonisol}{\ensuremath{\Simpl{f}_{\textup{no-isol}}}}
	\nc{\sffan}{\ensuremath{\Simpl{f}_{\textup{fan-shaped}}}}
	\nc{\sfavg}{\ensuremath{\Simpl{f}_{\textup{avg}}}}

	\nc{\aof}[1]{\ensuremath{\alpha_{#1}}} %
	\nc{\gof}[1]{\ensuremath{G_{#1}}} %
	\nc{\Xof}[1]{\ensuremath{X_{#1}}} %
	\nc{\GSol}[0]{\ensuremath{Sol^G(G)}} %
	\nc{\GOpt}[0]{\ensuremath{Sol^G_\text{Opt}(G)}} %
	\nc{\BSol}[0]{\ensuremath{Sol^B(G)}} %
	\nc{\BOpt}[0]{\ensuremath{Sol^B_\text{Opt}(G)}} %
	
	\nc{\punique}{\ensuremath{p_{\textup{fct}}}}
	\nc{\pac}{\ensuremath{p_{\textup{(3)}}}}
        
	\nc{\subs}{\ensuremath{\tilde{f}}}
	
	\usetikzlibrary{shapes.misc}

\newcommand{\ovxext}{\ensuremath{\ov{x}^{\textup{(e)}}}}
\newcommand{\xext}[1]{\ensuremath{
    x^{\textup{(e)}}_{#1}}}

\newcommand{\ovxsim}{\ensuremath{\ov{x}^{\textup{(s)}}}}
\newcommand{\xsim}[1]{\ensuremath{
    x^{\textup{(s)}}_{#1}}}

\newcommand{\generic}{\ensuremath{\textit{gen}}}   
\newcommand{\pgen}{\ensuremath{p_{\generic}}}
\newcommand{\dgen}{\ensuremath{{\myd}_{\generic}}}
\newcommand{\Cgen}{\ensuremath{C_{\generic}}}

\newcommand{\simplBtwo}{\ensuremath{B^{\text{(s)}}_2}}
\newcommand{\simplntwo}{\ensuremath{n^{\text{(s)}}_2}}

\newcommand{\simplified}{\ensuremath{\textit{sim}}}   
\newcommand{\psimpl}{\ensuremath{p_{\simplified}}}
\newcommand{\dsimpl}{\ensuremath{{\myd}_{\simplified}}}
\newcommand{\Csimpl}{\ensuremath{C_{\simplified}}}

\newcommand{\constrained}{\ensuremath{c}}

\newcommand{\pcgen}{\ensuremath{p^{\constrained}_{\generic}}}
\newcommand{\dcgen}{\ensuremath{{\myd}^{\constrained}_{\generic}}}
\newcommand{\Ccgen}{\ensuremath{C^{\constrained}_{\generic}}}

\newcommand{\pcsimpl}{\ensuremath{p^{\constrained}_{\simplified}}}
\newcommand{\dcsimpl}{\ensuremath{{\myd}^{\constrained}_{\simplified}}}
\newcommand{\Ccsimpl}{\ensuremath{C^{\constrained}_{\simplified}}}

\nc{\puniquesimpl}{\ensuremath{p^{\text{(s)}}_{\textup{fct}}}}
\nc{\pacsimpl}{\ensuremath{p^{\text{(s)}}_{\textup{(3)}}}}

\nc{\LosAngeles}{L.A.}%
\nc{\UnitedStates}{U.S.}%

\nc{\vvarx}{X}
\nc{\vvary}{Y}
\nc{\vvarz}{Z}
\nc{\vvarw}{W}
\nc{\cconstc}{c}

\nc{\myd}{d}

%% file: introduction.tex
The question of whether one query is contained in another is fundamental: it asks whether, for any possible input, the result of one query will always be a subset of the other.
Being able to decide query containment opens numerous opportunities in data management, in particular  query optimization, but also data integration, semantic caching, or access control.
Variations of the problem have been studied in database theory for decades~(cf.\ e.g.\ \cite{ahv-book,DBLP:books/cs/Maier83,DBLP:journals/pacmmod/MarcinkowskiO25}),
and for various query languages.
Yet, due to the high computational cost of deciding query containment (NP-complete for conjunctive queries with set semantics) or outright infeasibility (undecidable for relational algebra),
commercial database systems are restricted to using heuristics or rule-based rewrites.
In this work, we focus on the
query containment problem
for conjunctive queries under set semantics (\QCP)%
~\cite{ChandraMerlin}.
In particular, we investigate
whether this problem can be formulated for solving on quantum computers.

With access to commercially available quantum hardware,
the database research community has begun to explore which core database problems may benefit from offloading to quantum devices~\cite{DBLP:conf/vldb/YuanLCWYYL023,yuan2024quantumcomputingdatabasesoverview,DBLP:journals/pvldb/HaiHCLG25,DBLP:conf/sigmod/WinkerGUYLFM23,DBLP:journals/pvldb/CalikyilmazGGWPSAPG23}.
At this point, this line of research must remain very exploratory, as
today's early hardware prototypes are far from usable as co-processors in industrial-scale database systems. One limited resource is the number of available qubits or links (couplers) between them, often restricting experiments to solving toy examples. Further, today's systems suffer from high levels of hardware noise that accumulates with increasing circuit size, so that results for larger inputs are often erroneous. Given the early technological state-of-the-art, the goal cannot yet realistically be to demonstrate a quantum advantage.
Rather, we can ask how certain problems can be formulated  so that they may eventually be offloaded to quantum hardware, once the technology reaches the required maturity and performance. 
At the same time, joining this quest early rather than late allows the database research community to take an active role in shaping quantum hardware and software ecosystems, as they emerge.

Much of recent research in this field focuses on formulating instances of optimization problems in QUBO form~\cite{DBLP:conf/icde/0003HF24}, where the objective is to minimize polynomials of degree at most two. This enables execution on quantum annealers, which currently provide access to thousands of qubits. Optimization problems involving higher-degree polynomials are typically addressed using the QAOA algorithm \cite{farhi_quantum_2014}.
QAOA follows an annealing-inspired strategy but targets another hardware architecture:
gate-based quantum processing units (QPUs), which can perform universal computations, yet are currently more restricted in the number of available qubits.

In the context of query optimization, the database research community has mainly been exploring the important task of join ordering~\cite{dblp:journals/pacmmod/schonbergersm23,DBLP:journals/pvldb/LiuSS25,DBLP:conf/qce/FranzWGM24,DBLP:conf/q-data/LiuKSS25,10.1145/3579142.3594298,10.1145/3736393.3736690}, as well as plan selection in multi-query optimization~\cite{DBLP:journals/pvldb/Trummer016,DBLP:journals/access/FankhauserSFS23}. These problems fall into the category of \emph{cost-based query optimization}, where the aim is to 
choose among algebraically equivalent plans.
In contrast, being able to decide query containment is a tool for \emph{semantic query optimization}, and may allow to rewrite a given query into a simpler or smaller one. 
Yet to our knowledge, query containment has not yet been explored from a quantum computing perspective.
Unlike many of the aforementioned query optimization tasks, there are no established classical implementations in today’s relational database systems with which to compete.
This makes this problem particularly intriguing.
\smallskip

\myparagraph{Contributions}
This paper makes the following contributions.
\begin{enumerate}[1.]
    \item We present a novel formulation of \QCP, the query containment problem for conjunctive queries under set semantics, as the problem of minimizing a polynomial. While related work oftentimes introduces such polynomials on the level of intuition, we formally prove our formulation to be correct. 

    \item     
    Our problem formulation can be solved on quantum hardware. Specifically, polynomials with degree of up to~2 can be directly mapped to quantum annealers (which currently provide qubits in the range of thousands). For polynomials of arbitrary degree, we use QAOA on circuit-based QPUs (given enough qubits).

    \item For future integration into a query engine, it is vital that our approach to containment checking only errs on the safe side. Otherwise, a query optimizer risks producing incorrect results. In fact, our approach is guaranteed to never produce any false positives, as all positives come with a verifiable witness. 
    
    \item We present a 3-stage workflow to decide \QCP. During preprocessing, we perform a simplification step, which we prove to be correct. Our experiments show that simplification is highly effective: by reducing the number of variables and the degree of the polynomial, the problem can be solved more efficiently.

    \item Our experiments include simulators and real quantum hardware from leading vendors, using both quantum annealers and circuit-based QPUs.
      We evaluate soundness, scalability, and parameter sensitivity. In particular, we demonstrate the positive effects of using constraints in QAOA. 
    
    \item Our experiments indicate that there are instances of the containment problem that display energy landscapes that are particularly well-suited for our approach. This opens up a fresh perspective on the well-studied problem of query containment.
      
\end{enumerate}

\myparagraph{Structure} We review preliminaries on query containment and quantum computing (Sec.~\ref{sec:preliminaries}), before we derive our formalization for checking \QCP\ (Sec.~\ref{sec:theory}).
We present our workflow (Sec.~\ref{sec:system}) and its experimental evaluation (Sec.~\ref{sec:experiments}). We introduce a new perspective on containment problems by investigating their energy landscapes (Sec.~\ref{sec:investigations}), before we discuss our insights at a higher level of abstraction (Sec.~\ref{sec:discussion}). We review the related work (Sec.~\ref{sec:related_work}) and conclude (Sec.~\ref{sec:conclusion}).

%% file: preliminaries.tex
\input{prelim_query-containment}

\input{prelim_quantum_comp}

%% file: prelim_query-containment.tex
\subsection{Conjunctive Queries \& Query Containment}
\label{sec:prelims_qc}

We introduce conjunctive queries and the query containment problem, first by example and then formally. 
We assume set semantics, so relations and query results do not contain duplicate tuples.
In formalizations, we adopt notation from Abiteboul, Hull, and Vianu~\cite{ahv-book}.

\myparagraph{Conjunctive Queries}
We consider conjunctive queries that can be expressed by core
SQL queries of the form ``$\sql{select distinct}$ \ldots $\sql{from}$ \ldots
$\sql{where}$ \ldots'', whose \sql{where}-clause is a conjunction of
equalities, and we use the syntax of \emph{tableau queries} (cf.,
\cite{ahv-book}) to express these queries. The following
example is taken from \cite[Chapter~15]{ArenasEtAl_DBT-Book}.

\begin{example}\label{example:queriesSQLtoCQ}
Consider the schema consisting of three relations
\sql{Person[pid,pname,cid]}, \sql{Profession[pid,prname]},
\sql{City[cid,\allowbreak cname,\allowbreak country]} and the SQL
query $q_1$ shown in Figure~\ref{fig:sql-vs-tableau} on the left.

\begin{figure}[thb]
\centering

\makebox[\columnwidth]{%
\begin{minipage}[t]{0.55\columnwidth}
\vspace{0pt}

{\footnotesize\ttfamily
\begin{tabular}{@{}l@{}}
SELECT DISTINCT A.pname\\
FROM Person A, Profession B, City C\\
WHERE A.pid = B.pid\\
\hspace{1.4em}AND B.prname = \textquotesingle actor\textquotesingle\\
\hspace{1.4em}AND A.cid = C.cid\\
\hspace{1.4em}AND C.cname = \textquotesingle L.A.\textquotesingle\\
\hspace{1.4em}AND C.country = \textquotesingle U.S.\textquotesingle
\end{tabular}
}
\end{minipage}%
\hfill
\begin{minipage}[t]{0.42\columnwidth}
\vspace{0pt}

{\footnotesize
\setlength{\tabcolsep}{3pt}

\begin{tabular}{@{}l|>{\centering\arraybackslash}p{2em}
                 >{\centering\arraybackslash}p{3em}
                 >{\centering\arraybackslash}p{1.8em}@{}}
\sql{Person} & \sql{pid} & \sql{pname} & \sql{cid}\\
\hline
& $\vvarx_1$ & $\vvary_1$ & $\vvarz_1$
\end{tabular}

\vspace{0.6ex}

\begin{tabular}{@{}l|>{\centering\arraybackslash}p{2em}
                 >{\centering\arraybackslash}p{3.4em}@{}}
\sql{Profession} & \sql{pid} & \sql{prname}\\
\hline
& $\vvarx_1$ & $\sql{\textquotesingle actor\textquotesingle}$
\end{tabular}

\vspace{0.6ex}

\begin{tabular}{@{}l|>{\centering\arraybackslash}p{2em}
                 >{\centering\arraybackslash}p{3em}
                 >{\centering\arraybackslash}p{3em}@{}}
\sql{City} & \sql{cid} & \sql{cname} & \sql{country}\\
\hline
& $\vvarz_1$ & \sql{\textquotesingle L.A.\textquotesingle} & \sql{\textquotesingle U.S.\textquotesingle}
\end{tabular}
} %
\end{minipage}%
} %

\caption{SQL query $q_1$ (left) and relations (right).}
\label{fig:sql-vs-tableau}
\end{figure}

The same query can be expressed in the syntax of \emph{tableau
  queries} (a language similar to \emph{Query By
  Example (QBE)}~\cite{ahv-book}) %
via the \emph{answer tuple} $(\vvary_1)$ and
the three relations shown 
 on the right in Figure~\ref{fig:sql-vs-tableau}.

For answering this tableau query on a database instance, one
strives for \emph{valuations}, i.e., mappings $\valuation$ from the
variables $\{\vvarx_1,\vvary_1,\vvarz_1\}$ that occur in the query to the
database entries, such that tuple
$(\valuation(\vvarx_1),\valuation(\vvary_1),\valuation(\vvarz_1))$ belongs to relation $\sql{Person}$, %
$(\valuation(\vvarx_1),\sql{\textquotesingle actor\textquotesingle})$ belongs to relation
$\sql{Profession}$, and %
$(\valuation(\vvarz_1),\sql{\textquotesingle \LosAngeles\textquotesingle},\sql{\textquotesingle \UnitedStates\textquotesingle})$ belongs to relation $\sql{City}$.
Note how the tableau query uses shared variables, such as~$\vvarx_1$, to
declare joins (as familiar from query languages like Datalog).
The \emph{query result}
is the set of
values $\valuation(\vvary_1)$ for all such valuations~$\valuation$. Thus, the answer tuple is a means of declaring projections.
\end{example}

We now formalize these concepts.
We write $\ZZ$ for integers, $\NN$ 
for non-negative integers, and we let
$\NNpos\deff\NN\setminus\{0\}$.

The \emph{domain} of potential database entries is a fixed countably
infinite set $\Domain$; elements in $\Domain$ are called
\emph{constants}, and in concrete examples we will usually write them
in quotes \textquotesingle\ldots\textquotesingle.
A \emph{schema} is a finite set $\dbschema$
of \emph{relation names}, where each $R\in\dbschema$ is equipped
with a fixed \emph{arity} $\ar(R)\in\NN$.
A \emph{database instance} of schema $\dbschema$ (for short:
$\dbschema$-db) is a mapping $\DB$ that associates with each
$R\in\dbschema$ a finite relation $\DB(R)\subseteq \Domain^{\ar(R)}$.

The conjunctive queries considered in this paper are
defined as follows.
We let $\Var$ be 
an
infinite set of \emph{variables}
with $\Var\cap\Domain=\emptyset$.
A \emph{free tuple} is an expression of the form 
$(u_1,\ldots,u_r)$ with $r\in\NN$ and
$u_1,\ldots,u_r\in\Domain\cup\Var$.
For a free tuple $u=(u_1,\ldots,u_r)$ we let $\ar(u)\deff
r$ be the tuple's \emph{arity}, and we let
$\Vars(u)$
and
$\Cons(u)$
be the set of
variables and constants, resp., that occur in $u$. 
For a set~$S$ of free tuples, we let
$\Vars(S)$
and
$\Cons(S)$
be the set of variables and
constants, resp., that occur in some tuple in $S$.

A \emph{conjunctive query} of schema $\dbschema$ (for short:
$\dbschema$-CQ)
is of the form~$(\Tableau,\AnswerTuple)$ 
where $\Tableau$ is a mapping that associates with each
$R\in\dbschema$ a finite set $\Tableau(R)$ of free tuples of arity
$\ar(R)$, and $v$ is a free tuple with
$\Vars(v)\subseteq\Vars(\Tableau)\deff
\bigcup_{R\in\dbschema}\Vars(\Tableau(R))$. Accordingly, %
$\Cons(\Tableau)\deff \bigcup_{R\in\dbschema}\Cons(\Tableau(R))$.
For a query $\query=(\Tableau,\AnswerTuple)$, we call~$\Tableau$ and
$\AnswerTuple$ the query's \emph{tableau} and \emph{answer tuple},
resp., and we let $\Vars(\query)\deff
\Vars(\Tableau)$ be the set of variables
that occur in $\query$. The \emph{active domain} of $\query$ is the
set $\Adom(\query)\deff
\Cons(\AnswerTuple)\cup\Cons(\Tableau)$ of all
constants %
in $\query$.

\begin{example}\label{example:queriesSQLtoCQcontinued}
With this notation, query $\queryExampleOne$ from
Example~\ref{example:queriesSQLtoCQ} is the query
$(\TableauExampleOne,\AnswerTupleExampleOne)$ with
$\AnswerTupleExampleOne = (\vvary_1)$
and
$\TableauExampleOne(\sql{Person}) = \set{(\vvarx_1,\vvary_1,\vvarz_1)}$,
$\TableauExampleOne(\sql{Profession}) = \set{(\vvarx_1,\sql{\textquotesingle actor\textquotesingle})}$,
$\TableauExampleOne(\sql{City}) =$
$ \set{(\vvarz_1, \sql{\textquotesingle \LosAngeles\textquotesingle}, \sql{\textquotesingle \UnitedStates\textquotesingle})}$.
\end{example}

We consider the \emph{set semantics} of conjunctive queries, which
is defined as follows. A \emph{valuation} for a query $\query$ is a
mapping $\valuation\colon\Vars(\query)\cup\Domain \to \Domain$ with
$\valuation(\cconstc)=\cconstc$ for every $\cconstc\in\Domain$.
For a free tuple $u=(u_1,\ldots,u_r)$ of arity $r\geq 0$ we write
$\valuation(u)$ as a shorthand for $(\valuation(u_1),\ldots,\valuation(u_r))$.
An \emph{embedding} of
an $\dbschema$-CQ $\query=(\Tableau,\AnswerTuple)$
into an $\dbschema$-db $\DB$ is a
valuation $\valuation$ for $\query$ such that for every
$R\in\dbschema$ and every free tuple $u\in \Tableau(R)$ we have
$\valuation(u)\in\DB(R)$. We write $\Embeddings(\query,\DB)$ for
the set of all embeddings of $\query$ into $\DB$.
The \emph{query result}
of $\query=(\Tableau,\AnswerTuple)$ on $\DB$ is the set
$\sem{\query}(\DB)
 \deff  
 \setc{\valuation(v)}{\valuation\in\Embeddings(\query,\DB)}$.

A \emph{Boolean query} is an
$\dbschema$-CQ $\query=(\Tableau,\AnswerTuple)$ where $\AnswerTuple$
is the empty tuple~$\emptytuple$, i.e., the tuple of arity $0$.  The query result
$\sem{\query}(\DB)$ then is either the set $\{\emptytuple\}$ or the empty
set, corresponding to the answers ``yes'' and ``no'', respectively. 
See Example~\ref{example:preparation:cycle_chain} for two example queries.

\myparagraph{Query Containment}
For two queries $\queryOne$, $\queryTwo$ of schema $\dbschema$
we write $\queryOne\querycont\queryTwo$ and say that $\queryOne$ is
\emph{contained} in $\queryTwo$, if for all $\dbschema$-dbs $\DB$ we
have $\sem{\queryOne}(\DB)\subseteq\sem{\queryTwo}(\DB)$.
We write $\queryOne\not\querycont\queryTwo$ to indicate that
$\queryOne$ is \emph{not} contained in $\queryTwo$ (i.e.,
$\sem{\queryOne}(\DB)\not\subseteq\sem{\queryTwo}(\DB)$ for some
$\schema$-db $\DB$).

\begin{example}\label{example:QueryContainment}
Consider the schema and query $\queryExampleOne$ from
Examples~\ref{example:queriesSQLtoCQ} and
\ref{example:queriesSQLtoCQcontinued}, and consider query

\noindent
$\queryExampleTwo$:\quad
{\small\ttfamily
\renewcommand{\arraystretch}{0.85} 
\begin{tabular}[t]{@{}l@{}}
SELECT DISTINCT A.pname\\
FROM Person A, Profession B\\ 
WHERE A.pid = B.pid
\end{tabular}
}

\noindent
which is also taken from \cite[Chapter~15]{ArenasEtAl_DBT-Book}.
In the syntax defined above, the query $\queryExampleTwo$ is the query
$(\TableauExampleTwo,\AnswerTupleExampleTwo)$ with
$\AnswerTupleExampleTwo=(\vvary_2)$ and
$\TableauExampleTwo(\sql{Person}) = \set{(\vvarx_2,\vvary_2,\vvarz_2)}$,
$\TableauExampleTwo(\sql{Profession}) = \set{(\vvarx_2,\vvarw_2)}$,
$\TableauExampleTwo(\sql{City}) = \emptyset$.

By a close inspection of $\queryExampleOne$ and $\queryExampleTwo$, it
is easy to see that $\queryExampleTwo$ is more general than
$\queryExampleOne$, i.e.,
$\sem{\queryExampleOne}(\DB)\subseteq\sem{\queryExampleTwo}(\DB)$
holds for every database instance $\DB$. Thus, we have
$\queryExampleOne\querycont\queryExampleTwo$. But there exist
database instances~$\DB$ such that
$\sem{\queryExampleTwo}(\DB)\not\subseteq\sem{\queryExampleOne}(\DB)$,
and thus
we have $\queryExampleTwo\not\querycont\queryExampleOne$.
\end{example}

Chandra and Merlin's seminal paper \cite{ChandraMerlin} provided a
characterization of query containment in terms of homomorphisms. To
state their theorem, we need the following notation.
Given a function $h\colon X\to Y$ (for some sets $X$ and
$Y$) and a $k$-tuple $v=(x_1,\ldots,x_k)\in X^k$, we write
$h(v)$ for the $k$-tuple $(h(x_1),\ldots,h(x_k))$. For
$S\subseteq X^k$ we let $h(S)\deff \setc{h(v)}{v\in S}$.

\begin{definition}\label{def:HomomorphismForQueries}
Let $\queryOne=(\TableauOne,\AnswerTupleOne)$ and $\queryTwo=(\TableauTwo,\AnswerTupleTwo)$ be $\dbschema$-CQs,
for some schema $\dbschema$. A
\emph{homomorphism} from $\queryTwo$ to $\queryOne$ is a mapping
$h\colon\Vars(\queryTwo)\cup\Adom(\queryTwo) \to
\Vars(\queryOne)\cup\Adom(\queryOne)$ satisfying

\begin{enumerate}
\item\label{itemOne:def:homom}
  $h(\cconstc)=\cconstc$, for every $\cconstc\in\Adom(\queryTwo)$,
\item\label{itemTwo:def:homom}
  $h(\AnswerTupleTwo)=\AnswerTupleOne$, and
\item\label{itemThree:def:homom}
  for every $R\in\dbschema$ and every $u\in \TableauTwo(R)$
  we have $h(u)\in\TableauOne(R)$.
\end{enumerate}  
\end{definition}

\begin{theorem}[Chandra and Merlin's Homomorphism Theorem \cite{ChandraMerlin}]\label{thm:ChandraMerlin}
For all schemas $\dbschema$ and all $\dbschema$-CQs $\queryOne$ and $\queryTwo$ we have
\(
  \queryOne\querycont\queryTwo
  \ \ \iff \ \ 
  \text{there exists a homomorphism from $\queryTwo$ to $\queryOne$.}
\)  
\end{theorem}

\begin{example}\label{example:QueryContainmentContinued}
Consider the queries $\queryExampleOne,\queryExampleTwo$ from
Examples~\ref{example:queriesSQLtoCQcontinued} and~\ref{example:QueryContainment}.
Let
$h\colon\set{\vvarx_2,\vvary_2,\vvarz_2,\vvarw_2}\to\set{\vvarx_1,\vvary_1,\vvarz_1,\sql{\textquotesingle actor\textquotesingle},\allowbreak\sql{\textquotesingle\LosAngeles\textquotesingle},\sql{\textquotesingle \UnitedStates\textquotesingle}}$
be the mapping defined via
$h(\vvarx_2)=\vvarx_1$,
$h(\vvary_2)=\vvary_1$,  $h(\vvarz_2)=\vvarz_1$, and $h(\vvarw_2)=\sql{\textquotesingle
  actor\textquotesingle}$.
It can easily be verified that $h$
is a homomorphism from $\queryExampleTwo$ to $\queryExampleOne$ (because
$h(\AnswerTupleExampleTwo)=h((\vvary_2))=(h(\vvary_2))=(\vvary_1)
=\AnswerTupleExampleOne$
and \allowbreak
$h((\vvarx_2,\vvary_2,\vvarz_2)) \allowbreak =(\vvarx_1,\vvary_1,\vvarz_1)\in\TableauExampleOne(\sql{Person})$ and
\allowbreak $h((\vvarx_2,\vvarw_2))\allowbreak =(\vvarx_1,\sql{\textquotesingle actor\textquotesingle})\in\TableauExampleOne(\sql{Profession})$).
Theorem~\ref{thm:ChandraMerlin} yields that $\queryExampleOne\querycont\queryExampleTwo$.
Yet there is no homomorphism from $\queryOne$ to
$\queryTwo$, as $\TableauExampleOne(\sql{Profession})$ contains
the tuple $(\vvarx_1,\sql{\textquotesingle actor\textquotesingle})$, and any homomorphism~$h$ would have
to map the constant $\sql{\textquotesingle actor\textquotesingle}$ to itself, i.e.,
$h(\sql{\textquotesingle actor\textquotesingle})=\sql{\textquotesingle actor\textquotesingle}$, but
$\TableauExampleTwo(\sql{Profession})$ does not contain any tuple with
an entry $\sql{\textquotesingle actor\textquotesingle}$ in its second component. Hence,
Theorem~\ref{thm:ChandraMerlin} yields that $\queryExampleTwo\not\querycont\queryExampleOne$.
\end{example}

The \emph{query containment problem for conjunctive queries} w.r.t.\
set semantics ($\QCP$) is the following decision problem:
given a pair $(\queryOne, \queryTwo)$ of $\dbschema$-CQs (for some schema $\dbschema$),
decide whether $\queryOne\querycont\queryTwo$.

By Theorem~\ref{thm:ChandraMerlin}, $\QCP$ can be solved by an
algorithm that checks (e.g., by an exhaustive search)
whether there is a homomorphism from $\queryTwo$ to $\queryOne$. 
This problem is NP-complete~\cite{ChandraMerlin},  
even when
inputs are restricted to Boolean queries and any fixed schema consisting of one (or
several) relation(s) of arity~2.

Any algorithm that solves $\QCP$ also yields an algorithm for checking if two queries $\queryOne$
and $\queryTwo$ are \emph{equivalent} --- simply, by checking whether both
$\queryOne\querycont\queryTwo$ and $\queryTwo\querycont\queryOne$ hold.

%% file: prelim_quantum_comp.tex
\subsection{Quantum\,Computing\;\&\;Simulated\,Annealing}
\label{sec:preliminaries-quantum}

\newcommand{\R}{\mathbb R}
\newcommand{\Z}{\mathbb Z}
\newcommand{\N}{\mathbb N}
\newcommand{\Q}{\mathbb Q}
\newcommand{\D}{\mathcal D}
\newcommand{\F}{\mathcal F}

This section is written from a non-physicist perspective and aims to provide an overview, giving sufficient information to database researchers for informed use of the concepts from quantum computing considered in this paper; more background can be found in Appendix~\ref{sec:appendix_quantum_comp}. 
Particular focus is put on quantum computing methods for solving optimization problems of the form
\begin{align}
\label{optimization_problem}
\ov{x}\ \in \ \argmin_{\ov{x}\in C}p(\ov{x}).
\end{align}
Here, $p\colon \myset{0,1}^n\rightarrow\R$ is the \emph{objective function}; it is 
a polynomial on $n$ binary variables $\ov{x}=(x_1,\dotsc,x_n)$ with real-valued coefficients.
The \emph{search space} is the set
$C\subseteq\myset{0,1}^n$. The goal is to find a tuple $\ov{x}\in C$ such that the value $p(\ov{x})$ is as small as possible. This optimization problem is called \emph{unconstrained} if $C=\myset{0,1}^n$ and \emph{constrained} if $C\subsetneq\myset{0,1}^n$.
We write $p$ as a sum of \emph{monomials}, with each monomial being of the form $c\cdot\prod_{i=1}^{n}(x_i)^{j_i}$ where $c$ is a real-valued coefficient, and $j_i$ is a non-negative integer exponent associated with variable~$x_i$ (and, by convention, $(x_i)^0=1$).
The \emph{degree} of such a monomial is defined as the sum $\sum_{i=1}^n{j_i}$, and the degree of~$p$ is the maximum degree of its monomials.
Since each variable $x_i$ can only be assigned the values~0 and~1, and $0\cdot0=0$ and $1\cdot1=1$, any exponent $j_i\geq 1$ can be replaced with any positive integer $j'_i$ without changing the function values of the polynomial.
If the degree of $p$ is $\leq 2$, %
we can assume without loss of generality that \emph{every}
monomial of $p$ has degree exactly 2;
this setting is known as \emph{Quadratic Unconstrained Binary Optimization} problems (for short: QUBOs), cf.\ e.g.~\cite{DBLP:conf/icde/0003HF24}.

\myparagraph{Quantum Approximate Optimization Algorithm for gate-based quantum computers}\label{sec:QAOA}
QAOA \cite{farhi_quantum_2014} is a hybrid quantum-classical algorithm designed for combinatorial optimization. It iteratively evaluates the objective function's expected value for a parameter-dependent quantum state on a quantum processor, using classical feedback to update parameters and converge towards an optimal state. The quantum circuit consists of several \emph{layers} that transition the system from the ground state of a so-called \emph{mixing Hamiltonian} (describes the initial state) to that of the so-called \emph{problem Hamiltonian} (describes the objective function). The problem solution is found by sampling the final state. 
Effectiveness is enhanced by %
incorporating problem constraints into the QAOA quantum circuit design~\cite{bartschi2020grover, zander-solving-2024}. Then, only a subset of feasible solutions is explored, and the reduced size of the search space improves results. 

The main complexity parameters of QAOA are the number of \emph{layers} of the circuit, the number of \emph{iterations} of the classical optimizer, and the number of \emph{shots} performed to estimate the expected value of the objective function in each iteration.

The complexity of the quantum circuit depends on the specific problem constraints and the objective function $p$. 
In the following, we argue that the complexity scales polynomially in the number~$n$ of variables. This indicates that the problem can be efficiently executed on a quantum computer:
For an arbitrary \emph{quadratic} polynomial $p$ of~$n$ binary variables $\ov{x}=(x_1,\ldots,x_n)$ and the constraints considered in this work, the circuits for QAOA with~$\ell$ layers have a favorable scaling of $\mathcal O(\ell\cdot n^2)$ gates, $\mathcal O(\ell\cdot n)$ depth (runtime), and $\mathcal O(n)$ qubits for both the constrained and unconstrained version \cite{bartschi2020grover, zander-solving-2024}. 
In general, for a fixed degree $d\leq n/2$ square-free polynomial $p$ of $n$ binary variables, the maximum number of monomials is $\mathcal O(n^d)$ (practical instances are typically much sparser.). A degree $d$ monomial requires $\mathcal O(d)$ gates and circuit depth, and only a single ancilla qubit \cite{Khattar2025}. It follows that the circuits for QAOA with~$\ell$ layers require $\mathcal O(\ell\cdot d\cdot n^d)$ gates and depth, and $\mathcal O(n)$ qubits, which is a polynomial scaling in the number of variables $n$.
Hence,
they can be efficiently executed on a quantum computer.

Currently available quantum hardware --- referred to as Noisy Intermediate-Scale Quantum (NISQ) hardware --- suffers from significant drawbacks: only a small number of physical qubits (e.g., up to 156 qubits on state-of-the art IBM devices) are available, and computations are affected by limitations such as short coherence times that allow only shallow circuits to be run, and noise leading to gate and readout errors. Therefore, quantum algorithms currently are still mostly studied using simulators that run on classical hardware. In general, handling quantum states of $n$ qubits requires storing and manipulating $2^n$ complex amplitudes, an exponential scaling that limits classical emulation to the small-scale regime. In practice, matrix product state simulators (e.g., provided by Qiskit Aer \cite{qiskit2024}) offer an efficient tool for approximately simulating quantum computations that rely on only a restricted amount of entanglement~\cite{vidal2003efficient}.

\myparagraph{Simulated annealing and quantum annealing}
\label{sec:SAQA}
Next, we present the foundations of two further
methods that can be used to solve a QUBO, i.e., problem~\eqref{optimization_problem} (see above)
in the \emph{unconstrained} case where all monomials of $p$ have \emph{degree~2}.
These methods are referred to as \emph{simulated annealing} and \emph{quantum annealing}. The term ``annealing'' is inspired by a cooling technique from metallurgy: heated metal is cooled in a controlled way such that its atoms have sufficient time to arrange themselves and form stable crystals. This results in a low-energy state of the metal that is considered to be stable. Both simulated annealing and quantum annealing are heuristic search methods moving in a controlled way over a so-called \emph{solution landscape} or \emph{energy landscape}
(i.e., possible %
assignments of the binary variables; henceforth these are called \emph{states})
for a sufficiently long time until they reach a stable and minimizing (w.r.t. the polynomial value) state. Thus, both methods contain the term ``annealing'' in their names. A minimum state corresponds to an optimal solution.
In this work, we use three different solvers provided by the D-Wave library
\cite{DWaveLibrary-Link},
which are briefly described in the following.

\myparagraph{Simulated annealing} This solver emulates a controlled cooling procedure by considering a parameter $\beta$, which is related to a possible temperature $T$ of a cooling process via
$\beta = 1/\left(k_b\cdot T\right)$, where $k_b$ denotes the Boltzmann constant. The cooling process is defined by a certain range for $\beta$ limited by some $\beta_\text{start}$ and $\beta_\text{end}$ with $\beta_\text{start}\leq \beta\leq\beta_\text{end}$.
Typically, one chooses for $\beta_\text{start}$ a small value to start at a relatively high temperature, while $\beta_\text{end}$ is large corresponding to a low temperature;
$\beta$ influences the probability of circumventing a barrier (value difference) in the energy landscape between neighboring states
(i.e., states that differ in one bit).
As $\beta$ increases (i.e., $T$ decreases), the probability of leaving a lower position decreases, such that the solver becomes less exploratory with time
and finally terminates at a possible global minimum. Besides $\beta$, the difference between the values of neighboring states influences the probability of moving between them. In case of large positive differences
w.r.t.\
$1/\beta$, it is less likely that the solver leaves a current state. Otherwise, if there are negative differences, i.e. if the value of the neighboring state is lower, the current state will be left and the lower-valued neighbor will always be attained (cf., \cite[Sec. 3.1]{mcgeoch2022adiabatic}). 

\myparagraph{Quantum annealing} Contrary to simulated annealing, quantum annealing does not emulate a cooling process, but rather a controlled transfer from a quantum state that can be easily prepared to a quantum state that yields with high probability a global minimizer within a given energy landscape.
This is done by introducing an annealing time $t$ that ranges from $0$ to some final time $t_f$.
For a certain annealing time~$t$, the quantum state is a weighted sum of the initial state and the final state. The weights in this sum 
are slowly
adapted such that at the end of the annealing process, 
the final state is overweighted compared to the initial state (cf., \cite[Sec.~3.2 \& Chap.~4]{mcgeoch2022adiabatic} and \cite{rajak2023quantum}). In contrast to simulated annealing, the probability of circumventing energy barriers for quantum annealing only depends on the width of the barrier 
(distance between states)
rather than the height
(difference between values).

\myparagraph{Simulated quantum annealing} This solver is a classical version of quantum annealing, i.e.\ it emulates the behavior of a quantum annealer on a classical computer. The weights for the initial state and the final state are again determined by a parameter $\beta$, which is chosen in a similar way as in the
case of simulated annealing \cite{crosson2016simulated}.

Since only a finite number of betas and annealing times can be considered, the user has to provide a parameter denoted by \emph{num\_sweeps} (number of sweeps). As a result, the range for $\beta$ and the annealing time are discretized by \emph{num\_sweeps} different values that are used within the annealing process. In addition to that, a \emph{schedule} can be chosen to determine how $\beta$ or the annealing time can be changed during an annealing process --- either in a linear, or a geometric or a user-defined way. Since heuristic solution methods may not find a global minimizer after a single run, it might be required to run the solution method
for multiple times. In D-Wave, one can fix the number of runs by providing an appropriate value for  \emph{num\_reads} (number of reads). In terms of the solvers discussed above, a single run/read
is equivalent to a single annealing process.

%% file: theory.tex
This section's aim is to translate an input instance $(\queryOne,\queryTwo)$ of the query containment problem $\QCP$ into a triple $(p,\myd,C)$, where $p$ is an integer-valued polynomial on binary variables, $\myd$ is an integer, and $C$ is a subset of the domain of $p$, such that the following is true:
\begin{equation}\label{eq:theory-goal}
  \min_{\ov{x}\in C} p(\ov{x}) \ = \myd
  \ \ \iff \ \
  \queryOne\querycont \queryTwo.
\end{equation}
Notice that the minimization problem on the left-hand-side is similar to the one in equation \eqref{optimization_problem}; it can thus be solved using the methods described in Section~\ref{sec:preliminaries-quantum}. Concerning the right-hand-side, by Theorem~\ref{thm:ChandraMerlin}, we can replace ``$\queryOne\querycont\queryTwo$'' by ``there exists a homomorphism from $\queryTwo$ to $\queryOne$''.
Recall that such a homomorphism is a mapping $h\colon A_2\to A_1$, where $A_\ell=\Vars(\queryell)\cup\Adom(\queryell)$ for $\ell\in\{1,2\}$, that satisfies conditions~(1)--(3) of Definition~\ref{def:HomomorphismForQueries}.

In the following, we give a step-by-step description of how to construct the triple $(p,\myd,C)$ for an input instance $(\queryOne,\queryTwo)$ of $\QCP$.
Throughout this section, we let $\queryOne=(\TableauOne,\AnswerTupleOne)$ and
$\queryTwo=(\TableauTwo,\AnswerTupleTwo)$, and 
for $\ell\in\{1,2\}$ we let $r_\ell=\ar(\AnswerTupleell)$ and $\AnswerTupleell=(v_{\ell,1},\ldots,v_{\ell,r_\ell})$.

\subsection{Preparation Step}\label{sec:preparation}

The \emph{preparation step} allows us to recognize certain cases that trivially violate query containment. We also establish some notation that will help us in the construction of the polynomial $p$.

\myparagraph{Trivial violations of containment}
The conditions of Definition~\ref{def:HomomorphismForQueries} are obviously violated if (but not necessarily ``only if'')
\begin{enumerate}[1.]
	\item the answer tuples have different arities, i.e., $r_2\neq r_1$
          (then, condition \eqref{itemTwo:def:homom} %
          is violated), or
	\item a constant in the answer tuple of~$\queryTwo$ would not be mapped to itself in the answer tuple of~$\queryOne$, i.e., there is a $k\leq r_2$ such that $v_{2,k}\in\Domain$ and $v_{1,k}\neq v_{2,k}$ (then, satisfying condition \eqref{itemTwo:def:homom} will violate condition \eqref{itemOne:def:homom})%
    , or
	\item the same variable appears more than once in the answer tuple of~$\queryTwo$ but would be mapped to different elements in the answer tuple of~$\queryOne$, i.e., there are $k,k'\leq r_2$ such that $v_{2,k}=v_{2,k'}$ but $v_{1,k}\neq v_{1,k'}$ (then, condition \eqref{itemTwo:def:homom} %
    is violated), or
	\item the tableau of~$\queryTwo$ contains a non-empty relation that is empty in the tableau of~$\queryOne$, i.e., there is an $R\in\dbschema$ such that $\TableauTwo(R)\neq \emptyset=\TableauOne(R)$ (then, condition \eqref{itemThree:def:homom} %
    is violated).
\end{enumerate}

In each of these four cases, we can terminate and conclude that ``$\queryOne\not\querycont\queryTwo$''.
Otherwise, we proceed as described in the remainder of this section in order to construct the desired triple $(p,\myd,C)$.

\begin{exampleWithEndmarker}\label{example:RunningExamplePrepOne}
	Let $\queryExampleOne$, $\queryExampleTwo$ be the queries considered in the examples from Section~\ref{sec:prelims_qc}. Upon input of the query pair $(\queryExampleOne,\queryExampleTwo)$, none of the cases 1.--4.\ apply, and we therefore will proceed as described in the rest of this section.
        But when exchanging the roles of $\queryExampleOne$ and $\queryExampleTwo$,
the preparation step
	terminates and returns that $\queryExampleTwo\not\querycont\queryExampleOne$, because case~4.\ applies since $\TableauExampleOne(\sql{City})\neq\emptyset= \TableauExampleTwo(\sql{City})$.
\end{exampleWithEndmarker}

\myparagraph{Bit-matrix}
From now on, we assume that none of the cases 1.--4.\ apply. %
Our goal is to construct a polynomial $p$ on binary variables~$\ov{x}$, an integer $\myd$ and a set~$C$ such that Equivalence~\eqref{eq:theory-goal} holds.

The polynomial $p$ will use binary variables of the form $x_{i,j}$ for $(i,j)\in A_2\times A_1$, with
$A_2\deff\Vars(\queryTwo)\cup\Adom(\queryTwo)$ and $A_1\deff\Vars(\queryOne)\cup\Adom(\queryOne)$.
We view an \emph{assignment} of the variables $x_{i,j}$ with values in $\myset{0,1}$ as a \emph{bit-matrix} that can serve as a
representation of a function $h\colon A_2\to A_1$. 
Namely, $h$ is represented by the
\emph{bit-matrix} $\ov{x}\deff \big(x_{i,j})_{(i,j)\in A_2\times A_1}$
where, for all $i\in A_2$ and $j\in A_1$, we have $x_{i,j}\in\{0,1\}$ and $x_{i,j}=1$ $\iff$ $h(i)=j$.

\begin{exampleWithEndmarker}\label{example:bit-matrix}
  Consider queries $\queryExampleOne,\queryExampleTwo$ and homomorphism~$h$ from Example~\ref{example:QueryContainmentContinued}.
  We have
  $A_2=\myset{\vvarx_2,\vvary_2,\vvarz_2,\vvarw_2}$ and
  $A_1=\myset{\vvarx_1,\vvary_1,\vvarz_1,\allowbreak \sql{\textquotesingle actor\textquotesingle},\sql{\textquotesingle\LosAngeles\textquotesingle},\sql{\textquotesingle\UnitedStates\textquotesingle}}$.
  The bit-matrix representation of $h$ is:
\[
  \footnotesize
	 \setlength{\arraycolsep}{3pt}%
	 \renewcommand{\arraystretch}{0.95}%
	 \bordermatrix{
		 & \vvarx_1 & \vvary_1 & \vvarz_1 & \sql{\textquotesingle actor\textquotesingle} & \sql{\textquotesingle\LosAngeles\textquotesingle} & \sql{\textquotesingle\UnitedStates\textquotesingle}   \cr
		 \vvarx_2   & 1   & 0   & 0   & 0   & 0   & 0    \cr
		 \vvary_2   & 0   & 1   & 0   & 0   & 0   & 0    \cr
		 \vvarz_2   & 0   & 0   & 1   & 0   & 0   & 0    \cr
		 \vvarw_2   & 0   & 0   & 0   & 1   & 0   & 0    \cr
		 }
               \]
               E.g., the entry in row $\vvarx_2$ and column $\vvarx_1$ is addressed as the bit $x_{i,j}$ with $i=\vvarx_2$ and $j=\vvarx_1$;
              in the above matrix, this bit is set to 1.
 \end{exampleWithEndmarker}

Note that a bit-matrix $\ov{x}\in\{0,1\}^{A_2\times A_1}$ represents a mapping from $A_2$ to $A_1$ if, and only if, for every $i\in A_2$ there is exactly one $j\in A_1$ with $x_{i,j}=1$ (i.e., exactly one bit per matrix row is 1).

Next, we make use of the fact that we already know that the query pair $(\queryOne,\queryTwo)$ satisfies none of the cases 1.--4. %
We can restrict our  attention to mappings $h\colon A_2\to A_1$ where
$h(v_{2,k})=v_{1,k}$ for every $k\leq r_2$.
Furthermore, condition \eqref{itemOne:def:homom} in Definition~\ref{def:HomomorphismForQueries} demands that
$h(\cconstc)=\cconstc$ for all $\cconstc\in \Adom(\queryTwo)$.
Thus, the values of $h$ are predefined on all elements in $\Adom(\queryTwo)\cup\Vars(\AnswerTupleTwo)$. Therefore, instead of $A_2$, we only consider $B_2\deff A_2\setminus(\Adom(\queryTwo)\cup\Vars(\AnswerTupleTwo))$, and
we extend any bit-matrix $\ov{x}\in\{0,1\}^{B_2\times A_1}$ to a bit-matrix $\ovxext\in\{0,1\}^{A_2\times A_1}$ by letting $\xext{i,j}=x_{i,j}$ for all $(i,j)\in B_2\times A_1$ and
\begin{enumerate}[(a)]
	\item\label{item:aForExtendedBitMatrix} for all $\cconstc\in\Adom(\queryTwo)$ and $j\in A_1$ we let
	$\xext{\cconstc,j}= 1$ $\Longleftrightarrow$ $j=\cconstc$,
	
	\item\label{item:bForExtendedBitMatrix} for all $i\in\Vars(\AnswerTupleTwo)$ and $j\in A_1$ we let \\
	$\xext{i,j}=1$ $\Longleftrightarrow$ there is a $k\leq r_2$ with $i=v_{2,k}$ and $j=v_{1,k}$.
\end{enumerate}
This ensures that if $\ovxext$ represents some mapping $h\colon
A_2\to A_1$, then this mapping satisfies the conditions \eqref{itemOne:def:homom} and \eqref{itemTwo:def:homom}
of Definition~\ref{def:HomomorphismForQueries}.

We use the same notation when considering the \emph{binary variables} $\ov{x}=(x_{i,j})_{(i,j)\in B_2\times A_1}$ and write $\ovxext$ for the matrix obtained by extending the variable-matrix $\ov{x}$ with additional rows $i$ for all $i\in \Adom(\queryTwo)\cup\Vars(\AnswerTupleTwo)$; each additional row  is already filled with elements from $\myset{0,1}$,
defined %
as in the above items~\eqref{item:aForExtendedBitMatrix} and~\eqref{item:bForExtendedBitMatrix}.

\begin{exampleWithEndmarker}\label{example:preparation:sqlqueries}
Recall from Example~\ref{example:bit-matrix} the
particular sets $A_2$ and $A_1$. Note that $\Adom(\queryExampleTwo)=\emptyset$ and $\Vars(\AnswerTupleExampleTwo)=\myset{\vvary_2}$. Hence, $B_2=A_2\setminus\myset{\vvary_2}=\myset{\vvarx_2,\vvarz_2,\vvarw_2}$.
We will use binary variables $x_{i,j}$ for $(i,j)\in B_2\times A_1$. I.e., we consider the variable-matrix
$\ov{x}= (x_{i,j})_{(i,j)\in B_2\times A_1}$.
The extension $\ovxext$ of $\ov{x}$ is obtained by adding a row for $\vvary_2$ (since $\Adom(\queryExampleTwo)\cup\Vars(\queryExampleTwo)=\myset{\vvary_2}$). This yields the matrix $\ovxext =$
\[
 \footnotesize
\setlength{\arraycolsep}{3pt}%
\renewcommand{\arraystretch}{0.95}%
\bordermatrix{
	& \vvarx_1 & \vvary_1 & \vvarz_1 & \sql{\textquotesingle actor\textquotesingle} & \sql{\textquotesingle \LosAngeles\textquotesingle} & \sql{\textquotesingle \UnitedStates\textquotesingle}   \cr
    \noalign{\vskip -0.9ex} %
	\vvarx_2   & x_{\vvarx_2,\vvarx_1}   & x_{\vvarx_2,\vvary_1}   & x_{\vvarx_2,\vvarz_1}   & x_{\vvarx_2,\sql{\textquotesingle actor\textquotesingle}}   & x_{\vvarx_2,\sql{\textquotesingle \LosAngeles\textquotesingle}}   & x_{\vvarx_2,\sql{\textquotesingle \UnitedStates\textquotesingle}}    \cr
	\vvary_2   & 0   & 1   & 0   & 0   & 0   & 0    \cr
	\vvarz_2   & x_{\vvarz_2,\vvarx_1}   & x_{\vvarz_2,\vvary_1}   & x_{\vvarz_2,\vvarz_1}   & x_{\vvarz_2,\sql{\textquotesingle actor\textquotesingle}}   & x_{\vvarz_2,\sql{\textquotesingle \LosAngeles\textquotesingle}}   & x_{\vvarz_2,\sql{\textquotesingle \UnitedStates\textquotesingle}}    \cr
	\vvarw_2   & x_{\vvarw_2,\vvarx_1}   & x_{\vvarw_2,\vvary_1}   & x_{\vvarw_2,\vvarz_1}   & x_{\vvarw_2,\sql{\textquotesingle actor\textquotesingle}}   & x_{\vvarw_2,\sql{\textquotesingle \LosAngeles\textquotesingle}}   & x_{\vvarw_2,\sql{\textquotesingle \UnitedStates\textquotesingle}}    \cr
}
\]
This reflects that the answer tuples of $\queryExampleTwo$ and $\queryExampleOne$ are $(\vvary_2)$ and $(\vvary_1)$, and any homomorphism from $\queryExampleTwo$ to $\queryExampleOne$ must map $\vvary_2$ to $\vvary_1$.
\end{exampleWithEndmarker}

As a result of the preparation step upon input of an arbitrary query pair $(\queryOne,\queryTwo)$, we either output that $\queryOne\not\querycont\queryTwo$ %
or we output the sets $A_2$, $B_2$, $A_1$ and the variable-matrix $\ov{x}$ and its extension $\ovxext$. In the next subsection, we explain how these are used to define a
triple $(p,\myd,C)$ that satisfies Equivalence~\eqref{eq:theory-goal}. Polynomial~$p$ will use the binary variables $\ov{x}$ and will also utilize the additional 0-1-entries in the extended matrix $\ovxext$.
Yet first, let us consider another example.

\begin{exampleWithEndmarker}\label{example:preparation:cycle_chain}
We consider directed graphs
in a relational encoding.
Accordingly, let $\dbschema_{g}$ be the schema consisting of a binary
relation name $E$, for representing the edges of a graph.

Consider the Boolean queries \emph{``2-cycle''} $q_{\textit{2cy}}\deff
(\Tableau_{\textit{2cy}},\AnswerTuple_{\textit{2cy}})$ and \emph{``2-chain''} $q_{\textit{2ch}}\deff
(\Tableau_{\textit{2ch}},\AnswerTuple_{\textit{2ch}})$ defined via
$\AnswerTuple_{\textit{2cy}}=\AnswerTuple_{\textit{2ch}}=\emptytuple$ and
$\Tableau_{\textit{2cy}}(E)=\myset{(\vvarz,\vvarz'),\allowbreak (\vvarz',\vvarz)}$
and
$\Tableau_{\textit{2ch}}(E)=\myset{(\vvarz_0,\vvarz_1),  \allowbreak (\vvarz_1,\vvarz_2)}$.
An $\dbschema_{g}$-db $D$ represents a directed graph, and
$q_{\textit{2cy}}(D)=$ ``yes'' iff this graph contains a directed cycle on at
most 2 nodes, while
$q_{\textit{2ch}}(D)=$ ``yes'' iff this graph contains a directed path of length~2.
Upon input of the query pair $(q_{\textit{2cy}},q_{\textit{2ch}})$, the preparation step
yields
$A_2=B_2=\myset{\vvarz_0,\vvarz_1,\vvarz_2}$,
$A_1=\myset{\vvarz,\vvarz'}$, and
$\ov{x}$ consists of the binary variables $x_{i,j}$ for all $(i,j)\in B_2\times A_1$.
Since $A_2=B_2$, we have $\ovxext=\ov{x}$, i.e., \emph{no} additional rows with 0-1-entries are
inserted into
$\ov{x}$.
\end{exampleWithEndmarker}

 \subsection{Generic choice of $(p,\myd,C)$}\label{sec:genericchoice}

In our generic choice of $(p,\myd,C)$ we let $C = \{0,1\}^{B_2\times A_1}$.
Our goal is to define $\myd$ and $p$ in such a way that Equivalence~\eqref{eq:theory-goal} holds.
Here, $p$ will be an integer-valued polynomial on the binary variables $\ov{x}=(x_{i,j})_{(i,j)\in B_2\times A_1}$; we will define $p$ in such a way, that for an  assignment $\ov{x}\in\myset{0,1}^{B_2\times A_1}$ of the binary variables with actual bits, the following is true: the value
$p(\ov{x})$ is
minimal and equal to $\myd$
$\iff$ the extended bit-matrix $\ovxext$ represents a function $h\colon A_2\to A_1$ that, in fact, is a homomorphism from $\queryTwo$ to $\queryOne$.

For defining $p$, we will
use
a number of intermediate polynomials that are introduced in the following.
The first of these %
is
\[ \textstyle
  \punique(\ov{x}) \ \ \deff \ \
  \sum_{i\in B_2} \sum_{j,j' \in A_1\atop j<j'} \big(\, x_{i,j} \cdot x_{i,j'}\, \big).
\]
Here, $<$ denotes an arbitrary linear order on $A_1$ (which corresponds to the arrangement of the columns in the variable-matrix $\ov{x}$).
When assigning the binary variables in $\ov{x}$ to actual bits in $\myset{0,1}$, this polynomial evaluates to an integer $\geq 0$; and it evaluates to exactly~$0$ if, and only if, in every row $i\in B_2$ there is at most one column~$j\in A_1$ such that the variable $x_{i,j}$ is assigned the value~1.
Hence, if it evaluates to a value $>0$, then the assigned bit-matrix (and its extension) for sure does \emph{not} represent any function $h\colon A_2\to A_1$, let alone an actual homomorphism from $\queryTwo$ to $\queryOne$.

\begin{exampleWithEndmarker}\label{example:unique:cycle_chain}
Continuing Example~\ref{example:preparation:cycle_chain}, we have
$\punique(\ov{x}) = $
\\
\(
  \big(\,x_{\vvarz_0,\vvarz} \cdot x_{\vvarz_0,\vvarz'}\,\big)
  \ + \
  \big(\,x_{\vvarz_1,\vvarz} \cdot x_{\vvarz_1,\vvarz'}\,\big)
  \ + \
  \big(\,x_{\vvarz_2,\vvarz} \cdot x_{\vvarz_2,\vvarz'}\,\big).
\)  
\end{exampleWithEndmarker}

Next, we introduce several polynomials that will help us handle condition~\eqref{itemThree:def:homom} of Definition~\ref{def:HomomorphismForQueries}.
Consider a relation
$R\in\dbschema$ and a  tuple $u\in \TableauTwo(R)$. Let $r=\ar(u)$ and $(u_1,\ldots,u_r)=u$, and note that $\myset{u_1,\ldots,u_r}\subseteq A_2=B_2\cup \Adom(\queryTwo)\cup\Vars(\AnswerTupleTwo)$.
Recall that $\ov{x}$ is a variable-matrix that contains a row $i$ for every $i\in B_2$; and $\ovxext$ extends $\ov{x}$ by additional rows for every $i\in\Adom(\queryTwo)\cup\Vars(\AnswerTupleTwo)$, and these rows are already filled with 0-1-entries.
We define
\[
  p_{R,u}(\ov{x}) \ \deff \
  \sum_{(w_1,\ldots,w_r)\in \TableauOne(R)} \left(
    - \prod_{k=1}^r \xext{u_{k},w_{k}}
  \right).
\]  
Since $\xext{u_k,w_k}$ will be assigned with values in $\myset{0,1}$, the product will evaluate to either 0 or 1.
By a close inspection of the polynomial, we can prove (cf.~Appendix~\ref{appendix:theory}):

\begin{lemma}\label{lemma:p_v_new}
	For all $R\in\schema$,  $u\in \TableauTwo(R)$,  $\ov{x}\in\myset{0,1}^{B_2\times A_1}$ we have
	\begin{enumerate}[(a)]
        \item\label{item:a:lemma:p_v_new}
 $p_{R,u}(\ov{x})$ is an integer with $-\setsize{\TableauOne(R)}\leq p_{R,u}(\ov{x}) \leq 0$,
          and
		\item\label{item:b:lemma:p_v_new} if $\punique(\ov{x})=0$ then
		$p_{R,u}(\ov{x})\in\{-1,0\}$ and, moreover,\\
		$p_{R,u}(\ov{x})=-1$ $\Longleftrightarrow$ for the extended bit-matrix $\ovxext$ associated with $\ov{x}$ and letting $r=\ar(R)$ and $(u_1,\ldots,u_r)=u$, for every $k\leq r$ we have
                $\sum_{z\in A_1} \xext{u_k,z}=1$, and letting $z_k$ be the particular $z\in A_1$ with $\xext{u_k,z}=1$ we have
                $(z_1,\ldots,z_r)\in\TableauOne(R)$. 
	\end{enumerate}
\end{lemma}

Part \eqref{item:b:lemma:p_v_new} of
this lemma tells us that if the binary variables are assigned with 0-1-values such that at most one bit per matrix row is~1, then the following is true:
$p_{R,u}(\ov{x})$ evaluates to $-1$ $\iff$ for every entry $u_k$ of $u$, the row $u_k$ of $\ovxext$ contains exactly one~1, and for the particular column $z_k$ that carries this 1, the tuple $(z_1,\ldots,z_r)$ belongs to the tableau $\TableauOne(R)$. The latter can serve as a witness that condition~\eqref{itemThree:def:homom} of Definition~\ref{def:HomomorphismForQueries} is satisfied.
Condition~\eqref{itemThree:def:homom} demands that this holds for all relations $R\in\dbschema$ and all $u\in\TableauTwo(R)$; and accordingly, we define the polynomial
\[ \textstyle
  \pac(\ov{x}) \ \ \deff \ \
  \sum_{R\in\dbschema} \ \sum_{u\in \TableauTwo(R)} \ p_{R,u}(\ov{x}).
\]

\begin{exampleWithEndmarker}\label{example:ac:cycle_chain}
Consider the query pair $(p_{\textit{2cy}},p_{\textit{2ch}})$ from
Examples~\ref{example:preparation:cycle_chain} and~\ref{example:unique:cycle_chain}.
The tuple $(\vvarz_0,\vvarz_1)\in \Tableau_{\textit{2ch}}(E)$ yields
$  p_{E,(\vvarz_0,\vvarz_1)}(\ov{x}) = $
\\
\(
  - \ x_{\vvarz_0,\vvarz}\cdot x_{\vvarz_1,\vvarz'} \ - \
  x_{\vvarz_0,\vvarz'}\cdot x_{\vvarz_1,\vvarz}\,; \
\)
the tuple $(\vvarz_1,\vvarz_2)\in \Tableau_{\textit{2ch}}(E)$ yields
\\
\(
  p_{E,(\vvarz_1,\vvarz_2)}(\ov{x}) \  =  \
  - \; x_{\vvarz_1,\vvarz}\cdot x_{\vvarz_2,\vvarz'} \; - \;
  x_{\vvarz_1,\vvarz'}\cdot x_{\vvarz_2,\vvarz}\,;
  \)
\ and in summary we have \
$\pac(\ov{x}) \ = \ p_{E,(\vvarz_0,\vvarz_1)}(\ov{x}) \ + \
p_{E,(\vvarz_1,\vvarz_2)}(\ov{x}) \ =$
\\
\(
  - \ x_{\vvarz_0,\vvarz}\cdot x_{\vvarz_1,\vvarz'} \ - \
  x_{\vvarz_0,\vvarz'}\cdot x_{\vvarz_1,\vvarz}
  - \ x_{\vvarz_1,\vvarz}\cdot x_{\vvarz_2,\vvarz'} \ - \
  x_{\vvarz_1,\vvarz'}\cdot x_{\vvarz_2,\vvarz}\,.
\) 
\end{exampleWithEndmarker}  

For each $i\in\myset{1,2}$ we let $|\Tableaui|\deff\sum_{R\in\dbschema}|\Tableaui(R)|$.
Using Lemma~\ref{lemma:p_v_new}, we can prove the following  (cf.~Appendix~\ref{appendix:theory}):

\begin{lemma}\label{lemma:p_ac_new}
For every $\ov{x}\in\myset{0,1}^{B_2\times A_1}$ we have:
\begin{enumerate}[(a)]
\item\label{item:a:lemma:p_ac_new} 
  $\pac(\ov{x})$ is an integer with $-\sum_{R\in\dbschema}\setsize{\TableauOne(R)}{\cdot} \setsize{\TableauTwo(R)}\leq \pac(\ov{x})\leq 0$.
\item\label{item:b:lemma:p_ac_new} If $\punique(\ov{x})=0$ then $-\setsize{\TableauTwo}\leq \pac(\ov{x}) \leq 0$ and, moreover,\\
  $\pac(\ov{x})=-\setsize{\TableauTwo}$ $\Longleftrightarrow$
  the extended bit-matrix $\ovxext$ associated with $\ov{x}$ satisfies the following: for every $i\in A_2$ we have
  $\sum_{j\in A_1}\xext{i,j}=1$, i.e., $\ovxext$ represents a mapping $h\colon A_2\to A_1$ via $h(i)=j$ iff $\xext{i,j}=1$, and
  this mapping is a homomorphism from $\queryTwo$ to $\queryOne$.
\end{enumerate}
\end{lemma}

Part \eqref{item:b:lemma:p_ac_new}
tells us that if the binary variables are assigned with 0-1-values such that at most one bit per matrix row is 1, then
we have:
$\pac(\ov{x})$ evaluates to $-\setsize{\TableauTwo}$ $\iff$ the  extended bit-matrix $\ovxext$ associated with $\ov{x}$ represents a mapping
$h\colon A_2\to A_1$ that, in fact, is a homomorphism from $\queryTwo$ to $\queryOne$.
This
leads us to define
\[
  \pgen(\ov{x}) \ \ \deff \ \
  \pac(\ov{x}) \ + \ (\setsize{\TableauOne}{\cdot}\setsize{\TableauTwo}+1)\cdot \punique(\ov{x}).
\]
If $\ov{x}$ is an assignment with 0-1-values such that $\punique(\ov{x})=0$, then we obtain from Lemma~\ref{lemma:p_ac_new} that
$\pgen(\ov{x})$ is an integer between $-\setsize{\TableauTwo}$ and $0$; and it is exactly $-\setsize{\TableauTwo}$ iff the extended bit-matrix $\ovxext$ represents a homomorphism from $\queryTwo$ to $\queryOne$ (witnessing that $\queryOne\querycont\queryTwo$).
On the other hand, for any assignment $\ov{x}$ with $\punique(\ov{x})\neq 0$ we have $\punique(\ov{x})\geq 1$, and thus $\pgen(\ov{x})\geq 0$. Actually, the same still holds if in the definition of $\pgen(\ov{x})$ we replace the term
``$\setsize{\TableauOne}{\cdot}\setsize{\TableauTwo}$'' with any number
$\geq \sum_{R\in\dbschema}\setsize{\TableauOne(R)}\cdot\setsize{\TableauTwo(R)}$.

In our generic choice for $(p,\myd,C)$, we let $p(\ov{x})\deff \pgen(\ov{x})$,
$\myd\deff\dgen= -|\TableauTwo|$ and $C\deff \Cgen=\myset{0,1}^{B_2\times A_1}$.
The following theorem states that this choice is correct;
the theorem's proof (cf.\ Appendix~\ref{appendix:theory}) relies on
Lemma~\ref{lemma:p_ac_new}
and
Theorem~\ref{thm:ChandraMerlin}.
\begin{theorem}\label{thm:correctnessOfGenericChoice}
  Let $(p,\myd,C)\deff (\pgen,\dgen,\Cgen)$.\\
  For every $\ov{x}\in C$ we have $p(\ov{x})\geq \myd$. 
  Furthermore, \ 
  $\min_{\ov{x}\in C} p(\ov{x})=\myd$ $\iff$ $\queryOne\querycont\queryTwo$. 
  Moreover, for every $\ov{x}\in C$ with $p(\ov{x})=\myd$, the extended bit-matrix $\ovxext$ represents a homomorphism from $\queryTwo$ to $\queryOne$.
\end{theorem}

This theorem
tells us that in order to
decide whether $\queryOne\querycont\queryTwo$, we can seek for a
bit-matrix $\ov{x}^*\in \Cgen$ such that $\pgen(\ov{x}^*)$ is as
small as possible. If $\pgen(\ov{x}^*)=\dgen$,
then we know that $\queryOne\querycont\queryTwo$
and, moreover, $\ov{x}^*$ represents a homomorphism from $\queryTwo$
to $\queryOne$ and can thus serve as a certificate which proves that
$\queryOne\querycont\queryTwo$. Otherwise, 
$\pgen(\ov{x}^*)>\dgen$, and we know that $\queryOne\not\querycont\queryTwo$.

In Sections~\ref{sec:system} \& \ref{sec:experiments} we present a
quantum computing solution for finding $\ov{x}^*$.
The size of the {search space}
is $\setsize{\Cgen}=2^{n_2\cdot n_1}$ for $n_2\deff\setsize{B_2}$ and $n_1\deff\setsize{A_1}$.
The number of
qubits (and thus, the computational resources) required for our
solution is closely related to the number of {binary variables} present in the polynomial $\pgen(\ov{x})$. The number of binary variables
occurring in $\pgen(\ov{x})$ is $n_2{\cdot} n_1 =
|B_2|{\cdot} |A_1|$.
The next two subsections present optional modifications that may
decrease the number of binary variables and the size of the search space.

\begin{exampleWithEndmarker}\label{example:gen:cycle_chain}
Consider the query pair $(p_{\textit{2cy}},p_{\textit{2ch}})$ from
Example~\ref{example:preparation:cycle_chain} and
recall the
polynomials $\punique(\ov{x})$ and $\pac(\ov{x})$ constructed in Examples~\ref{example:unique:cycle_chain} \& \ref{example:ac:cycle_chain}.
Then, $|\Tableau_{\textit{2cy}}|{\cdot}|\Tableau_{\textit{2ch}}|+1=2{\cdot} 2 + 1 = 5$ and
$\pgen(\ov{x}) = \pac(\ov{x})+ 5\cdot \punique(\ov{x})$.
Furthermore, $\dgen=-\setsize{\Tableau_{\textit{2ch}}}= -2$ and
$\Cgen=\myset{0,1}^{B_2\times A_1}$ for $B_2=\myset{\vvarz_0,\vvarz_1,\vvarz_2}$ and
$A_1=\myset{\vvarz,\vvarz'}$.
\end{exampleWithEndmarker}

\subsection{Simplification Step}\label{sec:simplification}

This section presents an optional \emph{simplification step} that
closely inspects the polynomial $\pac(\ov{x})$ with the aim of restricting
$B_2$ to a subset $\simplBtwo$, so that the total number of
binary variables is reduced from $n_2{\cdot} n_1$ to $\simplntwo{\cdot}
n_1$ for $\simplntwo=\setsize{\simplBtwo}$.
In the following, we describe a process that starts with
$\simplBtwo\deff B_2$ and $\ovxsim\deff \ovxext$, and that iteratively removes
elements from $\simplBtwo$ and adapts $\ovxsim$ accordingly.

The idea is that for
tuples
in $\queryTwo$ that have only one
tuple
in $\queryOne$ to which they can fit, we can already map the variables of that
tuple
accordingly. Moreover, if no
tuple
in $\queryOne$ fits (e.g. due to fixed mappings already found during the preparation step), then we can already conclude that $\queryOne\not\querycont\queryTwo$. Following this idea, we iteratively check all
tuples
in $\queryTwo$ for potential simplifications.

By definition,
$\pac(\ov{x})=\sum_{R\in\dbschema}\sum_{u\in
  \TableauTwo(R)} p_{R,u}(\ov{x})$, for
\[
  p_{R,u}(\ov{x})\ = \
  \sum_{(w_1,\ldots,w_r)\in \TableauOne(R)}
  \left(
     - \prod_{k=1}^{r} \xext{u_k,w_k}
  \right),
\]
where $(u_1,\ldots,u_r)=u$.
Here, $\xext{u_k,w_k}$ is the binary variable $x_{u_k,w_k}$ if $u_k\in B_2$, and
it is a fixed 0-1-value if $u_k\in A_2\setminus B_2$.
In particular, if there is a $k$ with
$u_k\in A_2\setminus B_2$ and $\xext{u_k,w_k}=0$, then the entire product
$\prod_{k=1}^{r}\xext{u_k,w_k}$ evaluates to $0$,
and hence
the tuple $(w_1,\ldots,w_r)\in\TableauOne(R)$ has zero
contribution in the polynomial $p_{R,u}(\ov{x})$.

If \emph{every} $w\in\TableauOne(R)$ has zero contribution, %
$p_{R,u}(\ov{x})$ is the constant polynomial $0$.
From Lemma~\ref{lemma:p_ac_new} we
obtain that there does not exist any homomorphism from $\queryTwo$ to $\queryOne$;
we 
stop and
output~``$\queryOne\not\querycont\queryTwo$''.

If there is exactly one $w=(w_1,\ldots,w_r)\in\TableauOne(R)$ that
does \emph{not} have zero contribution, then
$p_{R,u}(\ov{x}) = - \prod_{k=1}^{r}\xext{u_k,w_k}$.
From  Lemma~\ref{lemma:p_ac_new} we then obtain that \emph{if} there
exists a homomorphism~$h$ from $\queryTwo$ to $\queryOne$, then it must
satisfy $h(u_k)=w_k$ for all $k\leq r$.
We can therefore treat the tuples $u$ and $w$ in a similar way as we had
treated the queries' answer tuples $\AnswerTupleTwo$ and
$\AnswerTupleOne$ in the preparation step:
\\
If there are $k,k'\leq r$
such that $u_k=u_{k'}$ but $w_k\neq w_{k'}$, then we know that there
does not exist any homomorphism from $\queryTwo$ to $\queryOne$, and
hence we can stop and
output ``$\queryOne\not\querycont\queryTwo$''.
Similarly, if there is a $k\leq r$ such that $u_k\not\in \simplBtwo$
and $\xsim{u_k,w_k}=0$, we know that some previously considered
constraint had demanded that $h(u_k)\neq w_k$ while the currently
considered constraint demands that $h(u_k)=w_k$; hence, we can stop and
output ``$\queryOne\not\querycont\queryTwo$''.
Otherwise, we remove $u_1,\ldots,u_r$ from $\simplBtwo$ and modify the
definition of~$\ovxsim$ by setting $\xsim{u_k,w_k}\deff 1$ and
$\xsim{u_k,j}\deff 0$ for all $j\in A_1\setminus\myset{w_k}$ and all $k\leq r$.

We loop through all $R\in\dbschema$ and all
$u\in\TableauTwo(R)$ and apply this procedure.
We use the resulting matrix $\ovxsim$ and the set $\simplBtwo$ and let $\simplntwo\deff
\setsize{\simplBtwo}$. Our \emph{simplified choice} of $(p,\myd,C)$,
which we denote as $(\psimpl,\dsimpl,\Csimpl)$, is
defined via $\Csimpl=\myset{0,1}^{\simplBtwo\times A_1}$, \
$\dsimpl=-\setsize{\TableauTwo}$,  and 
$\psimpl(\ov{x})$ is the polynomial on binary variables $x_{i,j}$ for $(i,j)\in \simplBtwo\times A_1$ with
$\psimpl(\ov{x})  \deff
  \pacsimpl(\ov{x})  + 
  (\setsize{\TableauOne}{\cdot}\setsize{\TableauTwo}+1)\cdot
  \puniquesimpl(\ov{x})$,
where
$\pacsimpl(\ov{x})$ is defined in the same way as $\pac(\ov{x})$, but $\xext{u_k,w_k}$ is replaced by $\xsim{u_k,w_k}$,
and $\puniquesimpl(\ov{x})$ is defined as $\punique(\ov{x})$, but the leftmost sum runs only over $i\in \simplBtwo$.  
Using Theorem~\ref{thm:correctnessOfGenericChoice}, the %
construction of $(\psimpl,\dsimpl,\Csimpl)$ immediately yields:

\begin{theorem}\label{thm:correctnessOfSimplifiedChoice}
The statement of Theorem~\ref{thm:correctnessOfGenericChoice} holds
for $(p,\myd,C)\deff (\psimpl,\dsimpl,\Csimpl)$ when replacing $\ovxext$
with $\ovxsim$.
\end{theorem}

The number of binary variables occurring in the polynomial $\psimpl(\ov{x})$ is
$\simplntwo{\cdot} n_1$, where $\simplntwo\deff\setsize{\simplBtwo}\leq n_2$.

\begin{exampleWithEndmarker}\label{example:simpl:cycle_chain}
When considering the query pair $(p_{\textit{2cy}},p_{\textit{2ch}})$ from
Examples~\ref{example:preparation:cycle_chain},
\ref{example:unique:cycle_chain}, \ref{example:ac:cycle_chain},
\ref{example:gen:cycle_chain},
we have $A_2\setminus B_2=\emptyset$ and
$|\Tableau_{\textit{2cy}}(E)|>1$,
and hence the simplification
step does not change anything and we obtain $\simplBtwo=B_2$ and $(\psimpl,\dsimpl,\Csimpl)=(\pgen,\dgen,\Cgen)$.
\end{exampleWithEndmarker}

\begin{exampleWithEndmarker}\label{example:simpl:sql}
When considering the query pair $(\queryExampleOne,\queryExampleTwo)$
from Examples~\ref{example:queriesSQLtoCQcontinued},
\ref{example:QueryContainment}, \ref{example:preparation:sqlqueries}, 
recall that  $B_2=\myset{\vvarx_2,\vvarz_2,\vvarw_2}$ and $A_2\setminus
B_2=\myset{\vvary_2}$ and
$\xext{\vvary_2,\vvary_1}=1$ and $\xext{\vvary_2,j}=0$ for every $j\in
A_1\setminus\myset{\vvary_2}$.
We have
\begin{eqnarray*}
  \pac(\ov{x})
& =
& p_{\sql{Person},(\vvarx_2,\vvary_2,\vvarz_2)}(\ov{x}) \ + \     
p_{\sql{Profession},(\vvarx_2,\vvarw_2)}(\ov{x})
  \\
& =  
& - \ \xext{\vvarx_2,\vvarx_1}\cdot \xext{\vvary_2,\vvary_1}\cdot \xext{\vvarz_2,\vvarz_1}
\ - \ \xext{\vvarx_2,\vvarx_1}\cdot \xext{\vvarw_2,\sql{\textquotesingle actor\textquotesingle}}
\\
  & =
& - \ x_{\vvarx_2,\vvarx_1}\cdot  1 \cdot x_{\vvarz_2,\vvarz_1}
\ - \ x_{\vvarx_2,\vvarx_1}\cdot x_{\vvarw_2,\sql{\textquotesingle actor\textquotesingle}}\,.
\end{eqnarray*}
During simplification, we initialize $\simplBtwo$ to $B_2$ and
then remove from~$\simplBtwo$ the variables $\vvarx_2,\vvarz_2$, and let
$\xsim{\vvarx_2,\vvarx_1}=1$, $\xsim{\vvarz_2,\vvarz_1}=1$, $\xsim{\vvarx_2,j}=0$ for all
$j\in A_1\setminus\myset{\vvarx_1}$ and $\xsim{\vvarz_2,j}=0$ for all $j\in
A_1\setminus\myset{\vvarz_1}$.
Afterwards,
we remove $\vvarw_2$ from $\simplBtwo$ and let
$\xsim{\vvarw_2,\sql{\textquotesingle actor\textquotesingle}}=1$ and $\xsim{\vvarw_2,j}=0$ for all $j\in A_1\setminus\myset{\sql{\textquotesingle actor\textquotesingle}}$.
Then,
$\simplBtwo=\emptyset$ and
$\pacsimpl(\ov{x}) = 
  -  \xsim{\vvarx_2,\vvarx_1}{\cdot} \xsim{\vvarz_2,\vvarz_1}
   -  \xsim{\vvarx_2,\vvarx_1}{\cdot} \xsim{\vvarw_2,\sql{\textquotesingle actor\textquotesingle}}
  = -2$ 
  and $\puniquesimpl(\ov{x})=0$. \
Thus,
we obtain: \
$\psimpl(\ov{x})=-2$, \
$\dsimpl=-\setsize{\TableauExampleTwo}=-2$,  and
$\Csimpl=\myset{0,1}^{\emptyset\times A_1}$ which, by convention, is
the set~$\myset{\emptytuple}$ consisting only of the empty tuple $\emptytuple$. We leverage this in Example~\ref{example:trivial:sql}.
\end{exampleWithEndmarker}  

\subsection{Using Constraints}\label{sec:contraints}

A further optional step is to \emph{use constraints}, i.e., to
restrict the set~$C$ of considered bit-matrices $\ov{x}$ to
those bit-matrices that actually represent functions (i.e.\ where there is exactly one 1 per row). Such constraints can be imposed using an adapted QAOA version as explained in Section~\ref{sec:QAOA}.
This
drastically
reduces the search space, and 
additionally, it allows
to simplify the polynomial $p(\ov{x})$ by completely dropping the
term involving $\punique(\ov{x})$.
Accordingly, for our choices $(\pgen,\dgen,\Cgen)$ and
$(\psimpl,\dsimpl,\Csimpl)$, we get new \emph{constrained} versions
$(\pcgen,\dcgen,\Ccgen)$ and
$(\pcsimpl,\dcsimpl,\Ccsimpl)$, which are defined as follows: \
$\pcgen(\ov{x})\deff \pac(\ov{x})$, \
$\pcsimpl(\ov{x}) \deff  \pacsimpl(\ov{x})$, \
$\dcgen\deff\dgen=-\setsize{\TableauTwo}$, \
$\dcsimpl\deff\dsimpl=-\setsize{\TableauTwo}$, and
$\Ccgen$ is the set of all bit-matrices
$\ov{x}\in\myset{0,1}^{B_2\times A_1}$ where for all $i\in B_2$ there
is exactly one $j\in A_1$ with $x_{i,j}=1$, while
$\Ccsimpl$
is defined analogously, by replacing $B_2$ with $\simplBtwo$.
Note that
$\setsize{\Ccgen} \ = \ (n_1)^{n_2} \ \leq \ \setsize{\Cgen} =
2^{n_2\cdot n_1}$
and  
$\setsize{\Ccsimpl} \ = \ (n_1)^{\simplntwo} \ \leq \ \setsize{\Csimpl} =
2^{\simplntwo\cdot n_1}$.
Theorems~\ref{thm:correctnessOfGenericChoice} and
\ref{thm:correctnessOfSimplifiedChoice} yield:

\begin{corollary}\label{cor:correctnessOfConstrainedChoices} \ \\
The statement of Theorem~\ref{thm:correctnessOfGenericChoice} holds
for $(p,\myd,C)\deff (\pcgen,\dcgen,\Ccgen)$. \\
The statement of
Theorem~\ref{thm:correctnessOfSimplifiedChoice} holds for
$(p,\myd,C)\deff (\pcsimpl,\dcsimpl,\Ccsimpl)$.
\end{corollary}  

\subsection{Identifying Trivial Cases}\label{sec:trivialcases}

Let $(p,\myd,C)$ be any of the choices $(\pgen,\dgen,\Cgen)$,
$(\pcgen,\dcgen,\Ccgen)$,
$(\psimpl,\dsimpl,\Csimpl)$,
$(\pcsimpl,\dcsimpl,\Ccsimpl)$.
If the polynomial $p$ actually does not contain any binary variable
(e.g.\ if all variables have been assigned during preparation and/or simplification;
see Appendix~\ref{appendix:theory} for details), 
then $p(\ov{x})$ is the constant function $\ell$, for some integer
$\ell$ (for short: $p(\ov{x})\equiv \ell$).
In case $p(\ov{x})\equiv\ell$ for an integer $\ell$, we
obtain from Theorems~\ref{thm:correctnessOfGenericChoice}
and~\ref{thm:correctnessOfSimplifiedChoice},
and Corollary~\ref{cor:correctnessOfConstrainedChoices}, that
$\queryOne\not\querycont\queryTwo$ if $\ell\neq d$, and
$\queryOne\querycont\queryTwo$ if $\ell=d$; in the latter case we
can construct a homomorphism from $\queryTwo$ to $\queryOne$ as a certificate proving   $\queryOne\querycont\queryTwo$.

\begin{exampleWithEndmarker}\label{example:trivial:sql}
We continue Example~\ref{example:simpl:sql} and consider the choice
$(\psimpl,\dsimpl,\Csimpl)$.
Recall that $\psimpl(\ov{x})\equiv -2$ and $\dsimpl=-2$.
As explained above, this constitutes a trivial case where we can stop
and output ``$\queryExampleOne\querycont\queryExampleTwo$''. Inspecting the definition of $\ovxsim$ yields a homomorphism from
$\queryExampleTwo$ to $\queryExampleOne$, namely the mapping $h\colon
A_2\to A_1$ with
$h(\vvarx_2)=\vvarx_1$, $h(\vvary_2)=\vvary_1$, $h(\vvarz_2)=\vvarz_1$, and $h(\vvarw_2)=\sql{\textquotesingle actor\textquotesingle}$.
\end{exampleWithEndmarker}

Recall that during preparation and simplification, we have also identified trivial cases in which we can
stop and output ``$\queryOne\not\querycont\queryTwo$''.

%% file: system.tex
\begin{figure*}[tb]
\centering
\includegraphics[width=\linewidth]{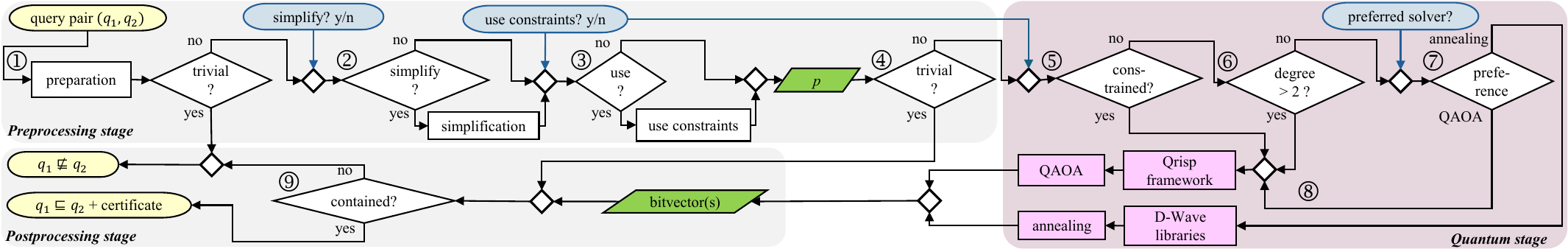}
\caption{3-stage workflow: Input problem and output (yellow ovals), parameters (blue ovals), and the decisions leading to evaluation using QAOA or annealing (pink boxes). The polynomial~($p$) and bitvectors are central artifacts (green parallelograms).}

\label{fig:workflow}

\end{figure*}

Figure~\ref{fig:workflow} shows our 3-stage workflow.  The \emph{preprocessing stage}
translates the input, a pair of conjunctive queries, into a polynomial. First is the preparation step~(\textcircled{\scriptsize 1}), as described in Section~\ref{sec:preparation}, which can detect certain
trivial cases.
Simplification~(\textcircled{\scriptsize 2}) and the addition of constraints~(\textcircled{\scriptsize 3}) are optional, and either combination can result in different polynomials.
Trivial problems~(\textcircled{\scriptsize 4}) are
 answered directly.

Next is the decision whether to use QAOA or annealing in the \emph{quantum stage}. If the problem is constrained or has a degree higher than two, the system defaults to QAOA
(steps~\textcircled{\scriptsize 5} and~\textcircled{\scriptsize 6}). In all other cases, the system follows the user's preference~(\textcircled{\scriptsize 7}).
Our solvers can be either simulator software, or quantum devices~(\textcircled{\scriptsize 8}). Software libraries (from Qrisp and D-Wave) provide access to solvers (Appendix~\ref{appendix:system} provides  architectural details).

Performing several measurements, we obtain a set of bitvectors, where
each encodes a bit-matrix. %

In \emph{postprocessing}, we check each bitvector~$\ov{x}$ in
polynomial time for validity and optimality:
We
check whether the bitvector represents
a bit-matrix with at most one 1 per row
(validity), and we
compute the function value $p(\ov{x})$ and
compare it to the known optimal value (optimality).
Despite the theoretical correctness, the additional validity check
is necessary for constrained QAOA on real hardware, due to potential hardware errors and noise.
Invalid results are directly discarded.
Each valid and optimal 
bitvector can be interpreted as a homomorphism certifying query containment~(\textcircled{\scriptsize 9}).
If no such bitvector is found, the system reports that containment does not hold.

%% file: experiments.tex
We explore the following research questions.
To not be limited by current hardware, we rely on simulators for the first three. 
\ 
\begin{enumerate}[RQ1:]

	\item
	Soundness: Can our approach correctly solve \QCP ?
	
	\item Scalability: How does the probability of solving \QCP\ problems scale with the size of these problems?
	
	\item Parameter sensitivity: How do QAOA-specific parameters affect the probability of  solving \QCP? 

	\item 
	Feasibility on current hardware: Can we successfully demonstrate a proof-of-concept on current quantum hardware?	
\end{enumerate}

\subsection{Metrics}
\label{sec:metrics}

Experiments with \emph{positive} and \emph{negative} query pairs (containment holds or does not hold) allow to capture \emph{classification outcomes} (true/false positives/negatives). Since our approach produces a certificate, it cannot produce any false positives.

We determine how likely it is to successfully solve a \QCP\ problem
and track the \emph{solution probability} as the fraction of
\emph{measurements} (reads for annealing, shots
for QAOA) for which an optimal solution is found. Here, we only consider positive query pairs, as the solution probability of negative pairs is always 100\%. %
As long as the solution probability is non-zero, we find a certifying
homomorphism which proves
that containment holds.\footnote{The probability $P^*$ 
of obtaining at least one solution grows exponentially with the number of measurements.
For example, with a solution probability of only~10\% and 20 measurements, the probability for a true positive reaches $P^*=1-0.9^{20}\approx 87.8\%$.}

\subsection{Experimental Setup}\label{subsec:ExpSetup}

\myparagraph{Query Pairs}
We manually extracted query pairs from textbooks and benchmarks. Thus,
these ``\emph{T\&B problems}'' %
are community-authored.
The textbook problems\footnote{From (1)~``Foundations of Databases'' by Abiteboul, Hull, and Vianu~\cite{ahv-book}, (2)~the online textbook ``Database Theory'' by Arenas, Barcel\'o, Libkin, Martens, and Pieris \cite{ArenasEtAl_DBT-Book} (version from 2022-08-19), (3) the booklet ``Answering Queries using Views'' by Afrati, Chirkova, and Jagadish \cite{10.5555/3360153}, (4)~lecture material by Koch taught at Saarland University in 2006 (shared in personal communication), and (5)~lecture material by Schweikardt~\cite{schweikardt_lecture_2023}.} 
include the example from Sec.~\ref{sec:prelims_qc}.
We harvested various query equivalence or query containment benchmarks~\cite{uwdbcosette,idni-tml,jena-sparql,sparql-qc-paper, sparql-qc}. 
Conjunctive SQL and SPARQL queries were manually converted to tableau form.

Figure~\ref{fig:tb_total} states the numbers of query pairs by 
their ground truth w.r.t.\ \QCP\ (positive/P or negative/N, derived manually). 
This includes problems that may later be identified as trivial.
Figure~\ref{fig:tb_nontrivial}
states the subsets of non-trivial problems that are evaluated in the
quantum stage, via annealing (for a polynomial of degree up to~2) and QAOA (for any degree, with maximum occurring degree~4).
The former are included in the latter.
We state the original number of problems, as well as the reduced
number, after simplification.

\begin{figure}[thb]
  \centering
  \footnotesize
  \setlength{\tabcolsep}{3pt}
  \renewcommand{\arraystretch}{0.9}  %
  \captionsetup[subfigure]{skip=2pt} %
    
  \begin{subfigure}[t]{0.15\columnwidth}

    \begin{tabular}{rr}
      \toprule
      \multicolumn{2}{c}{\textbf{Total}} \\
      \cmidrule(lr){1-2}
      \textbf{P} & \textbf{N} \\
      \midrule
      108 & 81 \\
      \bottomrule
      \\
    \end{tabular}

    \centering
    \caption{T\&B}
    \label{fig:tb_total}
  \end{subfigure}
  \hfill
  \begin{subfigure}[t]{0.5\columnwidth}
    \centering

    \begin{tabular}{lrrrr}
      \toprule
      \textbf{Collection}
      & \multicolumn{2}{c}{\textbf{Anneal.}}
      & \multicolumn{2}{c}{\textbf{QAOA}} \\
      \cmidrule(lr){2-3} \cmidrule(lr){4-5}
      & \textbf{P} & \textbf{N} & \textbf{P} & \textbf{N} \\
      \midrule
      Original T\&B
      & \textbf{74} & \textbf{22}
      & \textbf{87} & \textbf{26} \\
      Simplified T\&B
      & \textbf{26} & \textbf{2}
      & \textbf{32} & \textbf{3} \\
      \bottomrule
    \end{tabular}
    \caption{Non-trivial T\&B}
    \label{fig:tb_nontrivial}
  \end{subfigure}
  \hfill
  \begin{subfigure}[t]{0.32\columnwidth}
    \centering

	\tikzset{
		graphnode/.style={circle, fill=black, draw=none, inner sep=0pt, minimum size=3.2pt},
		graphedge/.style={line width=0.7pt, -{Stealth[length=1.8mm,width=1.1mm]},
			shorten <= 0.5mm, shorten >= 0.5mm, line cap=round}
	}

	\begin{tabular}{@{}c@{\hspace{4mm}}c@{}}
		\begin{tabular}{@{}c@{}}
			\begin{tikzpicture}[scale=0.75, baseline=(ref.base)]
				\coordinate (ref) at (0,0);
				\node[graphnode] (a) at (-0.55,0) {};
				\node[graphnode] (b) at ( 0.55,0) {};
				\draw[graphedge] (a) to[out=45, in=135, looseness=1.05] (b);
				\draw[graphedge] (b) to[out=-135, in=-45, looseness=1.05] (a);
			\end{tikzpicture}\\[1.5pt]
			2-cycle
			\\[4pt]
			\begin{tikzpicture}[scale=0.75, baseline=(ref.base)]
				\coordinate (ref) at (0,0);
				\node[graphnode] (u) at (-0.60,0) {};
				\node[graphnode] (v) at ( 0.00,0) {};
				\node[graphnode] (w) at ( 0.60,0) {};
				\draw[graphedge] (u) -- (v);
				\draw[graphedge] (v) -- (w);
			\end{tikzpicture}\\[1.5pt]
			2-chain
		\end{tabular}
		&
		\begin{tabular}{@{}c@{}}
			\begin{tikzpicture}[scale=0.75, baseline=(ref.base)]
				\coordinate (ref) at (0,0);
				\node[graphnode] (s0) at (0,0) {};
				\foreach \ang in {90,162,234,306,18}{
					\node[graphnode] (s\ang) at (\ang:0.75) {};
					\draw[graphedge] (s0) -- (s\ang);
				}
			\end{tikzpicture}\\[2pt]
			5-star
		\end{tabular}
	\end{tabular}

    \caption{Graphs}
    \label{fig:blocks}
  \end{subfigure}

  \caption{(a) Query pairs from textbooks and benchmarks, (b)~non-trivial subsets, and (c)~generic graph building blocks.}
  \label{fig:tb_and_blocks}
\end{figure}

We designed two synthetic families of \emph{parameterized \QCP\  problems}, based on the directed graph schema~$\dbschema_{g}$ from Example~\ref{example:preparation:cycle_chain}.
Figure~\ref{fig:blocks} illustrates the archetypal building blocks.
Family \textbf{2-cycle-to-$i$-chain} contains pairs of a 2-cycle query
and a query describing a chain of length~$i$  for $i\in\NN_{\geq
  1}$. For illustration, revisit the queries 2-cycle ($q_{\textit{2cy}}$) and 2-chain~($q_{\textit{2ch}}$) presented in Example~\ref{example:preparation:cycle_chain}.
Family \textbf{2-chain-to-$i$-star}, contains pairs of a 2-chain query and an $i$-star query for $i\in\NN_{\geq 1}$.
All these problems can be solved with annealing and QAOA.
They allow for small-stepped scaling in the number of binary variables, which here, map directly to qubits.
Moreover, preparation or simplification do not reduce the number of binary variables. (While simplification has no effect, we will enable it by default in our experiments.)
The search space is non-trivially exponential, yet  we can
manually derive two optimal (and only) solutions for each instance as ground truth (details in Appendix~\ref{app:families}).

\myparagraph{Implementation}
We implemented the workflow described in Section~\ref{sec:system}, utilizing
Eclipse Qrisp for gate-based QAOA across multiple backends and D-Wave
libraries for
annealing.  Specifically, we use  Python~3.10, qrisp~0.7.9,
dwave-ocean-sdk~9.0.0,
\allowbreak%
qiskit~2.2.1 and qiskit-aer~0.17.2. To access IQM
quantum hardware, we used Python~3.11, qrisp~0.7.17, and qiskit~2.1.2.

\myparagraph{Backends}
We  first describe the \emph{simulators}.
For annealing, we use the classical simulated annealer software SimulatedAnnealingSampler. 
The PathIntegralAnnealingSampler emulates the behavior of a quantum annealer. 
To simulate quantum circuit experiments,  we use IBM Qiskit’s CPU-based Aer.
We use four \emph{quantum hardware devices}.
For quantum annealing, we use D-Wave's Advantage System~6.4
and Advantage2 System~1.7. %
For QAOA, we use IQM Emerald %
and IBM Aachen.
Table~\ref{tab:solvers} provides details on all solvers.

\begin{table}[tbh]
\caption{Simulators (SW) with their software libraries and quantum hardware devices (HW) with number of qubits.}
	\footnotesize
    \setlength{\tabcolsep}{3pt}
    \renewcommand{\arraystretch}{0.9}  %
	\begin{tabular}{@{}llcl@{} r}
		\toprule
		\textbf{Abbrev/} & \textbf{Description} &  & \textbf{Library/Device} & \textbf{Qubits} \\
		\midrule
		SA & Sim.\ annealer & SW & dwave-neal 0.6.0 \\
		SQA & Sim.\ quantum annealer & SW & dwave-samplers 1.6.0\\
		QA-adv &D-Wave quantum annealer & HW & Advantage System 6.4 & 5,612\\
		QA-adv2 &D-Wave quantum annealer & HW & Advantage2 System 1.7 & 4,592\\
		\midrule
		AerSim & CPU-based simulator & SW & qiskit-aer 0.17.2 \\
		IQM-Emerald & Gate-based QPU & HW & IQM Emerald & 54 \\
        IBM-Aachen & Gate-based QPU & HW & IBM Aachen & 156\\
		\bottomrule
	\end{tabular}

	\label{tab:solvers}
\end{table}

\myparagraph{Standard configuration}
Unless stated otherwise, we use a standard configuration. For easier comparison, we use the same configuration for experiments with simulators and quantum hardware. 
We deliberately keep the number of measurements low to ensure that
successful outcomes cannot be attributed to random chance and to meet
the practical limitations of current hardware.

In annealing, we perform 500 reads. %
For both simulated annealing and simulated quantum annealing, we use a beta range of $[0.5,10]$ and the (default) geometric beta schedule. For simulated annealing, we use a fixed initial value. For quantum annealing, we use the (default) linear annealing schedule and an annealing time of~\mbox{$30\mu s$}.

For QAOA, we configure 30~iterations, 2~layers, 500~shots, and constraints. 
To efficiently simulate large instances, we use the matrix product state method.
In simulating QAOA, we limit the number of binary variables to~42. All simulations have a 60-minute timeout.

\myparagraph{Execution environment}
Our experiments run on a Linux server with 2 Intel Xeon Gold 6242R processors (3.1 GHz, 40 cores) and 384~GB RAM. Quantum hardware is accessed as cloud service.

\subsection{RQ1: Soundness}

We first explore soundness of the entire workflow and then of the quantum stage. We use simulators with standard configuration.

\myparagraph{Soundness of end-to-end workflow}
We process all positive T\&B query pairs from Figure~\ref{fig:tb_total} with our 3-stage workflow. This includes problems that are trivial to start with, that become trivial after simplification, and that cannot be solved, as the degree of the polynomial is too high for annealing or where the number of qubits exceeds the simulator capabilities.
We measure solution probabilities, and therefore only consider the positive query pairs.

We distinguish two evaluation paths. For the first, we configure annealing as
preferred solver and deviate from the workflow in
Figure~\ref{fig:workflow}: For polynomial degree greater two, we
return that the problem instance could not be solved (rather than defaulting to QAOA).
For the second, we configure QAOA as preferred solver.

\begin{figure}[tb]
    \centering
     \captionsetup[subfigure]{skip=2pt} %
    \begin{subfigure}[t]{0.48\columnwidth}
        \centering
        \includegraphics[%
           width=\columnwidth,trim=0 7 0 5, clip]{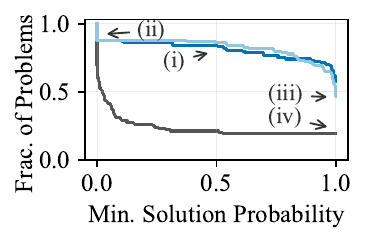}
        \caption{Annealing w/o simplification.}
        \label{fig:survival-anneal-orig}
    \end{subfigure}
    \hfill
    \begin{subfigure}[t]{0.48\columnwidth}
        \centering
        \includegraphics[%
        width=\columnwidth,trim=0 7 0 5, clip]{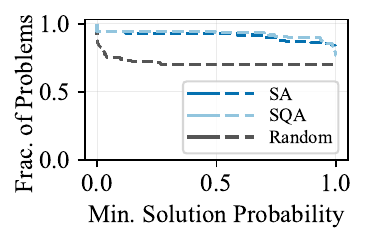}
        \caption{Annealing w/ simplification.}
        \label{fig:survival-anneal-simpl}
    \end{subfigure}

    \begin{subfigure}[t]{0.48\columnwidth}
        \centering
        \includegraphics[%
        width=\columnwidth,trim=0 7 0 0, clip]{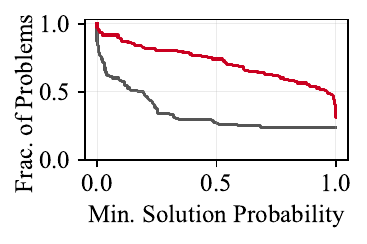}
        \caption{QAOA w/o simplification.}
        \label{fig:survival-qaoa-orig}
    \end{subfigure}
    \hfill
    \begin{subfigure}[t]{0.48\columnwidth}
        \centering
        \includegraphics[%
        width=\columnwidth,trim=0 7 0 0, clip]{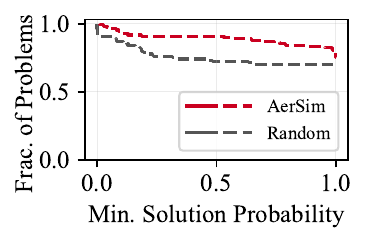}
        \caption{QAOA w/ simplification.}
        \label{fig:survival-qaoa-simpl}
    \end{subfigure}

  \caption{Soundness of 3-stage workflow of positive T\&B query pairs (on simulators). Showing fraction of problems where solution probability meets/exceeds a given threshold. 
  }

  \label{fig:solution_probabilities}
    
\end{figure}

\myparagraph{Results}
Figure~\ref{fig:solution_probabilities}
visualizes the results.
The plots~(a) and~(b) show annealing with simplification disabled/enabled. Likewise, the plots~(c) and~(d) show QAOA with  simplification disabled/enabled.

We plot the fraction of query pairs whose solution probability meets/exceeds a given threshold. 
That is, for any value on the $x$-axis, the corresponding value on the $y$-axis indicates the fraction of pairs that achieved a solution probability greater than or equal to this threshold.
To illustrate, we highlight three observations in Figure~\ref{fig:survival-anneal-orig}. Arrow~(i) points out that a solution probability of 50\% and greater is obtained with simulated annealing for about 80\% of the query pairs (which accounts for almost all QUBO instances).

A vertical function block denotes the fraction of problems that have this exact solution probability.
Specifically, the leftmost vertical function block (indicated by Arrow~(ii) at $x=0$) shows the fraction of query pairs that could not be solved by our workflow (e.g., too many qubits, or polynomial degree too high for annealing).

The $y$-value at which the curve ends (see Arrow~(iii) at $x=1$) shows the share of problem instances solved with close to 100\% solution probability.
For reference, we plot the likelihood of solving the problem by chance (``random''), obtained by generating 500 bitvectors at random. For QAOA, we generate bitvectors that are \emph{valid}: Each represents a matrix with one bit per row set to~1, to be comparable with constrained QAOA where this is enforced.

This likelihood is considerably lower than the measured probabilities (apart from the corner cases of solution probabilities close to zero or one).
The $y$-value at which the random curve ends (Arrow~(iv) at $x=1$) shows the fraction of problems that are trivial. Thus, the difference in fraction between Arrows~(iii) and~(iv) is the gain that can be contributed to our approach.

Throughout, we observe even higher solution probabilities when simplification is enabled (even for guessing at random).

\myparagraph{Discussion}
With simplification, more input problems become trivial to start with, and many problems become easier to solve, since the number of binary variables can often be reduced.
Our results show that the simplification step is indeed very effective.

In general, problems with high solution probabilities are easy to solve for our approach. It is plausible that we observe many such cases, given that the textbook problems were designed as pencil-and-paper exercises.
Even for problems with very low solution probabilities, a single solution found is sufficient to solve a \QCP\  problem. Overall, this \textbf{confirms that our workflow can correctly solve nearly all input problems}, esp.\ with simplification enabled.

\smallskip
\myparagraph{Soundness of quantum stage} 
We now focus on the quantum stage in isolation and only consider the nontrivial T\&B problems (Figure~\ref{fig:tb_nontrivial}).
We deliberately include the negative query pairs to check whether our
objective function is implemented correctly.

\begin{table}[tb]

	\caption{Classification outcomes for nontriv.\ T\&B problems.}
	\label{tab:soundness}
	\centering
	\footnotesize
    \setlength{\tabcolsep}{3pt}
    \renewcommand{\arraystretch}{0.9}  %
    
	\begin{tabular}{lrrrrr rrrrr}
		\toprule
		\textbf{Collection} 
		& \multicolumn{5}{c}{\textbf{Annealing} (SA, SQA)} 
		& \multicolumn{5}{c}{\textbf{QAOA} (AerSim)} \\
		\cmidrule(lr){2-6} \cmidrule(lr){7-11}
		& \textbf{TP} & \textbf{FP} & \textbf{FN} & \textbf{TN} & $\boldsymbol{\bot}$ 
		& \textbf{TP} & \textbf{FP} & \textbf{FN} & \textbf{TN} & $\boldsymbol{\bot}$ \\
		\midrule
		Original T\&B   
		& 74 & \multicolumn{1}{c}{/} & 0 & 22 & 0 
		& 84 & \multicolumn{1}{c}{/} & 0 & 25 & 4 \\
		Simplified T\&B 
		& 26 & \multicolumn{1}{c}{/} & 0 & 2 & 0 
		& 32 & \multicolumn{1}{c}{/} & 0 & 3 & 0 \\
		\bottomrule
	\end{tabular}
    
    \textbf{Legend:} True/False positives/negatives; $\bot$: simulator capacity exceeded.
\end{table}

\myparagraph{Results}
Table~\ref{tab:soundness} reports perfect true positive~(TP) and true negative~(TN) rates for all T\&B problems that the simulators can process
(in a small number of QAOA experiments,
the number of qubits exceeded the capacities of the simulator, denoted ``$\bot$'').
Recall that by design, false positives~(FP) cannot occur~(denoted ``/''). In the experiment shown, we did not observe any false negatives~(FN). With annealing, both simulated annealing and simulated quantum annealing (SA and SQA) produce identical results. 

We found our experiments to be highly repeatable, with one exception, as  QAOA sometimes produces a single false negative.

\myparagraph{Discussion}
The classification outcomes \textbf{confirm the soundness of our implementation for the quantum stage}.

\subsection{RQ2: Scalability}\label{sec:experiments-scalability}

Using the parameterized \QCP\ problems, we
explore how the probability of solving \QCP\ problems scales with problem size.
We use the same setup as before. In addition, we run
AerSim with 4~Layers, 100~iterations, and 5,000~shots (referred to as ``bigshot'').

\begin{figure}[tb]
    \centering
    \captionsetup[subfigure]{skip=2pt} %
         
    \begin{subfigure}{0.48\columnwidth}
        \centering
        \includegraphics[%
        width=\columnwidth,trim=0 3 0 0, clip]{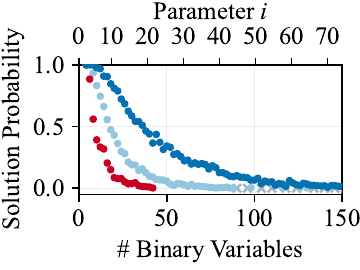}
        \subcaption{2-cycle-to-$i$-chain}
        \label{fig:2-cycle-to-i-chain}
    \end{subfigure} 
    \hfill
    \begin{subfigure}{0.48\columnwidth}
        \centering
        \includegraphics[%
        width=\columnwidth,trim=0 3 0 0, clip]{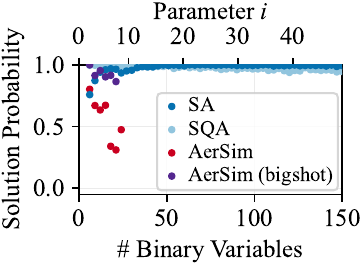}
        \subcaption{2-chain-to-$i$-star}
        \label{fig:2-chain-to-i-star}
    \end{subfigure}
    \caption{Problem size vs.\ solution probability for parameterized \QCP\ problems. Scaling behavior is problem-specific.}

    \label{fig:scaling_behavior}
\end{figure}

\myparagraph{Results}
Figure~\ref{fig:scaling_behavior} shows how the solution probability changes as the size of the input problem increases. The bottom horizontal axis states the number of binary variables in the polynomial, which directly maps to the number of qubits; the top axis states the parameter
controlling the problem size.
Grey crosses indicate a solution probability of zero (crosses sampled to avoid overplotting).

The scaling behavior for 2-cycle-to-$i$-chain (Figure~\ref{fig:2-cycle-to-i-chain}) is as expected, and both annealing and QAOA show the same trend: the solution probability decreases as the problem size grows. Simulated annealing reaches larger problem instances, due to the more favorable scaling behavior of the simulator~(SA).

However, 2-chain-to-$i$-star (Figure~\ref{fig:2-chain-to-i-star})
displays a strikingly different scaling behavior, and much higher solution probabilities. With simulated  annealing (SA), smaller instances even yield lower solution probabilities than larger ones, and the curve quickly reaches 100\%,
i.e., each read returns an optimal solution. 
With QAOA (AerSim), the solution probabilities remain higher compared to 2-cycle-to-$i$-chain, but there is a steep drop until all simulations exceed the timeout.
To investigate, we compare with the ``bigshot'' configuration (purple dots in Figure~\ref{fig:2-chain-to-i-star}). Evidently, increasing the number of layers, shots, and iterations improves the solution probabilities.

\myparagraph{Discussion} %
2-cycle-to-$i$-chain scales as expected:
the solution probability decreases with increasing problem size: Each additional variable increases the probability that the respective heuristic algorithm does not terminate in a global minimum.
With QAOA, we exceed the capacity of the simulator sooner.

To better understand the outstanding performance of 2-chain-to-$i$-star, we investigated the characteristics of this problem family. 
Section~\ref{sec:investigations} discusses the problem-specific energy
landscape,
where each local minimum is at the same time a global minimum. This provides favorable conditions for all solvers. Section~\ref{sec:investigations} also elaborates our hypothesis on what causes slightly inferior performance for small values of the parameter~$i$ with simulated annealing~(SA).

For QAOA, we observe a clear drop in performance in Figure~\ref{fig:2-chain-to-i-star}, yet the solution probabilities are nevertheless higher than in Figure~\ref{fig:2-cycle-to-i-chain}. Our ``bigshot'' experiment indicates that improving the configuration parameters also improves QAOA performance.

Clearly, the \textbf{scaling behavior is problem-specific}. 2-cycle-to-$i$-chain scales as can be expected based on state-of-the-art, 
and chain queries of length 20 are no longer small queries.
However, \textbf{2-cycle-to-$i$-star surprises with exceptionally good performance}. This strongly motivates further research to investigate \QCP\ problems in terms of their energy landscapes.

\subsection{RQ3: Parameter Sensitivity for QAOA}
\label{sec:experiments_parameters}

We explore the parameter sensitivity of QAOA on the 2-cycle-to-$i$-chain family, given its generic scaling behavior.
Again, the experiments are run on simulators. We vary the number of layers (0, 2, 4, and~6) and compare the unconstrained and the constrained variant.

\begin{figure}[tb]
  \centering
  \captionsetup[subfigure]{skip=2pt} %

  \begin{subfigure}{0.455\columnwidth}
    \centering
    \includegraphics[width=\columnwidth,trim=0 3 0 0, clip]{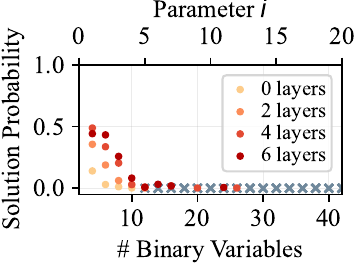}
    \caption{30 iterations, w/o constr.}
    \label{fig:rainbow-30its-unconstr}
  \end{subfigure}\hfill
  \begin{subfigure}{0.455\columnwidth}
    \centering
    \includegraphics[width=\columnwidth,trim=0 3 0 0, clip]{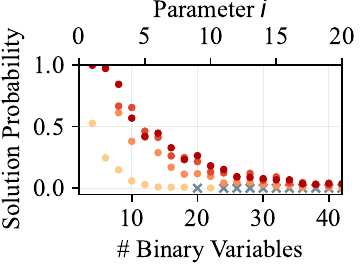}
    \caption{30 iterations, w/ constraints.}
    \label{fig:rainbow-30its-constr}
  \end{subfigure}

    \caption{Solution probabilities for 2-cycle-to-$i$-chain, varying input size and QAOA parameters in simulation (AerSim). The constrained variants (shown right) are always superior.}
   
	\label{fig:rainbow_chart_qrisp_cpu}
\end{figure}

\myparagraph{Results}
Figure~\ref{fig:rainbow_chart_qrisp_cpu} 
shows how the solution probability changes with the input size as we vary QAOA parameters.
Again, gray crosses indicate a solution probability of zero. 
The \QCP\ problem is rarely solved beyond 20~binary variables with zero layers; adding layers lifts the probability. %
Enabling constraints is very effective.

\myparagraph{Discussion} 
The results confirm that our standard configuration (30 iterations, 2 layers, 500 shots) is reasonable.
\textbf{QAOA with constraints is very effective} in reducing the solution space and thus increasing the solution probabilities. There are scenarios where we still find solutions for more than~40
binary variables, which is impressive given the exponential search space.

\begin{figure*}[ht]
  \centering
  \captionsetup[subfigure]{skip=2pt} %
  
  \begin{subfigure}{0.31\textwidth} %
    \centering
    \includegraphics[%
    width=\columnwidth,trim= 0 3 0 0, clip]{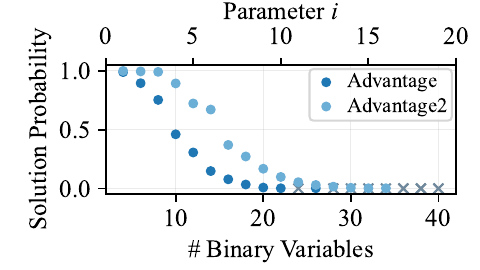}
    \caption{D-Wave}
    \label{fig:hardware-d-wave}
  \end{subfigure}
  \begin{subfigure}{0.31\textwidth}
    \centering
    \includegraphics[%
    width=\columnwidth,trim= 0 3 0 0, clip]{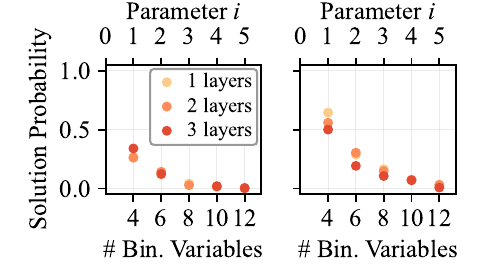}
    \caption{IQM-Emerald: w/o and w/ constraints}
    \label{fig:hardware-iqm-emerald}
  \end{subfigure}
  \begin{subfigure}{0.31\textwidth}
    \centering
   \includegraphics[%
   width=\columnwidth,trim= 0 3 0 0, clip]{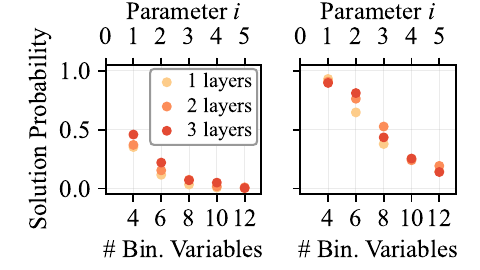}
    \caption{IBM-Aachen: w/o and w/ constraints}
    \label{fig:hardware-ibm-aachen}
  \end{subfigure}

  \caption{Quantum devices: 2-cycle-to-$i$-chain on (a) quantum annealers and (b+c) QAOA on gate-based QPUs.}
  \label{fig:quantum_hardware}

\end{figure*}

\subsection{RQ4: Experiments on Quantum Hardware}
\label{sec:experiment_real_hw}

We explore how the effects observed with simulators transfer to current quantum hardware. 
The experiments are run on the four
quantum hardware devices
using the standard configuration.

\myparagraph{Results}
Figure~\ref{fig:quantum_hardware} 
shows how the solution probabilities scale for the 2-cycle-to-$i$-chain problem.
As expected, the \QCP\ problems that we can  solve on real hardware are much smaller-scale.

With quantum annealing (Figure~\ref{fig:hardware-d-wave}), we can solve larger problem instances than with QAOA, yet we are limited to QUBO problems. Here, the more recent generation of hardware (D-Wave Advantage2) outperforms the earlier generation (D-Wave Advantage).
Figures~\ref{fig:hardware-iqm-emerald}\footnote{The left plot in Figure~\ref{fig:hardware-iqm-emerald} is missing a single data point (for 8~binary variables and 3~layers), due to a vendor-side timeout when connecting to the remote device.
} and~\ref{fig:hardware-ibm-aachen} show QAOA experiments on devices by IQM and IBM.
As in the simulation,  the constrained variant outperforms the unconstrained variant. 
A key difference concerns the effect of increasing the number of layers. Different from the simulator, more layers do not always improve the solution probability.

\myparagraph{Discussion} 
We can \textbf{successfully solve \QCP\ problems on real quantum hardware, albeit of smaller size}. Our QAOA experiments on hardware confirm that the \textbf{constrained variant is superior}. On an idealized simulator, adding layers increases expressivity and can improve solution probability. Yet on real devices, deeper circuits can accumulate more gate and decoherence errors, which may cause %
a loss in fidelity.

%% file: investigation.tex
To better understand the different scaling behavior of the two parameterized \QCP\ problems
(cf.\ Figure~\ref{fig:scaling_behavior}),
we investigated their
{energy landscapes} (cf.\ Sec.~\ref{sec:preliminaries-quantum}).
Recall that we aim to minimize the polynomial
$p(\ov{x})$ that represents
the given query pair $(\queryOne,\queryTwo)$.
The \emph{energy landscape}
of $p(\ov{x})$ refers to the distribution of polynomial values across all possible input bitvectors
(in this context called \emph{states}), where (values of) neighboring
bitvectors (i.e., they differ in one bit)  are considered to be directly connected.

While the two problem families seem similar at first, the energy landscape of 2-chain-to-$i$-star exhibits two distinguishing structural characteristics.
First, it is shaped such that every local minimum is also a global minimum, whereas 2-cycle-to-$i$-chain exhibits an increasing number of local but not global minima. Therefore, the optimum is much more likely to be reached for the former, as the solvers do not \enquote{get stuck} locally. %
Second, the energy landscape of 2-chain-to-$i$-star 
has a large plateau of zero-energy states that is directly connected to all negatively valued (and in particular, the optimal-valued) states. Due to these connections, the optimal value can be found more frequently, which explains the high solution probabilities in Figure~\ref{fig:2-chain-to-i-star}.
For simulated annealing, however, the direct connections can cause the solver to reversely
jump out of the optimum and then \enquote{get lost} in the zero-plateau. The probability
for that jump
decreases with growing~$i$ due to larger energy differences. Moreover, the likelihood of \enquote{getting lost} in the zero-plateau
also decreases, as the relative size of the zero-plateau decreases with growing~$i$.
Regarding Figure~\ref{fig:2-chain-to-i-star}, we believe
that for small~$i$, the simulated annealer sometimes escapes into the zero-plateau, leading to the lower solution probabilities observed on smaller instances.
This effect is specific to SA and not observed for (S)QA or QAOA.
Appendix~\ref{app:energy_landscapes} provides further details.

%% file: discussion.tex
We discuss our insights along the three stages of our workflow. 

\myparagraph{Preprocessing}
Our experiments show that  simplification can be highly effective.
Of the query pairs from textbooks and benchmarks, a large share becomes trivial or is reduced to quadratic form, so that they can be directly executed on quantum annealers. This has practical relevance, as it provides access to thousands of qubits, currently a magnitude more than gate-based quantum computers.

\myparagraph{Quantum stage}
Our experiments confirm that our approach is sound. We did not custom-tailor the \QCP\ problems to the low-level specifics of the quantum hardware (e.g., QPU topology), as often done in academic explorations.
Rather, we included community-authored problem instances, which underlines our success.

We investigated the scaling behavior of two families of parameterized \QCP\ problems. While one exhibits the expected behavior, the other displays consistently high solution probabilities. Our analysis suggests that this can be attributed to the underlying energy landscapes. Exploring the relationship between problem formulation and energy landscape may provide a new research perspective.

We evaluated our algorithm on current quantum hardware. Results across heterogeneous backends suggest generalizability and portability. Despite current hardware limitations, we successfully solved  \QCP\
 problems at smaller scale. Looking ahead, the gap between simulators and hardware may narrow as vendors improve calibration, reduce error rates, and increase qubit counts.

\myparagraph{Postprocessing}
Our problem formulation ensures the absence of false positives, and postprocessing produces verifiable witnesses for positive answers.
We did not observe high rates of false negatives, which constitute missed opportunities in a database query engine using containment checks for query optimization.
These are crucial prerequisites for a future database system integration. 

\myparagraph{Outlook}
We explore a basic three-stage workflow. An orthogonal line of research considers hybrid solvers that combine classical and quantum methods, with QAOA as one component. %
We also plan to investigate whether auxiliary variables allow reformulating all \QCP\ problems into quadratic form for execution on annealers.

%% file: related.tex
\emph{Query containment} is a fundamental decision problem with numerous practical applications. For conjunctive queries under set semantics (the focus here), containment is known to be NP-complete~\cite{ChandraMerlin}. Under bag semantics, multiplicities matter and the problem becomes substantially more intricate. A characterization analogous to Theorem~\ref{thm:ChandraMerlin} is not known, and it remains open whether the query containment problem %
is decidable.
 Recent advances in database theory have established decidability for restricted fragments
and undecidability for more expressive query languages, cf.~\cite{DBLP:conf/pods/KonstantinidisM19,DBLP:journals/pacmmod/MarcinkowskiO25}. 

The problem of deciding \emph{query equivalence} is related. Recent work on equivalence of restricted SQL queries encodes the problem for constraint solvers~\cite{10.14778/3236187.3236200}. In the case of inequivalence,  \emph{Cosette}~\cite{DBLP:conf/sigmod/ChuLWCS17} produces a data instance serving as a
counterexample. This validates the decision, similarly to the certifying witnesses in our approach.

A line of research has been investigating \emph{quantum computing for cost-based query optimization}~\cite{dblp:journals/pacmmod/schonbergersm23,DBLP:journals/pvldb/LiuSS25,DBLP:conf/qce/FranzWGM24,DBLP:journals/pvldb/LiuSS25,10.1145/3579142.3594298,10.1145/3736393.3736690,DBLP:journals/pvldb/Trummer016,DBLP:journals/access/FankhauserSFS23}. Like most of
these works, we also utilize QUBO and QAOA formulations, yet we target semantic query optimization for the first time.

There is a recent trend showing that formulations of database problems originally designed for quantum computing can perform remarkably well when executed on classical hardware. These \emph{quantum-inspired approaches}~\cite{DBLP:journals/pvldb/SchonbergerTM23,DBLP:journals/pacmmod/SchonbergerTM25} demonstrate that reformulating classical problems  can unlock new solution strategies.
Inversely to applying quantum computing to database problems, there are also proposals to apply database theory research to solve quantum computing problems, such as simulating  quantum circuits~\cite{DBLP:journals/pvldb/HaiHCLG25,10.1145/3736393.3736696}. %

Within the \emph{quantum algorithms} community, there is an active line of research on quantum formulations for graph structured problems. This includes graph isomorphism~\cite{CALUDE201754,10.1145/3569095,DBLP:journals/sncs/HuaD20}, which can be seen as a restricted form of graph homomorphism, as well as quantum notions of homomorphisms~\cite{DBLP:journals/jct/MancinskaR16}.
While this work is related, the modeling assumptions differ. In quantum algorithms literature, graphs are typically represented by adjacency matrices, while we consider homomorphisms that arise from tableau queries. Thus, existing formulations would require substantial modification.

%% file: conclusion.tex
We present a constructive method for discovering a homomorphism that
certifies query containment for conjunctive queries under set
semantics. Our approach guarantees the absence of false positives and,
for the first time, enables experimental
evaluation on quantum hardware.
Although current quantum devices are 
not yet suitable for real-world deployment, our results provide a proof of concept that query containment can be
translated
into an optimization problem 
and solved on quantum devices, taking a first step toward quantum-powered semantic query optimization.

%% file: appendix_quantum_comp.tex
\section{Details omitted in Section~\ref{sec:preliminaries-quantum}}
\label{sec:appendix_quantum_comp}

In this section, we provide for the interested reader more details of the solution techniques
(quantum algorithms and simulated annealing)
mentioned in
Section~\ref{sec:preliminaries-quantum} and used to solve the discrete optimization problems formulated in
Section~\ref{sec:theory}.

Particular focus is put on quantum computing methods for solving optimization problems of the form
\begin{align}
\label{app_optimization}
\ov{x}\ \in \ \argmin_{\ov{x}\in C}p(\ov{x}).
\end{align}
Here, $p\colon \myset{0,1}^n\rightarrow\R$ is the \emph{objective function}; it is a polynomial on $n$ binary variables $\ov{x}=(x_1,\dotsc,x_n)$ with real-valued coefficients.
The \emph{search space} is the set
$C\subseteq\myset{0,1}^n$, and the goal is to find a tuple $\ov{x}\in C$ such that the value $p(\ov{x})$ is as small as possible. This optimization problem is called \emph{unconstrained} if $C=\myset{0,1}^n$, and \emph{constrained} if $C\subsetneq\myset{0,1}^n$.

The polynomial $p$ is given by a finite sum of monomials
\begin{align}
\label{app_polynomial}
p(\ov{x})=\sum_{k\in K} m_k(\ov{x}),\quad \text{with} \ \ m_k(\ov{x})=c_k{\cdot}\prod_{i\in I_k}x_i\,,
\end{align}
where $K$ is a finite set, $c_k$ is a real-valued coefficient and $I_k\subseteq\myset{1,\ldots,n}$ for every $k\in K$. The degree of such a monomial $m_k(\ov{x})$ is $\deg m_k(\ov{x}) \deff |I_k|$, and the degree of $p$ is the maximum degree of its monomials.
Since each variable $x_i$ can only be assigned the values 0 and 1, every linear term can also be expressed as a quadratic term by replacing a variable $x_i$ by $x_i^2$ (due to $0\cdot0=0$ and $1\cdot1=1$). In particular, if the degree of $p$ is $\leq 2$, this implies that we can assume without loss of generality that \emph{every} monomial of $p$ has degree exactly 2; this setting is known as \emph{Quadratic Unconstrained Binary Optimization} problems (for short: QUBOs). Therefore, such a polynomial can alternatively be represented by a real-valued upper triangular $n\times n$-matrix $Q$, representing the objective function
\begin{equation}
\label{app_eq:objective}
p(\ov{x}) \ \deff \ \ \ov{x}^T \, Q\, \ov{x} \ \ = \ \ \sum_{i \leq j} \, q_{ij} \cdot x_i \cdot x_j\,,
\end{equation}
where $q_{ij}$ is the entry in row $i$ and column $j$ of $Q$.
Quantum annealing can then be used to find
\begin{equation}
\label{app_eq:QUBO}
\ov{x}\ \in \ \argmin_{\ov{x}\in C}p(\ov{x}).
\end{equation}

In this work, we utilize the following quantum computing methods to solve the optimization problems: 

\begin{itemize}

\item Quantum Approximate Optimization Algorithm (QAOA) for gate-based quantum computers

\item Simulated annealing and quantum annealing 

\end{itemize}

\subsection{QAOA for gate-based quantum computers}
\label{app:QAOA}

The \emph{Quantum Approximate Optimization Algorithm}
(QAOA)
\cite{farhi_quantum_2014} was developed to find approximate solutions to combinatorial optimization problems. 
The algorithm has a hybrid quantum-classical structure in which quantum and classical components alternate. During the quantum part, QAOA evaluates the expected value of the objective function with respect to a parameter-dependent quantum state. The classical output of the measurements on the quantum circuit acts as a feedback which is processed in the classical part of the algorithm to heuristically compute new parameters for the next quantum iteration. The new parameters are computed with the objective of minimizing the expected value of the objective function, thereby increasing the probability of (near-)optimal results, which eventually converges. 

The quantum circuit for QAOA is composed of layers: Each layer is defined by a function of the so-called \emph{mixing Hamiltonian} $H_m$ that describes the initial state and the \emph{problem Hamiltonian}
\begin{align*}
H_p=\sum\limits_{\ov{x}\in\myset{0,1}^n}p(\ov{x})\ket{\ov{x}}\bra{\ov{x}}
\end{align*}
that describes the objective function. For simplicity, one can imagine that for advancing layers in the circuit, the quantum state shifts from the initial state (i.e., the ground state -- which is the eigenvector associated with the minimum eigenvalue -- of the \emph{mixing Hamiltonian}) towards the solution state (i.e., the ground state of the
\emph{problem Hamiltonian}). 

To be precise, an initial state $\ket{\psi_0}$ is evolved by alternatingly applying the problem-specific operator
\begin{align*}
U_p(\gamma) \ \ = \ \ e^{-\mathrm{i}\gamma H_p}
\end{align*}
and the mixer operator 
\begin{align*}
U_m(\beta) \ \ = \ \ e^{-\mathrm{i}\beta H_m},
\end{align*}
where $\mathrm{i}$ denotes the complex number $\sqrt{-1}$.

For $\ell$ layers of QAOA one obtains the parameter-dependent quantum state
\begin{align*}
\ket{\psi_{\ell}(\beta,\gamma)} \ = \ U_m(\beta_{\ell})U_p(\gamma_{\ell})\dotsb U_m(\beta_1)U_p(\gamma_1)\ket{\psi_0}
\end{align*}
for real parameters $\beta=(\beta_1,\dotsc,\beta_{\ell})$ and $\gamma=(\gamma_1,\dotsc,\gamma_{\ell})$. 
The expectation of an objective function is measured by repeated preparation and sampling
from the state $\ket{\psi_{\ell}(\beta,\gamma)}$. A classical optimizer that aims to minimize the expectation of the objective function is used to determine the updated parameters $(\beta',\gamma')$, and this procedure is repeated until the optimizer converges. The problem solution is found by sampling from the final state $\ket{\psi_{\ell}(\beta^*,\gamma^*)}$. 

In the unconstrained setting, i.e., when searching for
$\ov{x}\in\myset{0,1}^n$, the initial state is prepared as uniform superposition 
\begin{align*}
\ket{\psi_0} \ \ = \ \ \frac{1}{2^{n/2}}\sum\limits_{\ov{x}\in\{0,1\}^n}\ket{\ov{x}},
\end{align*}
and the mixer operator is implemented by parameterized single-qubit rotations in constant depth.
For the polynomial objective function \eqref{app_polynomial} the operator $U_p$ can be written as
\begin{align*}
U_p(\gamma)\ \ = \ \ \prod\limits_{k\in K}e^{-\mathrm{i}\gamma\sum c_k m_k(\ov{x})\ket{\ov{x}}\bra{\ov{x}}}
\end{align*}
where we use that $\exp(A+B)=\exp(A)\exp(B)$ if $A$ and $B$ commute, i.e., $[A,B]=AB-BA=0$.

Observe that the monomials satisfy
\begin{align*}
m_k(\ov{x})=
\begin{cases}
1 & \text{if } x_i=1 \text{ for all } i\in I_k\\
0 & \text{otherwise.}
\end{cases}
\end{align*}
Therefore, each unitary in the above product acts on the qubits in $I_k$ diagonally as
\begin{equation}
e^{-\mathrm{i}\gamma\sum c_k m_k(\ov{x})\ket{\ov{x}}\bra{\ov{x}}} \ \ = \ \
\begin{bmatrix}
1 & 0 &  0 & 0 & 0\\
0 & 1 & 0 &  0 & 0 \\
0 & 0 & \ddots &  0 & 0 \\
0 & 0 & 0 & 1 & 0 \\
0 & 0 & 0 & 0 & e^{-\mathrm{i}\gamma c_k} \\
\end{bmatrix}.
\end{equation}
This is realized by a so-called multi-controlled phase (MCP) gate that applies a phase shift $\exp(-\mathrm{i}\gamma c_k)$ to the states where all qubits in $I_k$ are in state $\ket{1}$. In practice, this multi-qubit gate is decomposed into $\mathcal O(|I_k|)$ single and two-qubit gates requiring a total depth of $\mathcal O(|I_k|)$ \cite{Khattar2025}. If the sets $I_{k_1}$ and $I_{k_2}$ for indices $k_1, k_2\in K$ are disjoint, the gates can be applied simultaneously. Therefore, the depth of the quantum circuit to apply the problem-specific operator depends on the degree and structure of the polynomial $p$. E.g., for $p(\ov{x})=\sum_{k=1}^{n-1}x_k x_{k+1}$ the circuit depth would be constant independently of the number of qubits $n$. Moreover, for a general quadratic polynomial, the circuit depth would be $\mathcal O(n)$ since the chromatic index of the complete graph on $n$ vertices is upper bounded by $n$. 

In general, if $p$ is a degree $d\leq n/2$ polynomial (in practice, the number of binary variables is usually much higher than the polynomial degree $d$), then the problem-specific operator is implemented with $\mathcal O(d\cdot n^d)$ gates and circuit depth: The number of square-free monomials of degree at most $d$ is
$\sum_{k=1}^d{n \choose k}$ which is in $\mathcal O(n^d)$ for a fixed $d\leq n/2$, and implementing a degree $d$ monomial requires $\mathcal O(d)$ gates and circuit depth.

\paragraph{Constrained QAOA}

For constrained optimization problems, the set of feasible solutions $C\subsetneq\{0,1\}^n$ is a strict subset of $\{0,1\}^n$.
Strategies for incorporating such constraints into the QAOA circuit have been introduced \cite{bartschi2020grover}. The main idea is that only the subset of feasible solutions is explored and the reduced size of the search space improves the optimization results.
This is achieved by constructing a procedure that prepares a uniform superposition of all feasible states 
\begin{align*}
\ket{\psi_0}\ \ = \ \ U_C\ket{0}\ \ = \ \ \frac{1}{\sqrt{|C|}}\sum\limits_{\ov{x}\in C}\ket{\ov{x}}.
\end{align*}
Thus, the ability to construct an efficient state preparation procedure depends on the constraints of a specific problem. 

For the query containment problem, the search space $C$ is the subset of $\{0,1\}^{n_1\cdot n_2}$ such that for $n_2$ groups of $n_1$ binary variables exactly one out of each group of binary variables assumes the value of 1 and the rest assume the value of 0, i.e.,
\begin{align*}
C\ =\ \left\{\ov{x}\in\myset{0,1}^{n_1\cdot n_2}\mid \sum_{r=1}^{n_1}x_{n_1{\cdot} s +r}=1\;\forall\; s\in\myset{0,\dotsc,n_2{-}1}\right\}.
\end{align*}
In terms of a quantum state, this is realized by an $n_2$-fold tensor product of so-called $W$-states of $n_1$ qubits each:
\begin{align*}
W \ \ = \ \ \frac{1}{\sqrt{n_1}}\left(\ket{10\dotsc 0}+\ket{010\dotsc 0}+\dotsb+\ket{0\dotsc01}\right)
\end{align*}
A suitable state preparation procedure and an appropriate
mixer operator ensure that the QAOA circuit leaves the subspace of feasible solutions invariant; they
are provided by circuits of depth~$\mathcal O(n_1)$~\citep{zander-solving-2024}.

\paragraph{Complexity of QAOA}
\label{app:QAOA complexity}

The main complexity parameters of QAOA are
the number of \emph{layers} of the circuit,
the number of \emph{iterations} of the classical optimizer, and
the number of \emph{shots} performed in each iteration:

The \emph{number of layers} refers to the number of steps between the initial state and the solution state, thus defining the granularity of the discretization. A higher number of layers per circuit is associated with an increasing approximation quality, but at the same time, higher circuit depth and more gates increase the difficulty of a reliable execution on noisy intermediate-scale (NISQ) quantum hardware, which is the kind of quantum hardware that is available nowadays. From the discussion in the previous paragraph it follows that for an arbitrary quadratic polynomial $p(\ov{x})$ of $n$ binary variables, the circuits for QAOA with $\ell$ layers have a favorable scaling of depth $\mathcal O(\ell{\cdot} n)$ for both the constrained and unconstrained version. In general, for a degree $d\leq n/2$ polynomial $p(\ov{x})$, the depth scales polynomially in the number of binary variables $n$ as $\mathcal O(\ell{\cdot} d{\cdot} n^d)$.
Hence,
they can be efficiently executed on a quantum computer. This is, however, not the case when a classical simulator is used to emulate the behavior of a quantum computer: handling quantum states of $n$ qubits requires storing and manipulating $2^n$ complex amplitudes.

The \emph{number of iterations} refers to the maximum allowed number of classical optimizations. In each iteration, the quantum circuit is optimized classically based on its performance in the previous iteration. If the quantum circuit can no longer be optimized further (based on the chosen classical optimizer) or the maximum number of iterations is reached, the algorithm terminates.

The explanation of the \emph{number of shots} requires a short review of the connection between the quantum and the classical part of the algorithm. The quantum part of the algorithm returns an expected value of the objective function with respect to the state represented by the quantum circuit, which is then fed into the classical part. In order to estimate the expected value of the objective function without knowing the full probability distribution of possible results of the quantum circuit, the circuit is executed a number of times and its results are measured, leading to an estimation of the expected value as the mean of the measured results. The number of shots refers to the number of executions that are performed for this estimation.

\paragraph{Limitations of NISQ hardware}
\label{app:QAOA NISQ}

Currently available quantum hardware is referred to as Noisy Intermediate-Scale Quantum (NISQ) hardware. Only a small number of physical qubits (e.g., up to 156 qubits on state-of-the art IBM devices) are available, and computations are affected by limitations such as short coherence times that allow only shallow circuits to be run and noise leading to gate and readout errors. 
In particular, generating highly entangled quantum states across multiple qubits -- a key resource in quantum computation -- remains challenging \cite{zander2025benchmarking}. 
To overcome these limitations, quantum error correction, which combines several physical qubits to build a logical qubit, will be required. To date, there have been encouraging demonstrations of this technology \cite{google2025quantum, putterman2025hardware}, yet significant challenges in hardware and software engineering need to be addressed for quantum computing to mature.

At this point, quantum hardware is accessed mainly through cloud-based services such as IBM Quantum Platform \cite{website-ibm-quantum} or IQM Resonance \cite{website-iqm-resonance}. Hybrid algorithms such as QAOA consist of a feedback loop between quantum and classical computation. This translates to repeated round-trips between the local client and the remote quantum backend. For each iteration, quantum circuits are sent to the remote backend and measurement results are transmitted back to the local client.
In addition to
latency issues, this also results in being re-queued in every iteration. The latter issue has been addressed by IBM with the ``Session'' execution mode that provides exclusive access to the QPU. Moving forward, a tighter integration between quantum and classical high-performance computing will enable hybrid programs to run seamlessly and efficiently. 

Due to the aforementioned restrictions regarding reliability and limited access to quantum hardware, quantum algorithms and applications thereof are mostly studied using simulators that run on classical hardware.
Simulators
are software tools that aim to emulate quantum computations. In general, the costs of simulating quantum computations on classical computers increase exponentially in the number of qubits: For a system of $n$ qubits the quantum state is represented by a $2^n$-dimensional state vector. E.g., for $n=40$ storing the state vector requires approximately 16 terabytes ($2^{40}\cdot 16$ bytes) of memory. In practice, if quantum states are sparse or factorized, state vector simulators such as the Qrisp simulator \cite{seidel2024qrisp} can enable efficient simulations of larger systems of qubits.
Moreover, matrix product state simulators (e.g., provided by Qiskit Aer \cite{qiskit2024}) offer an efficient tool for simulating quantum computations that rely only on a restricted amount of entanglement~\cite{vidal2003efficient}.

\subsection{Simulated annealing \& quantum annealing}
\label{app:SAQA_appendix}
After describing QAOA in the previous subsection, we now present the foundations of two further optimization methods that have been used to solve the QUBOs considered in this work. These methods are referred to as \emph{simulated annealing} and \emph{quantum annealing} (see also Section~\ref{sec:SAQA}).
\ \\ \\
\emph{Simulated annealing: } 	
The \emph{simulated annealing} method (cf., Algorithm~\ref{algo:SimAnn}) starts at some initial state
$\ov{x} \in C$ having the value $c\deff p\left(\ov{x}\right)$. From there, it iterates through $C$ to find a minimizer $\ov{x}_{\min} \in C$ or an approximation of $\ov{x}_{\min}$. In each iteration, a random neighbor $\ov{x}_\text{new}$ of $\ov{x}$ is chosen
(here, ``neighbor'' means that one bit is flipped). If $c_\text{new} \deff p\left(\ov{x}_\text{new}\right)$ is smaller than $c$, it is considered as the current solution. Otherwise, a deterioration is accepted if a probability denoted as $\textit{AcceptProb}$ exceeds a randomly chosen value. To compute $\textit{AcceptProb}$, a parameter $T$ as well as the difference between the old and new value $\left(c_\text{new}-c\right)$ are required:
\begin{equation}
	\label{eq:Boltzmannprob}
	\textit{AcceptProb}\left(T,c_\text{new}-c\right)  \ \ = \ \ \min \left(1,\exp\left(-\frac{c_\text{new}-c}{k_b \cdot T}\right)\right).
\end{equation}
This expression is adapted from statistical physics. It is known as \emph{Boltzmann distribution}, which yields the probability of measuring an energy $c_\text{new}$ given a reference energy $c$ and a temperature $T$. $k_b$ is denoted as \emph{Boltzmann constant}. $\textit{AcceptProb}$ implies that accepting $\ov{x}_\text{new}$ is less likely if $T$ is low and the difference $c_\text{new}-c$ is large. It remains to clarify how the parameter $T$ or the ``temperature'' $T$ is determined. Here, a monotonically decreasing function $\textit{coolingSchedule}$ is used, yielding for each iteration a value of $T$. There are several ways to decrease the temperature $T$. In this work, we vary $T$ geometrically between a high starting temperature $T_\text{high}$ and a low temperature $T_\text{low}$. The algorithm terminates if $T$ reaches the threshold $T_\text{low}$ or falls below $T_\text{low}$. Another termination criterion is that the number of iterations exceeds $N_\text{sweeps}$.

\begin{algorithm}[htbp]
\caption{\label{alg:simann} Pseudo code for the simulated annealing method.}\label{algo:SimAnn}	
\begin{algorithmic}[1]
   	\State \textbf{INPUT:} Function $p$, initial state $\ov{x}$, temperatures $T_\text{high}$ and $T_\text{low}$
   	\State \textbf{INPUT:} Number of sweeps $N_\text{sweeps}$
   	\State \textbf{OUTPUT:} Approximation of $\ov{x}_{\min}$
   	\State $\textit{iter} = 1$
   	\State $\textit{stop} = \textit{false}$
   	\While{!stop}
   	   \State $T=\textit{coolingSchedule}(T_\text{high},T_\text{low},\textit{iter})$
   	   \State $\ov{x}_\text{new} = \textit{neighborOf}\left(\ov{x}\right)$
   	   \State $c_\text{new} = p\left(\ov{x}_\text{new}\right)$
   	   \If{$c_\text{new}<c$}
   	      \State $\ov{x}=\ov{x}_\text{new}$,\;$c=c_\text{new}$
   	   \ElsIf{$\textit{AcceptProb}\left(T,c_\text{new}-c\right)>\textit{random}\left(0,1\right)$}
   	      \State $\ov{x}=\ov{x}_\text{new}$,\;$c=c_\text{new}$
   	   \EndIf
   	   \State $\textit{iter} = \textit{iter}+1$ 
   	   \If{$T \leq T_\text{low}$ or $\textit{iter}>N_\text{sweeps}$}
   	      \State $\textit{stop} = \textit{true}$
   	   \EndIf 
   	\EndWhile
   	\State $\ov{x}_{\min} = \ov{x}$
   	\ \\
   	\Return $x_{\min}$
\end{algorithmic}			
\end{algorithm}
\ \\
Contrary to \emph{simulated annealing}, the \emph{quantum annealing} method is not inspired by laws from statistical physics, but rather by phenomena occurring in quantum mechanics. One of them is the observation that a quantum mechanical state attains another state by quantum fluctuations or tunneling through thin energy barriers. Emulating this behavior using a heuristic optimization method, we assign to each binary component $x_i \in \left\{0,1\right\}$ of a state $\ov{x} = \left(x_1,\ldots,x_n\right) \in C$ a qubit $\mathbf{q}_i$. This results in a register of $n$ qubits:
$R = \mathbf{q}_1\ldots \mathbf{q}_n$.
The quantum state for qubit $\mathbf{q}_i$ is denoted by $\ket{\psi_i}$, which is a superposition of the following states: $\left\{\ket{0},\ket{1}\right\}$. The quantum state of $R$ is given by $\ket{\Psi}$. Next, we modify $\ket{\Psi}$ such that the objective function $p$ in \eqref{app_eq:QUBO} can be represented by a quantum state. Therefore, we consider the following Hamiltonian:
$$
H_{p,\textit{an}} \ \ = \ \ \sum_{i=1}^n h_i \sigma_i^z \ + \ \sum_{i<j} J_{ij} \sigma_i^z\sigma_j^z,
$$
where $\sigma_i^z$ is the Pauli-$z$ operator applied to qubit $\mathbf{q}_i$, and $h_i$ and $J_{ij}$ are real numbers. Applying $H_{p,\textit{an}}$ to $\ket{\Psi}$ yields the following state:
$$
H_{p,an} \ket{\Psi} \ \ = \ \ \left(\sum_{i=1}^n h_i \sigma_i^z + \sum_{i<j} J_{ij} \sigma_i^z\sigma_j^z\right) \ket{\Psi} = E_{p,\textit{an}} \ket{\Psi} 
$$
and the corresponding energy:
$$
E_{p,\textit{an}} = \sum_{i=1}^n h_i \left(-1\right)^{\psi_i} + \sum_{i<j} J_{ij} \left(-1\right)^{\psi_i}\left(-1\right)^{\psi_j} = \sum_{i=1}^n h_i s_i + \sum_{i<j} J_{ij} s_is_j,
$$
with $\;s_i \in \left\{-1,1\right\}$ for $i\in\myset{1,\ldots,n}$. Using the transformation $s_i = 1-2 x_i$, we can reformulate $E_{p,\textit{an}}$ as follows:
\[
	E_{p,an} \ =  \ \sum_{i\leq j} q_{ij} x_ix_j + K \ \ \text{ with } \ \ K \ =  \ \sum_{i<j} J_{ij} + \sum_{i=1}^n h_i \ \ \text{ and }
\]
\begin{align}
	\label{eq:weights}
	q_{ij} &= \begin{cases} 4 J_{ij}, & \text{ if } i < j, \\
		-2 h_i - 2 \displaystyle{\sum_{k=1}^{i-1}}\; J_{ki} - 2 \displaystyle{\sum_{k=i+1}^n}\; J_{ik} , & \text{ if } i = j. \\
	\end{cases}
\end{align}
The parameters $q_{ij}$ are the entries of the matrix $Q$ defining the objective function $p$ in \eqref{app_eq:objective}.
Apparently, this expression is similar to the objective function of the QUBO \eqref{app_eq:QUBO}. Thus, computing the energy of a quantum state produced by $H_{p,\textit{an}}$ yields the minimal cost for a QUBO \eqref{app_eq:QUBO}. The quantum annealer transforms an initial quantum state that can be prepared in an easy way by an initial Hamiltonian $H_I$ into a quantum state governed
by the problem Hamiltonian $H_{p,\textit{an}}$. Thereby, the \emph{adiabatic theorem} is considered. The adiabatic theorem states that if we consider the quantum state $Q_I$ in its ground state (energy minimizing state) and if we transfer this state to the final quantum state $Q_{p,\textit{an}}$ sufficiently slowly, the process will end up in the ground state of $Q_{p,\textit{an}}$
with high probability \cite[Section 2.2.2]{mcgeoch2022adiabatic}. Thus, after high probability measurement, we obtain the ground state and its energy $E_{p,\textit{an}}$. This is made possible by quantum fluctuations or tunneling through energy barriers in the energy landscape of the final quantum state. Now, the crucial step consists in preparing the state $H_{p,\textit{an}} \ket{\Psi}$ and measuring its energy. Therefore, we consider an initial state produced by the Hamiltonian
$$
H_I \ \ = \ \ \sum_{i=1}^n \frac{1-\sigma_i^x}{2},
$$
where $\sigma_i^x$ is the Pauli-$x$ operator for qubit $\mathbf{q}_i$. Applying $H_I$ yields a superposition of the different quantum states in $R$ such that in the ground state of $H_I$ (energy minimizing state) every quantum state in $R$ is equally likely to be observed. Combining $H_I$ and $H_{p,\textit{an}}$, we obtain a Hamiltonian depending on a variable $s$ parameterized by $t$:
$$
H\left(s\left(t\right)\right) \ \ = \ \ A\left(s\left(t\right)\right) H_I \ + \ B\left(s\left(t\right)\right)H_{p,\textit{an}}.
$$
The function $s:\left[0,t_f\right] \rightarrow \left[0,1\right]$ is known as \emph{annealing path} and has to be specified by the user. In this work, we consider a linear function:
$$ \textstyle
s\left(t\right) \ \ = \ \ \frac{t}{t_f},
$$
where $t_f$ is denoted as \emph{annealing time}. $A$ and $B$ are called \emph{envelope functions}, having the following properties:
$$
A\left(0\right) \gg B\left(0\right) \quad\text{ and }\quad A\left(1\right) \ll B\left(1\right),
$$
which means that for $t=0$ the Hamiltonian is governed by $H_I$, while at the end of the annealing path for $t=t_f$ the Hamiltonian is governed by $H_{p,\textit{an}}$. In addition to that, $A$ decreases monotonically in $s$, while $B$ increases monotonically in $s$. Thus, as $t$ increases, $H\left(s\left(t\right)\right)$ approaches $H_{p,\textit{an}}$. Finally, the qubit values $\ket{\psi_i}$ are measured. Using 
\begin{equation}
	\label{eq:psix}
	\left(-1\right)^{\psi_i} \ \ = \ \ s_i \ \ = \ \ 2x_i-1,
\end{equation}
a minimizer $\ov{x}_{\min}$ of \eqref{app_eq:QUBO} and a minimal value of $p$ can be determined. 
The physical interpretation of this method reads as follows: Applying the initial Hamiltonian introduces energy to the annealing process in the form of quantum fluctuations.

Adapting $A$ and $B$ allows us to change between different minima of the energy landscape by tunneling through energy barriers that separate them. If the energy barriers are thin, the probability of tunneling increases. In contrast to simulated annealing, the probability governing the change between local minima is not affected by the difference between the minima.
Exploiting the tunneling phenomenon, we can move faster and further across the energy landscape compared to the simulated annealing method, since within a single iteration of the simulated annealing method, one has to ``climb'' over energy barriers according to the Boltzmann distribution instead of tunneling through the energy barriers. During the annealing process, the energy landscape is more and more adapted to the QUBO considered. If $A$ and $B$ are modified sufficiently slowly, the whole process will finish in an energy minimizing ground state with high probability due to the adiabatic theorem. Since $A$ and $B$ change gradually, this method can be assigned to the class of annealing based optimization methods. The whole method is summarized in Algorithm~\ref{alg:quantann}.
\begin{algorithm}[htb]
	\caption{\label{alg:quantann} Pseudo code for the quantum annealing method.}	
	\begin{algorithmic}[1]
		\State \textbf{INPUT:} Parameters $q_{ij}$, annealing time $t_f$
		\State \textbf{INPUT:} Number of sweeps $N_\text{sweeps}$
		\State \textbf{OUTPUT:} Approximation of $\ov{x}_{\min}$
		\State $\textit{iter}=0$, $\textit{stop}=\textit{false}$
		\State $\Delta t = \frac{t_f}{N_\text{sweeps}}$
		\State Compute the weights $h_i$ and $J_{ij}$ using \eqref{eq:weights} to obtain $H_{p,\textit{an}}$
		\State Apply $H_I$ scaled by $A\left(0\right)$ to obtain a quantum state $\ket{\Psi}$
		\For{$\textit{iter}<N_\text{sweeps}$}
		\State $t= \textit{iter} \cdot \Delta t$, $s\left(t\right)=\frac{t}{t_f}$
		\State Determine $A\left( s\left(t\right)\right)$ and $B\left( s\left(t\right)\right)$
		\State Determine $H\left(s\left(t\right)\right)$ and apply it to $\ket{\Psi}$
		\State $\textit{iter} = \textit{iter}+1$ 
		\EndFor
		\State Measure $\ket{\Psi}$ to obtain the quantum states $\Psi_i$ of the qubits $\mathbf{q}_i$
		\State Use Equation \eqref{eq:psix} to compute the binary values $x_i$
		\State $\ov{x}_{\min} = \ov{x} = \left(x_1,\ldots,x_n\right)$
		\ \\
		\Return $\ov{x}_{\min}$		
	\end{algorithmic}
\end{algorithm}

%% file: appendix_theory.tex
\section{Details omitted in Section~\ref{sec:theory}}\label{appendix:theory}

The following can be easily verified:
\begin{fact}\label{fact:unique_new}
	For all $\ov{x}\in\{0,1\}^{B_2\times A_1}$ we have:
	\begin{enumerate}[(a)]
		\item $\punique(\ov{x})\in \NN$, and \ \
		\item $\punique(\ov{x}) = 0$ $\iff$ for every $i\in B_2$ we have $\sum_{j\in A_1} x_{i,j}\leq1$ (i.e.,
                  there is at most one $j\in A_1$ with $x_{i,j}=1$).
	\end{enumerate}	
\end{fact}
Clearly, when given an  $\ov{x}\in\{0,1\}^{B_2\times A_1}$ with $\punique(\ov{x})\geq 1$, then~$\ov{x}$ does \emph{not} represent any function $h\colon B_2\to A_1$.

\begin{proof}[Proof of Lemma~\ref{lemma:p_v_new}] \ \\
  (a) is obvious.
  For (b) consider an arbitrary assignment $\ov{x}\in\myset{0,1}^{B_2\times A_1}$ with $\punique(\ov{x})=0$. From
  Fact~\ref{fact:unique_new} we know that for every $i\in B_2$ there is at most one $j\in A_1$ with $x_{i,j}=1$.
  Consider the extended bit-matrix $\ovxext$ associated with $\ov{x}$. 

  First of all, for contradiction assume that $p_{R,u}(\ov{x})\not\in\myset{-1,0}$. Then there must exist two distinct tuples $w=(w_1,\ldots,w_r)$ and $w'=(w'_1,\ldots,w'_r)$ in $\TableauOne(R)$ such that
  \[
    \prod_{k=1}^{r}\xext{u_k,w_k}=1
    \quad\text{and}\quad
    \prod_{k=1}^{r}\xext{u_k,w'_k}=1.
  \]  
  Since $w\neq w'$, there is a $k\leq r$ such that $w_k\neq w'_k$. In particular, $\xext{u_k,w_k}=1=\xext{u_k,w'_k}$. By definition of $\ovxext$ this implies that $u_k\in B_2$ and $x_{u_k,w_k}=1=x_{u_k,w'_k}$. But this means that for $i\deff u_k\in B_2$ there are two distinct values $j\deff w_k$ and $j'\deff w'_k$ with $x_{i,j}=1$ and $x_{i,j'}=1$, contradicting Fact~\ref{fact:unique_new}. This proves that $p_{R,u}(\ov{x})\in\{-1,0\}$.

  Concerning the second statement of (b), we first prove the direction ``$\Longleftarrow$''. Let us assume that the condition to the right of the symbol ``$\Longleftrightarrow$'' is satisfied. Then, for every $k\leq r$ we have $\xext{u_k,z_k}=1$ and $\xext{u_k,z}=0$ for all $z\in A_1\setminus\myset{z_k}$. Hence, $\prod_{k=1}^{r}\xext{u_k,z_k}=1$ and for every $(w_1,\ldots,w_r)\in \TableauOne(R)$ with $(w_1,\ldots,w_r)\neq (z_1,\ldots,z_r)$ we have $\prod_{k=1}^{r}\xext{u_k,w_k}=0$. Thus, $p_{R,u}(\ov{x})=-1$. This completes the proof of the direction ``$\Longleftarrow$''.

  For proving the direction ``$\Longrightarrow$'', let us assume that $p_{R,u}(\ov{x})=-1$.
  Thus, there is a tuple $(z_1,\ldots,z_r)\in \TableauOne(R)$ such that $\prod_{k=1}^r \xext{u_k,z_k}=1$, i.e., $\xext{u_k,z_k}=1$ for every $k\leq r$. All that remains to be done is to show that for every $k\leq r$ we we have
  $\sum_{z\in A_2} \xext{u_k,z}=1$. Consider an arbitrary $k\leq r$. If $u_k\in B_2$ then $\xext{u_k,z}=x_{u_k,z}$ for all $z\in A_1$, and hence from $\punique(\ov{x})=0$ and Fact~\ref{fact:unique_new}(b) and $x_{u_k,z_k}=1$ we obtain that
  $\sum_{z\in A_1}\xext{u_k,z}=1$.
  On the other hand, if $u_k\in A_2\setminus B_2$, then by definition of the extended bit-matrix $\ovxext$ we obtain that   $\sum_{z\in A_1}\xext{u_k,z}\leq 1$. Combining this with the fact that $\xext{u_k,z_k}=1$ and $z_k\in A_1$ yields
  $\sum_{z\in A_1}\xext{u_k,z}= 1$.
This completes the proof of Lemma~\ref{lemma:p_v_new}.
\end{proof}  

\begin{proof}[Proof of Lemma~\ref{lemma:p_ac_new}] \ \\
  (a) immediately follows from  Lemma~\ref{lemma:p_v_new}(a) and the definition of $\pac(\ov{x})$.
  \\
  For (b) consider an arbitrary $\ov{x}\in\myset{0,1}^{B_2\times A_1}$ with $\punique(\ov{x})=0$.
  From Lemma~\ref{lemma:p_v_new}(b) we obtain that $p_{R,u}(\ov{x})\in\myset{-1,0}$ for all $R\in\dbschema$ and  $u\in \TableauTwo(R)$. Thus, according to the definition of $\pac(\ov{x})$ we obtain:
  $\pac(\ov{x})\geq -\sum_{R\in\dbschema} \setsize{\TableauTwo(R)} = -\setsize{\TableauTwo}$.
  This proves the first statement of (b).

  Concerning the second statement of (b), we first prove the direction ``$\Longrightarrow$''. By assumption,
  $\pac(\ov{x})=-\setsize{\TableauTwo}=-\sum_{R\in\dbschema}\setsize{\TableauTwo(R)}$.
  From the definition of $\pac(\ov{x})$ and the first statement of Lemma~\ref{lemma:p_v_new}(b) we obtain that
  $p_{R,u}(\ov{x})=-1$ for all $R\in\dbschema$ and all $u\in\TableauTwo(R)$. Applying the second statement of Lemma~\ref{lemma:p_v_new}(b) to all $R\in\dbschema$ and all $u\in\TableauTwo(R)$, we obtain that 
  for all $i\in\Vars(\TableauTwo)\cup\Cons(\TableauTwo)$ we have $\sum_{j\in A_1}\xext{i,j}=1$.
  Note that $A_2=\Vars(\TableauTwo)\cup\Cons(\TableauTwo)\cup\Cons(\AnswerTupleTwo)$. Furthermore, since we did not already stop during the preparation step, $\Cons(\AnswerTupleTwo)\subseteq\Cons(\AnswerTupleOne)\subseteq A_1$, and hence according to
the definition of the extended bit-matrix $\ovxext$ we have $\sum_{j\in A_1}\xext{i,j}= 1$ for every $i\in\Cons(\AnswerTupleTwo)$.
  In summary, this yields that $\ovxext$ represents the mapping $h\colon A_2\to A_1$ defined for every $i\in A_2$ by letting $h(i)=j_i$ where $j_i$ is the uniquely determined $j\in A_1$ with $\xext{i,j}=1$.
  By applying the second statement of  Lemma~\ref{lemma:p_v_new}(b) to all $R\in\dbschema$ and all $u\in\TableauTwo(R)$, we furthermore obtain that $(h(u_1),\ldots,h(u_r))\in \TableauOne(R)$ for all $R\in\dbschema$ and all $u=(u_1,\ldots,u_r)\in\TableauTwo(R)$.
  This proves that condition \eqref{itemThree:def:homom} of Definition~\ref{def:HomomorphismForQueries} is satisfied. As conditions \eqref{itemOne:def:homom} and \eqref{itemTwo:def:homom} are satisfied by the definition of $\ovxext$ and the fact that we did not stop already during the preparation step, the mapping $h$ is a homomorphism from $\queryTwo$ to $\queryOne$. This completes the proof of the direction
  ``$\Longrightarrow$''.

  Let us now turn to the direction ``$\Longleftarrow$''. We assume that the
  condition to the right of the symbol ``$\Longleftrightarrow$'' is satisfied. Then, for every $i\in A_2$ there is a $j_i\in A_1$ with $\xext{i,j_i}=1$ and $\xext{i,j}=0$ for all $j\in A_1\setminus\myset{j_i}$; and the mapping $h\colon A_2\to A_1$ defined via
  $h(i)=j_i$ for all $i\in A_2$ is a homomorphism from $\queryTwo$ to $\queryOne$. This implies that for every $R\in\dbschema$ and every $u\in\TableauTwo(R)$ the condition to the right of the symbol ``$\Longleftrightarrow$'' of Lemma~\ref{lemma:p_v_new}(b) is satisfied. Hence, from Lemma~\ref{lemma:p_v_new}(b) we obtain that
  $p_{R,u}(\ov{x})=-1$ for every $R\in\dbschema$ and every $u\in\TableauTwo(R)$. According to the definition of
  $\pac(\ov{x})$, we hence obtain that $\pac(\ov{x})=\sum_{R\in\dbschema}\sum_{u\in\TableauTwo(R)}-1 = -\setsize{\TableauTwo}$. This completes the proof of Lemma~\ref{lemma:p_ac_new}.
\end{proof}

\begin{proof}[Proof of Theorem~\ref{thm:correctnessOfGenericChoice}] \ \\
  Consider an arbitrary bit-matrix $\ov{x}\in\Cgen$.
  From Fact~\ref{fact:unique_new}(a) we know that $\punique(\ov{x})\in\NN$.
  Using Lemma~\ref{lemma:p_ac_new}(a) we obtain that
  $-\setsize{\TableauOne}{\cdot}\setsize{\TableauTwo}
  \leq
  -\sum_{R\in\dbschema}\setsize{\TableauOne(R)}{\cdot}\setsize{\TableauTwo(R)}
  \leq
  \pac(\ov{x})
  \leq 0$.

  If $\punique(\ov{x})\geq 1$, then
  $\pgen(\ov{x})\geq \pac(\ov{x})+(\setsize{\TableauOne}{\cdot}\setsize{\TableauTwo}+1) \geq 1$.
  If $\punique(\ov{x})=0$, then $\pgen(\ov{x})=\pac(\ov{x})$, and from Lemma~\ref{lemma:p_ac_new}(b) we obtain that
  $\dgen\leq \pac(\ov{x})\leq 0$.
  This, in particular, proves the theorem's first statement.

  Let us now turn to the direction ``$\Longleftarrow$'' of the theorem's second statement. I.e., we assume that $\queryOne\querycont\queryTwo$, and our aim is to show that $\min_{\ov{x}\in \Cgen}\pgen(\ov{x})=\dgen$. Above, we have already shown that $\min_{\ov{x}\in \Cgen}\pgen(\ov{x})\geq \dgen$.
  From Theorem~\ref{thm:ChandraMerlin} and $\queryOne\querycont\queryTwo$ we obtain that there exists a homomorphism $h$ from $\queryTwo$ to $\queryOne$. Let $\ov{x}\in\myset{0,1}^{B_2\times A_1}=\Cgen$ be the bit-matrix associated with the mapping $h$, and let $\ovxext$ be the corresponding extended bit-matrix. By definition we have
  $\punique(\ov{x})=0$, and from Lemma~\ref{lemma:p_ac_new}(b) we obtain that $\pac(\ov{x})=\dgen$. Considering the definition of the polynomial $\pgen$ and using the fact that $\punique(\ov{x})=0$, we obtain that $\pgen(\ov{x})=\pac(\ov{x})=\dgen$. This proves that $\min_{\ov{x}\in \Cgen}\pgen(\ov{x})=\dgen$, i.e., it proves the  direction ``$\Longleftarrow$'' of the theorem's second statement.

  Let us now turn to the theorem's third statement and to the  direction ``$\Longrightarrow$'' of the theorem's second statement. We have already proved the theorem's first statement, and thus $\min_{\ov{x}\in \Cgen}\pgen(\ov{x})\geq \dgen$. Consider an $\ov{x}\in \Cgen$ with $\pgen(\ov{x})=\dgen$.
  Since $\dgen\leq 0$, we know that $\punique(\ov{x})=0$ (because otherwise we would have $\punique(\ov{x})\geq 1$ and $\pgen(\ov{x})\geq 1$).
Hence, $\pgen(\ov{x})=\pac(\ov{x})$.
  Using that $\dgen=-\setsize{\TableauTwo}$, we obtain from Lemma~\ref{lemma:p_ac_new}(b) that the extended bit-matrix $\ovxext$ associated with $\ov{x}$ represents a homomorphism from $\queryTwo$ to $\queryOne$.
  From Theorem~\ref{thm:ChandraMerlin} we obtain that $\queryOne\querycont\queryTwo$.
  This completes the proof of Theorem~\ref{thm:correctnessOfGenericChoice}.
\end{proof}

\subsection*{Details on how it can happen that polynomial~$p$ does not contain any binary variable:}

If the polynomial $p$ actually does not contain any binary variable,
then $p(\ov{x})$ is the constant function $\ell$, for some integer
$\ell$ (for short: $p(\ov{x})\equiv \ell$).
This happens
\begin{enumerate}[$\bullet$]
  \item
for $\pcgen$, if for all $R\in\dbschema$, all
$u=(u_1,\ldots,u_r)\in\TableauTwo(R)$, and all $w=(w_1,\ldots,w_r)\in
\TableauOne(R)$ there either is a $k\leq r$ with $u_k\in A_2\setminus
B_2$ and $\xext{u_k,w_k}=0$, or $u_k\in A_2\setminus B_2$ for all $k\leq r$,
\item
  for $\pcsimpl$, if for all $R\in\dbschema$, all
$u=(u_1,\ldots,u_r)\in\TableauTwo(R)$, and all $w=(w_1,\ldots,w_r)\in
\TableauOne(R)$ there either is a $k\leq r$ with $u_k\in A_2\setminus
\simplBtwo$ and $\xsim{u_k,w_k}=0$, or $u_k\in A_2\setminus \simplBtwo$ for all $k\leq r$,
\item  
  for $\pgen$, if $B_2=\emptyset$, and
\item
  for $\psimpl$, if $\simplBtwo=\emptyset$.
\end{enumerate}

%% file: appendix_system.tex
\section{Details omitted in Section~\ref{sec:system}}\label{appendix:system}

Figure~\ref{fig:architecture} shows the backends of our system architecture required for the quantum stage in Figure~\ref{fig:workflow}. This includes software libraries, simulator software running on classical hardware, as well as quantum hardware accessed as cloud services. 

\begin{figure}[thb]
    \centering

    \begin{subfigure}[t]{0.48\columnwidth}
        \centering
        \includegraphics[width=\linewidth]{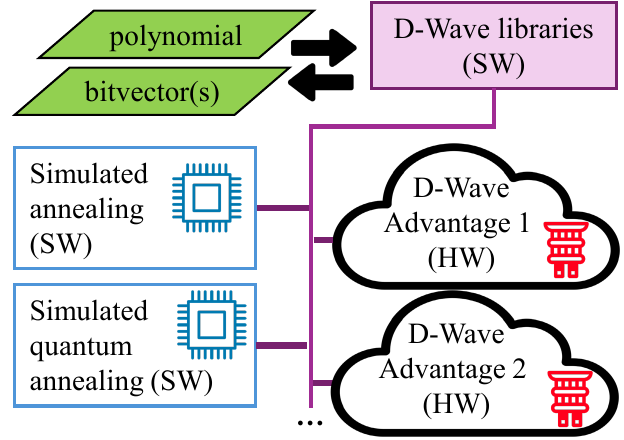}
        \caption{Annealing}
        \label{fig:solvers_annealing}
    \end{subfigure}
    \hfill
    \begin{subfigure}[t]{0.48\columnwidth}
        \centering
        \includegraphics[width=\linewidth]{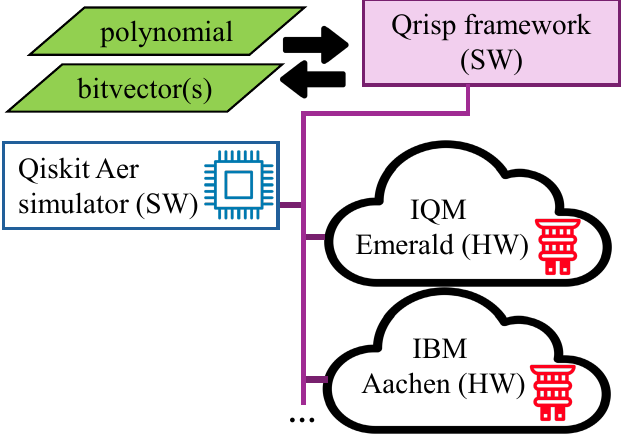}
        \caption{QAOA}
        \label{fig:solvers_qaoa}
    \end{subfigure}

\caption{Architecture for (a)~annealing and (b)~QAOA. Software libraries take polynomials and return bitvectors. Simulators (``SW''/blue icons) run on classical hardware,  quantum hardware (``HW''/red icons) accessed as cloud services.}
\label{fig:architecture}
\end{figure}

%% file: appendix_generated_families.tex
\section{Details omitted in Section~\ref{subsec:ExpSetup}}
\label{app:families}

Figure~\ref{fig:archetypes} visualizes building blocks for our generated query families.

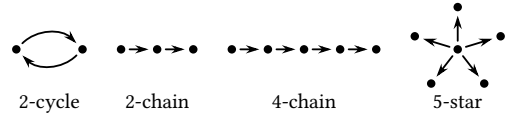
\begin{figure}[th]
	\centering
	\tikzset{
		graphnode/.style={circle, fill=black, draw=none, inner sep=0pt, minimum size=3.2pt},
		graphedge/.style={line width=0.7pt, -{Stealth[length=1.8mm,width=1.1mm]},
			shorten <= 0.5mm, shorten >= 0.5mm, line cap=round}
	}
	\begin{tabular}{@{}c@{\hspace{4mm}}c@{\hspace{4mm}}c@{\hspace{4mm}}c@{}}
		\begin{tikzpicture}[scale=0.8, baseline=(ref.base)]
			\coordinate (ref) at (0,0);
			\node[graphnode] (a) at (-0.55,0) {};
			\node[graphnode] (b) at ( 0.55,0) {};
			\draw[graphedge] (a) to[out=45, in=135, looseness=1.05] (b);
			\draw[graphedge] (b) to[out=-135, in=-45, looseness=1.05] (a);
		\end{tikzpicture}
		&
		\begin{tikzpicture}[scale=0.8, baseline=(ref.base)]
			\coordinate (ref) at (0,0);
			\node[graphnode] (u) at (-0.60,0) {};
			\node[graphnode] (v) at ( 0.00,0) {};
			\node[graphnode] (w) at ( 0.60,0) {};
			\draw[graphedge] (u) -- (v);
			\draw[graphedge] (v) -- (w);
		\end{tikzpicture}
		&
		\begin{tikzpicture}[scale=0.8, baseline=(ref.base)]
			\coordinate (ref) at (0,0);
			\foreach \i/\x in {1/-1.20,2/-0.60,3/0,4/0.60,5/1.20}{
				\node[graphnode] (c\i) at (\x,0) {};
			}
			\foreach \u/\v in {1/2,2/3,3/4,4/5}{
				\draw[graphedge] (c\u) -- (c\v);
			}
		\end{tikzpicture}
		&
		\begin{tikzpicture}[xscale=0.8, yscale=0.75, baseline=(ref.base)]
			\coordinate (ref) at (0,0);
			\node[graphnode] (s0) at (0,0) {};
			\foreach \ang in {90,162,234,306,18}{
				\node[graphnode] (s\ang) at (\ang:0.75) {};
				\draw[graphedge] (s0) -- (s\ang);
			}
		\end{tikzpicture}\\[2pt]
		\small 2-cycle & \small 2-chain & \small 4-chain & \small 5-star
	\end{tabular}
	\caption{Building blocks for parameterized QC problems.}

	\label{fig:archetypes}
\end{figure}

Both problem families have a configurable parameter~$i$ to scale up the problem size by adapting the size of the second query. In particular, we defined the problem families as follows.

\begin{definition}[2-cycle to $i$-chain]\label{def:TwoCyToiCh}
	\emph{2-cycle to $i$-chain} is defined as the family of query pairs
        $(q_{\textit{2cy}},q_{\textit{i-ch}})$ with $q_{\textit{2cy}}$
        being defined as in
        Example~\ref{example:preparation:cycle_chain} except for
renaming the variables from $\vvarz,\vvarz'$ to $\vvarz_0,\vvarz_1$, i.e.,
        choosing $\Tableau_{\textit{2cy}}(E)=\myset{(\vvarz_0,\vvarz_1),\allowbreak (\vvarz_1,\vvarz_0)}$ (this change in notation will come in handy later) and $q_{\textit{i-ch}}\deff(\Tableau_{\textit{i-ch}},\AnswerTuple_{\textit{i-ch}})$ being defined via $\AnswerTuple_{\textit{i-ch}}=\emptytuple$ and
	$\Tableau_{\textit{i-ch}}(E)=\setc{(\vvary_{k-1},\vvary_{k})}{1\leq
          k\leq i}$, for $i\in\NNpos$.
\end{definition}

\begin{definition}[2-chain to $i$-star]
	\emph{2-chain to $i$-star} is defined as the family of query
        pairs $(q_{\textit{2ch}},q_{\textit{i-st}})$ with
        $q_{\textit{2ch}}$ being defined as in
Example~\ref{example:preparation:cycle_chain}
        and $q_{\textit{i-st}}\deff(\Tableau_{\textit{i-st}},\AnswerTuple_{\textit{i-st}})$ being defined via $\AnswerTuple_{\textit{i-st}}=\emptytuple$ and
	$\Tableau_{\textit{i-st}}(E)=\setc{(\vvary_0,\vvary_{k})}{1\leq
          k\leq i}$, for $i\in\NNpos$.
\end{definition}

The number of nodes in the first query determines the stepsize by which we can scale up the number of binary variables. By choosing a query with two nodes (2-cycle) and a query with 3 nodes (2-chain) as the first query, we ensure that they can be scaled up evenly and in small steps, without requiring too many additional binary variables. This is particularly important for experiments on real quantum hardware, which is still very limited in its physical resources.

As mentioned earlier, the queries were carefully crafted to not be simplified in either the preparation or the simplification step. The former is achieved by treating our queries as Boolean queries. Otherwise, for example, fixing the endpoints of the regarded graph queries would lead to trivial query pairs. The latter is achieved by having more than one edge in each first query to avoid any trivial mappings.

For graph queries, certifying the query containment is equivalent to finding a homomorphism between the two represented graphs. In each generated problem instance, two distinct homomorphisms can be found, serving as a ground truth for our examples. In the following, we describe the generation of the ground truth and the problem formula for each family.

\paragraph{2-cycle to $i$-chain}
We can find a homomorphism from an $i$-chain to a 2-cycle by
\enquote{wrapping} the chain around the 2-cycle. Depending on where we
start this (left or right node of the 2-cycle), we obtain two
different possibilities. This provides two distinct homomorphisms
that solve the containment problem for every $i\in\NN_{\geq 1}$. More
formally, the first homomorphism $h$ is given via
$h(\vvary_k)=\vvarz_{0}$ if $k$ is even and $h(\vvary_k)=\vvarz_{1}$
if $k$ is odd (for $0\leq k\leq i$). The second homomorphism $h'$ is
defined the other way round, via $h'(\vvary_k)=\vvarz_{1}$ if $k$ is
even and $h'(\vvary_k)=\vvarz_{0}$ if $k$ is odd (for
$0\leq k\leq i$). Each of these two individual solutions can be specified
by the mapping of the first vertex $\vvary_0$ of the $i$-chain to
either $\vvarz_0$ or $\vvarz_1$ of the 2-cycle, because this uniquely
determines the rest of the homomorphism;
mapping the other vertices
follows accordingly. See Figure~\ref{fig:chain-cycle-graph} for a
visualization, where the two mappings are indicated in blue and
yellow. 
Since we can find a homomorphism from the $i$-chain to the 2-cycle, the latter is contained in the former (but not vice versa).
The ground truth is given by these two homomorphisms. Their 
bit-matrix representations have $i{+}1$ rows $\vvary_0,\ldots,\vvary_i$ and $2$
columns $\vvarz_0,\vvarz_1$;
for any two consecutive rows, one is of the form 10 and the other is
of the form 01; and the bit-matrix for $h$ starts with row 10, while the
bit-matrix for $h'$ starts with row 01.
Each of these two bit-matrices is
represented by the
bitvector obtained by concatenating the rows of the matrix, i.e., $h$
is represented by the
bitvector $10011001\cdots$,  and $h'$ is represented by the bitvector $01100110\cdots$.

However, finding such a homomorphism algorithmically via
minimizing the associated polynomial $p(\ov{x})$
is not trivial, since the polynomial contains $2(i{+}1)$ binary
variables ($i{+}1$ nodes from the $i$-chain each can be mapped to 2
possible nodes in the 2-cycle), leading to a search space size of
$2^{2(i+1)}$ in the unconstrained setting and $2^{i+1}$ in the
constrained setting. Therefore, this problem combines an exponentially
large search space (expensive) with small steps of 2 in the number of
qubits (scalable).
Concerning the resulting problem polynomials, note that
\begin{enumerate}[$\bullet$]  
\item
$\pac(\ov{x})=- \sum_{k=1}^i \big( x_{\vvary_{k-1},\vvarz_0} \cdot
x_{\vvary_{k},\vvarz_1} \ + \  x_{\vvary_{k-1},\vvarz_1} \cdot
x_{\vvary_{k},\vvarz_0}\big)$,
\item
$\punique(\ov{x}) = \sum_{k=0}^i \big( x_{\vvary_k,\vvarz_0}\cdot x_{\vvary_k,\vvarz_1}\big)$,
\item
$\setsize{\Tableau_{\textit{2cy}}}{\cdot}\setsize{\Tableau_{\textit{$i$-ch}}}
+1 = 2i{+}1$, and
\end{enumerate}  
thus, the resulting polynomial for \emph{2-cycle to $i$-chain}
has the form:
\[
p_{\textit{2cy\_$i$-ch}}(\ov{x}) \ \ = \  \
\pac(\ov{x}) \hspace{-3mm}\underbrace{+ \ \  \ (2i{+}1) \cdot
  \punique(\ov{x})\ .}_{\text{only present in unconstr.\ setting}}
\]  

It can be verified that this polynomial does not get reduced by
simplification, and every binary variable actually occurs in the polynomial. Moreover, the polynomial is of degree 2, thus the problem is in QUBO form, which allows to compare annealing and QAOA.

\paragraph{2-chain to $i$-star}
In order to find a homomorphism from an $i$-star to a 2-chain, a
necessary condition is to map the central node of the star to a node
with an outgoing edge. There are two possibilities to do so, which
lead to two homomorphisms $h$ and $h'$, given via
$h(\vvary_0)=\vvarz_0$ and $h(\vvary_k)=\vvarz_1$ or
$h'(\vvary_0)=\vvarz_1$ and $h'(\vvary_k)=\vvarz_2$ for all $k$ with
$1\leq k\leq i$.
The bit-matrix representations of $h$ and $h'$ have $i{+}1$ rows
$\vvary_0,\ldots,\vvary_i$ and $3$ columns
$\vvarz_0,\vvarz_1,\vvarz_2$. Concerning $h$, the row for $\vvary_0$
is of the form 100, and all further rows are of the form
010. Concerning $h'$,  row $\vvary_0$ has the form 010, and all
further rows are of the form 001. Concatenating the rows of each
bit-matrix yields their bitvector representations $100010010\cdots{}010$ and $010001001\cdots{}001$.
Importantly, mapping the central node $\vvary_0$ already determines the solution. This is illustrated in Figure~\ref{fig:star-chain-graph}, where the two mappings are indicated in blue and yellow.

Similarly to the other family, we chose a problem with a small-stepped
number of qubits, in this case $3(i+1)$, and a reasonably large search
space of $2^{3(i+1)}$
solution candidates in the unconstrained setting and $3^{i+1}$ candidates in the constrained setting. Furthermore, this problem family also leads to polynomials that are of degree 2 and are not affected by simplification. This enables us to perform a thorough scalability analysis and ensures comparability of the two algorithmic approaches.

Concerning the resulting problem polynomials, note that
\begin{enumerate}[$\bullet$]  
\item
$\pac(\ov{x})=- \sum_{k=1}^i \big( x_{\vvary_{0},\vvarz_0} \cdot
x_{\vvary_{k},\vvarz_1} \ + \  x_{\vvary_{0},\vvarz_1} \cdot
x_{\vvary_{k},\vvarz_2}\big)$,
\item
  $\punique(\ov{x}) = \sum_{k=0}^i \big(
  x_{\vvary_k,\vvarz_0}\cdot x_{\vvary_k,\vvarz_1}
  \ + \
  x_{\vvary_k,\vvarz_0}\cdot x_{\vvary_k,\vvarz_2}
  \ + \
  x_{\vvary_k,\vvarz_1}\cdot x_{\vvary_k,\vvarz_2}
  \big)$,
\item
$\setsize{\Tableau_{\textit{2ch}}}{\cdot}\setsize{\Tableau_{\textit{$i$-st}}}
+1 = 2i{+}1$, and
\end{enumerate}  
thus, the resulting polynomial for \emph{2-chain to $i$-star}
has the form:
\[
p_{\textit{2ch\_$i$-st}}(\ov{x}) \ \ = \  \
\pac(\ov{x}) \hspace{-3mm}\underbrace{+ \ \  \ (2i{+}1) \cdot
  \punique(\ov{x})\ .}_{\text{only present in unconstr.\ setting}}
\]

\begin{figure}[th]
	\centering
	\begin{subfigure}[th]{.49\linewidth}
		\includegraphics[scale=.25]{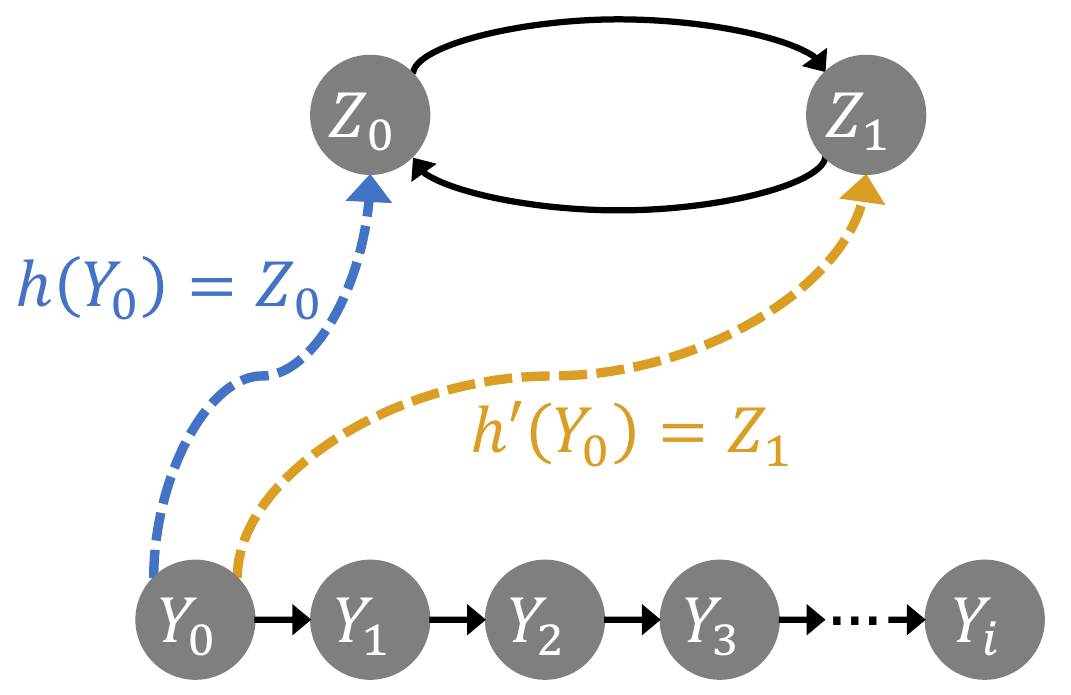}
		\caption{From $i$-chain to 2-cycle.}
		\label{fig:chain-cycle-graph}
	\end{subfigure}
	\hfill
	\begin{subfigure}[th]{.45\linewidth}
		\includegraphics[scale=.25]{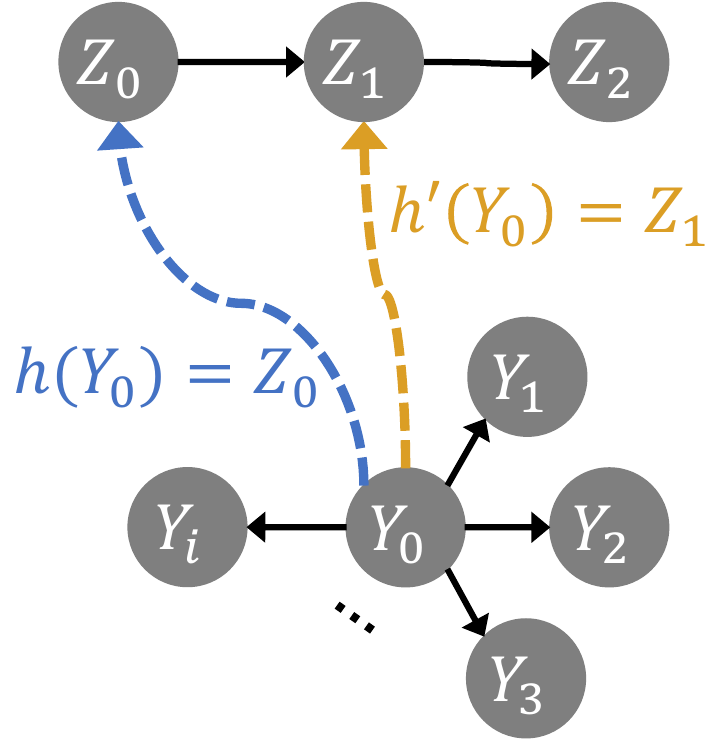}
		\caption{From $i$-star to 2-chain.}
		\label{fig:star-chain-graph}
	\end{subfigure}
	\caption{Two distinct homomorphisms, determined by blue/yellow mapping of $\vvary_0$, certify query containment.}
	\label{fig:problem_families}

\end{figure}

%% file: appendix_energy_landscapes.tex
\section{Details omitted in Section~\ref{sec:investigations}}
\label{app:energy_landscapes}

The actual performance of the algorithms considered in this paper depends on the individual problem structure.
While they have a certain expected performance on average, problem families can be constructed to precisely fit or not fit the nature of the algorithm. Such differences can become more apparent when artificial examples are considered. The course of our algorithms is influenced by the
energy landscape
that is explored during the run.
Recall that we represent a query pair $(\queryOne,\queryTwo)$, for which we want to decide whether $\queryOne\querycont\queryTwo$ holds, by a polynomial $p(\ov{x})$.
Every possible value of $p(\ov{x})$ for an assignment of the binary variables $\ov{x}$ with 0-1-values, can be regarded as an \enquote{energy level} or \enquote{value} of the objective function~$p(\ov{x})$. Thus, every possible solution candidate (i.e.,\ assignment of $\ov{x}$ with 0-1-values; henceforth they will be called \enquote{bitvectors},
\enquote{states} or \enquote{candidates}) corresponds to a point on the energy landscape.

The overall goal is to reach the global minimum by an annealing-inspired process. As explained in Section~\ref{sec:SAQA} and Section~\ref{app:SAQA_appendix}, the simulated annealing algorithm strives for states of lower or similar energy while sometimes (with exponentially decreasing probability) exploring neighboring states of higher energy than the current one (see the definition of $\textit{AcceptProb}$
in
Equation~\eqref{eq:Boltzmannprob} of
Section~\ref{app:SAQA_appendix}).

As we will see, the experimental results can be explained by differences in the distribution and connectedness of energy levels, so let us take a closer look at the corresponding energy landscapes. We first give an overview of the visible macro-effects and follow with a more detailed computation of the individual relevant sizes.

\paragraph{Overview}
Here, we focus on $(\queryOne,\queryTwo)$ being either the family of query pairs \emph{2-cycle to $i$-chain} or the family of query pairs \emph{2-chain to $i$-star}, for $i\in\NNpos$.
Our goal is to gain a better understanding of the observed scaling behavior depicted in Figure~\ref{fig:scaling_behavior}.
For both families, the energy levels can be partitioned into a very large portion of positive-valued states representing invalid
candidates, where at least one row of the represented matrix has more than one bit set to 1 (pos states), a number of zero-valued states representing incorrect mappings with respect to the tuples in $q_2$ (zero states), and an amount of negative-valued states representing at least partially correct mappings (neg states). The relative size of pos states approaches 100\% as $i$ grows, and there is an increasing energy gap of $2i{+}1$ to the zero and neg states, which is defined by the chosen \enquote{penalty weight}
$\setsize{\TableauOne}{\cdot}\setsize{\TableauTwo}+1$.
For the small portion of zero and neg states, it is worth investigating the development of their proportions towards each other. While both families start off with a relatively large zero plateau, its fraction declines towards 0 for  \emph{2-cycle  to $i$-chain} and approaches a 50/50 balance with the neg states for \emph{2-chain to $i$-star}. A summary of the development of all ratios is given in Table~\ref{tab:relative-sizes}. In order to interpret these ratios with respect to the experimental results, we also need to investigate the connectedness of the states.

\begin{table}[ht]
	\caption{Relative sizes of energy levels}
	\label{tab:relative-sizes}
	\centering
	\small
	\begin{subtable}{\linewidth}%
		\caption{2-cycle to $i$-chain}
		\centering
		\begin{tabular}{rrrrr}
			\toprule
			$i$ & \# bin. vars & $p_i(\text{pos})$ & $p_i(\text{neg}|\neg\text{pos})$ & $p_i(\text{0}|\neg\text{pos})$ \\
			\midrule
			1 & 4 & $\sim 44\%$ & $\sim 22\%$ & $\sim 78\%$ \\
			2 & 6 & $\sim 58\%$ & $\sim 37\%$ & $\sim 63\%$ \\
			3 & 8 & $\sim 68\%$ & $\sim 49\%$ & $\sim 51\%$ \\
			4 & 10 & $\sim 76\%$ & $\sim 59\%$ & $\sim 41\%$ \\
			\vdots & \vdots & \vdots & \vdots & \vdots \\
			$\infty$ & $\infty$ & 100\% & 100\% & 0\% \\
			\bottomrule
		\end{tabular}
		\label{tab:relative-sizes-chain-family}
	\end{subtable}
	\newline\vspace{3ex}\newline
	\begin{subtable}{\linewidth}%
		\caption{2-chain to $i$-star}
		\centering
		\begin{tabular}{rrrrr}
			\toprule
			$i$ & \# bin. vars & $p_i(\text{pos})$ & $p_i(\text{neg}|\neg\text{pos})$ & $p_i(\text{0}|\neg\text{pos})$ \\
			\midrule
			1 & 6 & $75\%$ & $\sim 13\%$ & $\sim 87\%$ \\
			2 & 9 & $\sim 88\%$ & $\sim 22\%$ & $\sim 78\%$ \\
			3 & 12 & $\sim 94\%$ & $\sim 29\%$ & $\sim 71\%$ \\
			4 & 15 & $\sim 97\%$ & $\sim 34\%$ & $\sim 66\%$ \\
			\vdots & \vdots & \vdots & \vdots & \vdots \\
			$\infty$ & $\infty$ & 100\% & 50\% & 50\% \\
			\bottomrule
		\end{tabular}
		\label{tab:relative-sizes-star-family}
	\end{subtable}
\end{table}

Figure~\ref{fig:energy_landscapes}
illustrates the connectedness of the energy landscapes. Each bubble represents a (set of) candidate(s) with the height indicating the corresponding energy level. The two lower-most states are the two distinct optimal solutions with their color emphasizing the distinct mapping of $\vvary_0$ to either $\vvarz_0$ (blue) or $\vvarz_1$ (yellow). An arrow between two bubbles indicates a direct connection between them, which means that they contain neighboring states. The arrows point towards the lower energy level, but there is also a low probability that the simulated annealing algorithm jumps in the reverse direction. 
\begin{figure}[th]
	\centering
	\begin{subfigure}[t]{0.48\textwidth}
		\centering
		\includegraphics[width=\textwidth]{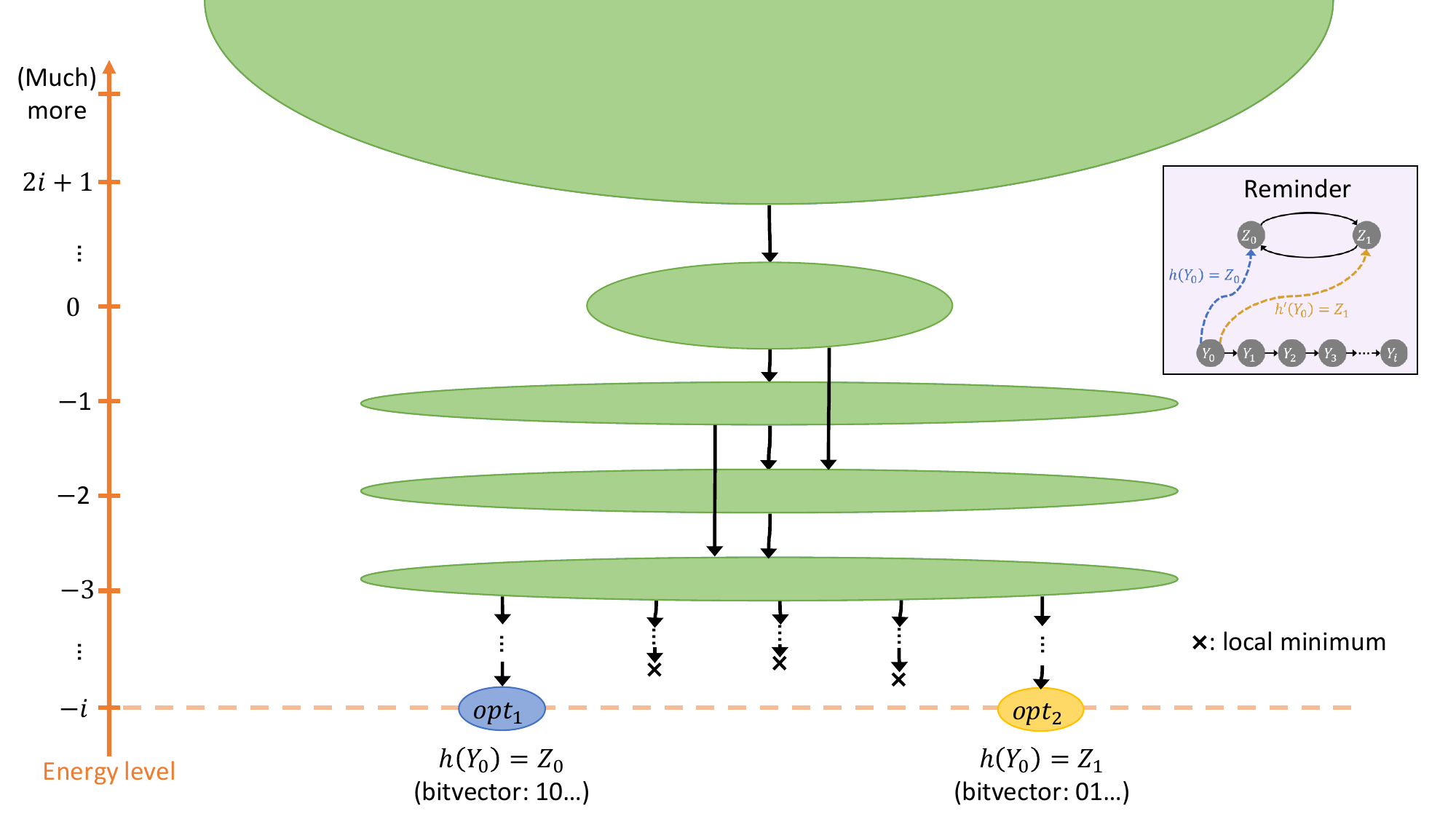}
		{\footnotesize 2-cycle-to-$i$-chain\par}
	\end{subfigure}%
	\hfill\vspace{6ex}
	\begin{subfigure}[t]{0.48\textwidth}
		\centering
		\includegraphics[width=\textwidth]{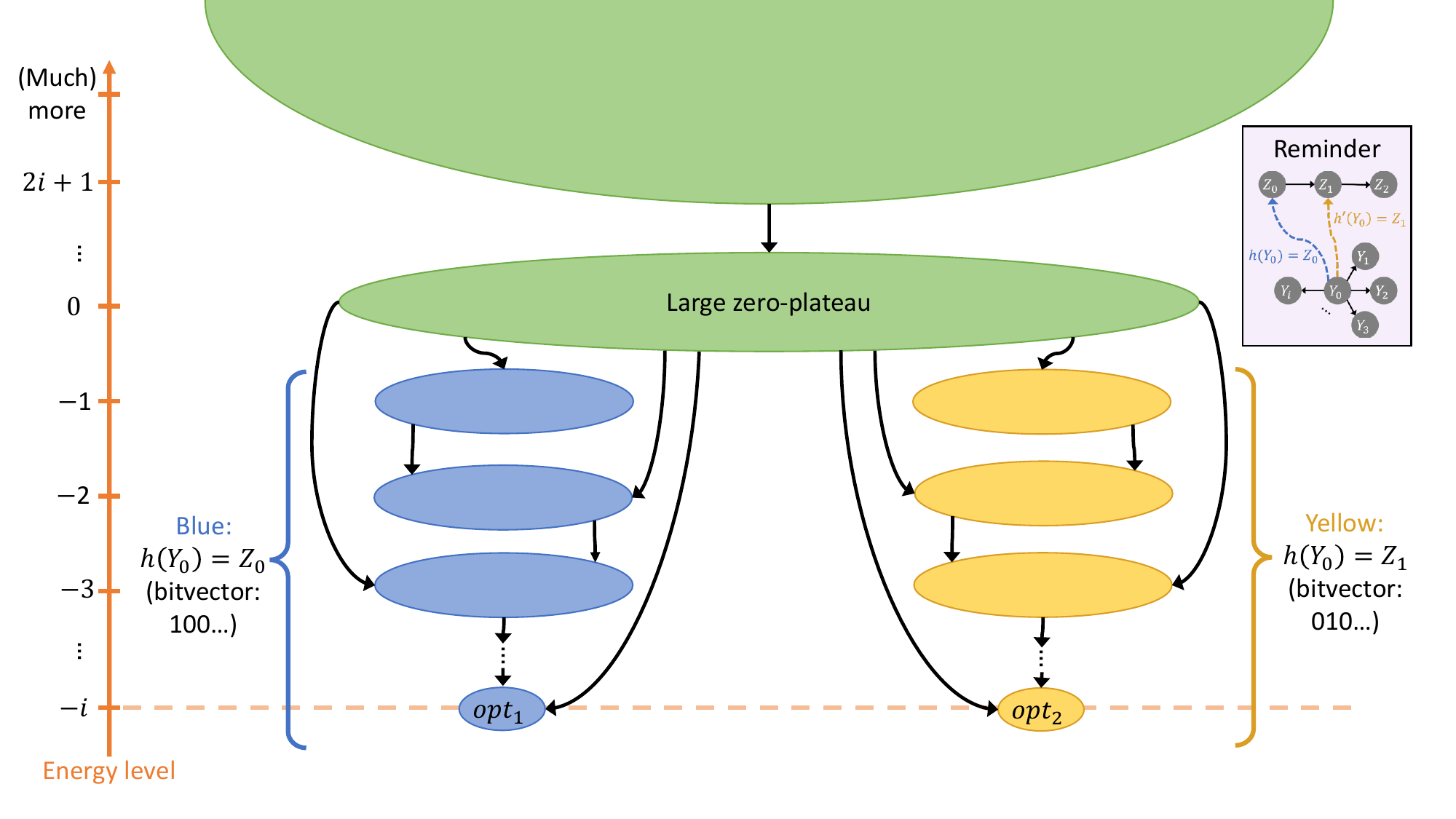}
		{\footnotesize 2-chain-to-$i$-star\par}
	\end{subfigure}
	\caption{Energy landscapes}

	\label{fig:energy_landscapes}
\end{figure}

For \emph{2-cycle to $i$-chain}, the negative energy levels are connected in a cascading/short-stepped way, such that connected neg states are at most two energy levels away from each other. This is due to the chain structure of the tuples in $q_2$. Every node of $q_2$ is part of at most two edges (and thus of at most two tuples), which means that changing one single bit can change at most two terms in $\pac(\ov{x})$. Therefore, reaching an optimal solution from zero level takes at least $\frac{i}{2}$ steps. Moreover, the energy landscape contains a number of local minima that are not optimal, which are indicated by crosses in the figure. They exist because independent parts of each of the two optimal solutions can contribute to a negative value within the same bitvector. Such a local but not global minimum could, for example, look like $LRLRLRRLRLLR$, where $L$ stands for $10$ (mapping a node from the chain to the left node $\vvarz_0$) and $R$ for $01$ (mapping it to the right node $\vvarz_1$). The example candidate consists of three parts ($LRLRLR;RLRL;LR$) that independently look like part of an optimal solution, and only the connections between the parts are off. In general,
those locally minimal (but not globally minimal) candidates look like piece-wise optimal candidates: sub-chains of size at least 2 are mapped correctly in itself (alternating between left and right) but alternate between the two optimal solutions. There is no way to decrease the value of that state without first increasing it by \enquote{unmapping} one of the nodes. These local minima occur below the energy levels of $-\frac{i}{2}$, and their number increases with growing~$i$, which plays a role in lowering the solution probabilities for that family.

For \emph{2-chain-to-$i$-star}, the neg states can be partitioned into two separate parts, half of which contain bitvectors starting with \enquote*{100} (central node $\vvary_0$ mapped to $\vvarz_0$, indicated in blue) and half of which start with \enquote*{010} (central node $\vvary_0$ mapped to $\vvarz_1$, indicated in yellow). This is because the central node of the star is part of every single tuple and thus of every negative term in the polynomial. Moreover, this means that there are no local minima except for the two global minima. Every neg state can reach the next lower energy level without increasing its energy by at most two steps. If one of the nodes in $\{\vvary_1,\ldots, \vvary_i\}$ has not been mapped correctly, it is either not mapped at all or incorrectly mapped. In the latter case, it can be unmapped without increasing the energy level and then correctly mapped in the next step. Regarding the first case, it can be directly mapped to reach the lower energy level. The absence of local-only minima is a key factor for the good performance observed in the experiments because the algorithm cannot get stuck.
\\
It remains to clarify why we observe lower solution probabilities for SA in the beginning,
i.e., for small values of $i$
(cf., Figure~\ref{fig:2-chain-to-i-star}). This can be explained as follows: The neg states are highly interconnected with the zero states, since every neg state has a direct zero neighbor where the first 1 is replaced by 0 (starting with \enquote*{000}, so the central node is unmapped). This holds specifically for the two optimal states, which also have a zero neighbor. As explained before, there is a low chance to jump out of the optimum into that higher-energy zero neighbor. For simulated annealing, we can directly compute that probability, which decreases exponentially as $i$ grows and annealing time passes by (see Equation~\eqref{eq:Boltzmannprob}). The counterprobability of staying in the optimum
is of the form $1-e^{-\gamma i}$ for a constant $\gamma>0$ (details on this are provided at the end of this section), which quickly approaches 1 for growing $i$.
Figure~\ref{fig:star_temperatures} illustrates the probability to stay in the optimum
for growing $i$ at $25\%$ of time and compares it to the measured solution probability of SA in the experiment that we showed in Figure~\ref{fig:2-chain-to-i-star}. Our hypothesis is that after 25\% of the time we are close to the critical point when the \enquote{escape probability} (i.e., the probability to escape from an optimal solution)
is still high enough to have an impact, but afterward quickly vanishes. This is backed up by the correlation between the 
two scatter plots shown in Figure~\ref{fig:star_temperatures}. In combination with the relatively large zero plateau for small $i$, jumping out of the optimum can cause the algorithm to get lost afterward, since the transition probability between states of the same energy level is always 1 (because then we have $c-c_\text{new}=0$ and $e^0=1$,
cf.\ Equation~\eqref{eq:Boltzmannprob})
and thus the algorithm can move to other zero states after escaping the optimum. Putting it all together, this explains the lower solution probabilities observed for small $i$ and their
growth as $i$ increases. For exact computations of the probabilities depicted in Figure~\ref{fig:star_temperatures}, see the end of the next paragraph.

\begin{figure}[htb]
	\centering

\includegraphics[width=0.6\linewidth]{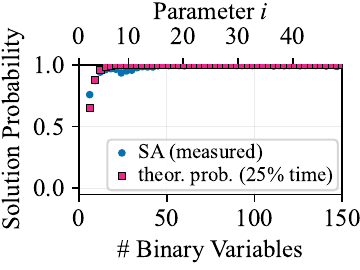}

  \caption{Measured solution probability (using SA simulator) vs.\ the theoretical probability for 2-chain-to-$i$-star to stay in optimal solution at 25\% of the annealing time}
    
 \label{fig:star_temperatures}
\end{figure}

\paragraph{Exact formulas for relative sizes and escape probabilities}

The relative sizes of the different state portions are computed
as follows.
A state $\ov{x}$ has a positive value if
$\punique(\ov{x})>0$, i.e., 
if the bit-matrix contains a row with at least two ones. For \emph{2-cycle to $i$-chain}, each row has length 2, so in three out of four possible 2-bit row-vectors, 
the row as at most one 1.
The portion of pos states is given by the counterprobability of all rows
having at most one 1:
\[
  \begin{array}{rcl}
    p_i^{(\text{chain})}(\text{pos})
 & = 
 & 1-p_i^{(\text{chain})}(\neg \text{pos})
\\
 & =
 & 1-p_i^{(\text{chain})}(\neg\text{pos}_\text{row})^{i+1}
    \ \ = \ \   1-\left(\frac{3}{4}\right)^{i+1}
\end{array}
\]
where $i{+}1$ comes from the number of rows in the bit-matrix.

Analogously, the probability for \emph{2-chain to $i$-star} can be computed with a row-length of 3, which leads to \[p_i^{(\text{star})}(\text{pos}) \ \ = \ \ \textstyle1-\left(\frac{1}{2}\right)^{i+1}\]
because in half of the eight possible 3-bit row-vectors,
the row has at most one 1.

Investigating only non-positive states (which is expressed in the following by using a conditional probability: $p(...|\neg\text{pos})$), we can further analyze their proportions. Since
each row contains at most one 1,
each row $\vvary_\ell$ can be interpreted as a mapping of the node $\vvary_\ell$, either to nothing (zero vector) or to one of the
variables
$\vvarz_k$
from \emph{2-cycle} or \emph{2-chain}, respectively.
\\
In the case of \emph{2-cycle to $i$-chain}, the state is negative if at least two consecutive nodes from the $i$-chain are mapped to opposite nodes from the 2-cycle. Abstractly, we can think of having an alphabet $\Sigma\deff\{N,L,R\}$
and the language $\mathbf{L}_{i}\deff \Sigma^{i+1}\cap(*(LR|RL)*)$ containing all words of length $i{+}1$ that contain the sub-word $LR$ or $RL$, which represents the portion of neg states. The size of that language can be computed (see \url{https://oeis.org/A193519}; accessed 2025/12/02) and is given by
\[|\mathbf{L}_{i}| \ \ = \ \ {\textstyle \frac{1}{2}}\cdot\Big(2 {\cdot} 3^{i+1}-(1{+}\sqrt{2})^{i+2}-(1{-}\sqrt{2})^{i+2}\Big).\]
The resulting fraction of all words of length $i{+}1$ is then given by
\[p_i^{(\text{chain})}(\text{neg} | \neg \text{pos}) \ = \ \frac{1}{2{\cdot} 3^{i+1}}\cdot \Big( 2 {\cdot} 3^{i+1}-(1{+}\sqrt{2})^{i+2}-(1{-}\sqrt{2})^{i+2}\Big)\]
which can also be written as
\ $p_i^{(\text{chain})}(\text{neg} | \neg \text{pos})  \ =$
\[
  1 \ - \
  \left(
    \frac{1{+}\sqrt{2}}{2} \cdot \left(\frac{1{+}\sqrt{2}}{3} \right)^{i+1} \ + \
    \frac{1{-}\sqrt{2}}{2} \cdot \left(\frac{1{-}\sqrt{2}}{3} \right)^{i+1} 
  \right)
\]  
and approaches 1 as $i$ increases.
Therefore, the opposite portion of zero states approaches 0.
\\
In the case of \emph{2-chain to $i$-star}, a state has a negative value if either the central node $\vvary_0$ of the $i$-star is mapped to $\vvarz_0$ and at least one of the other nodes is mapped to $\vvarz_1$, or if $\vvary_0$ is mapped to $\vvarz_1$ and at least one of the other nodes is mapped to $\vvarz_2$.
Given a correct mapping of the central node, each other node has a probability of~$\frac{3}{4}$ to not follow suit. So the fraction of zero states,
given that each row has at most one 1,
can be computed by
\[p_i^{(\text{star})}(0|\neg\text{pos})\ \ = \ \ \textstyle \frac{1}{2} \ + \ \frac{1}{2}\cdot\left(\frac{3}{4}\right)^{i}\ \ = \ \ \frac{1}{2}\left(1+\left(\frac{3}{4}\right)^{i}\right),\]
for which the probability of mapping the central node incorrectly is added to the probability of mapping the central node correctly, but all other nodes incorrectly.
The counterpart fraction of neg states is given by the counterprobability \[p_i^{(\text{star})}(\text{neg}|\neg\text{pos})\ \ = \ \ \textstyle\frac{1}{2}\left(1-\left(\frac{3}{4}\right)^{i}\right).\]
For growing $i$, their ratio approaches 50/50.

This explains how we computed the numbers in Table~\ref{tab:relative-sizes}.
Next, we explain the computations required to understand Figure~\ref{fig:star_temperatures}.

We assume that an optimal state has been reached and want to compute the probability to escape into the zero plateau when the algorithm encounters a zero neighbor. Following Algorithm~\ref{alg:simann}, this means that $\ov{x}$ is an
optimal state and $\ov{x}_\text{new}$ is
a zero neighbor. For their energy levels, we have $c=-i$ and $c_\text{new}=0$, so the probability of transition can be computed by
$\textit{AcceptProb}(T,i)$ with the given cooling schedule $T$. In the code implementation, the schedule's min and max are given by the values $\beta_\text{start}=\frac{1}{T_\text{high}k_b}\deff 0.5$ and $\beta_\text{end}=\frac{1}{T_\text{low}k_b}\deff 10$ and the algorithm follows a geometric schedule divided into 1000 steps, which means that for $\ell\in\NN$ with $1\leq \ell\leq 1000$,
we can compute
\[
\alpha=\sqrt[1000]{\frac{\beta_\text{end}}{\beta_\text{start}}}\quad\text{and}\quad
\beta_\ell=\beta_\text{start}\cdot\alpha^\ell.
\]
The accept probability for the energy difference of $i$ at $25\%$ of the annealing time is then given by
\[\textit{AcceptProb}(T_{25\%},i)\ \ = \ \ e^{-\beta_{250}\cdot i} \ \ = \ \ e^{-0.5\cdot\sqrt[4]{20}\cdot i},\]
where $T_{25\%}$ indicates using $\beta$ at 25\% of the time, which is (due to using 1000 steps) given by $\beta_{250}$.
Taking the counterprobability $1-e^{-\beta_{250}\cdot i}$ we have finally arrived at the probability of staying in the optimal solution after 25\% of the annealing time depending on~$i$, which is shown in Figure~\ref{fig:star_temperatures}.